\documentclass[letterpaper]{article} 
\usepackage[preprint]{aaai2027}  

\usepackage[hyphens]{url}  
\usepackage{graphicx} 
\usepackage{natbib}  
\usepackage{caption} 
\usepackage{algorithm}
\usepackage{algorithmic}

\usepackage{newfloat}
\usepackage{listings}
\DeclareCaptionStyle{ruled}{labelfont=normalfont,labelsep=colon,strut=off} 
\floatstyle{ruled}
\newfloat{listing}{tb}{lst}{}
\floatname{listing}{Listing}

\usepackage{booktabs}

\usepackage{microtype}
\usepackage{hyperref}
\usepackage{subcaption}
\usepackage{amsmath}
\usepackage{mathtools}
\usepackage{amsthm}
\usepackage[capitalize,noabbrev]{cleveref}
\usepackage{comment}
\usepackage{booktabs}
\usepackage{amsfonts}
\usepackage{nicefrac}
\usepackage{xcolor}
\usepackage{dsfont}
\usepackage{tikz}
\usepackage{makecell}

\DeclareMathOperator*{\R}{\mathbb{R}}

\DeclareMathOperator*{\argmax}{argmax}
\DeclareMathOperator*{\argmin}{argmin}

\newcommand{\sumi}{\sum_{i=1}^{n}}

\newcommand{\sumj}{\sum_{j=1}^{n}}
\newcommand{\K}{$k$}

\newcommand{\PRP}{argmax}
\newcommand{\di}{$p$}
\newtheorem{theorem}{Theorem}
\newtheorem{definition}{Definition}

\newtheorem{observation}{Observation}
\newtheorem{corollary}{Corollary}

\title{Learning Dynamics of Strategic Publishers in Generative AI Ecosystems}
\author{
    Sagie Dekel\textsuperscript{\rm 1}\corresponding,
    Omer Madmon\textsuperscript{\rm 1},
    Moshe Tennenholtz\textsuperscript{\rm 1},
    Oren Kurland\textsuperscript{\rm 1}
}

\affiliations {
    \textsuperscript{\rm 1}Technion - Israel Institute of Technology\\
    \{sagie.dekel, omermadmon\}@campus.technion.ac.il, 
    \{moshet, kurland\}@technion.ac.il
}

\begin{document}

\maketitle

\begin{abstract}
Generative AI (GenAI) search systems are transforming how users access information.
Unlike ranking-based search systems, where users observe a ranked list of documents, GenAI search systems, given a user's question, generate an answer, often accompanied by external sources (e.g., in the form of citations).
Content creators (publishers) seeking to increase exposure might behave strategically and compete with other creators for users' attention.
While publishers in ranking-based systems might strategically modify their content to improve its ranking, the incentives in generative systems take on a new form. Publishers may now gain exposure through generated responses and attributions to those responses. 
We introduce a novel game-theoretic model of the emerging GenAI ecosystem in which publishers compete for attribution-based exposure. 
We study the learning dynamics of strategic content creators under better-response dynamics.
We associate the convergence of learning dynamics to equilibrium with ecosystem stability. Employing the notion of potential games, we study the stability of GenAI ecosystems under several known content selection mechanisms.
We demonstrate the instability of mechanisms representing real-world modern systems and characterize a mechanism that induces a stable ecosystem.
We conduct extensive simulations to analyze the stability and welfare of GenAI ecosystems under various mechanisms.
The simulations support our theoretical findings and reveal an interplay among stability, publisher welfare, and user welfare.
In particular, stable mechanisms do not necessarily maximize welfare, demonstrating an important trade-off for platform designers.
We then introduce a study illustrating that the proper selection of the GenAI mechanism enables the manifestation of desired trade-offs between publisher welfare and the different sources of user welfare.
\end{abstract}


\section{Introduction}

Generative AI (GenAI) is rapidly reshaping how users search for information \cite{NBERw34255}.
Traditional search systems --- e.g., Google search --- present a ranked list of documents (e.g., the standard “10-blue-links”). In contrast, GenAI search systems --- e.g., ChatGPT or Perplexity\footnote{perplexity.com} --- use documents to synthesize a response; often leveraging retrieval-augmented generation (RAG) frameworks \cite{3495724.3496517}, and accompany it with attributed sources (e.g., in the form of citations embedded in the text \cite{DBLP:journals/corr/abs-2510-11560}).
Consequently, in GenAI systems, content creators (publishers) may influence user experience through responses \cite{dekel2026addressingcorpusknowledgepoisoning} or through the attributions associated with those responses. Importantly, there is an interplay between the two.
Regarding the latter influence type, publishers may be incentivized to modify their content so that it is more likely to be cited or referenced in responses.

We propose a novel game-theoretic model for GenAI search --- specifically, for generative question-answering systems --- as an ecosystem in which publishers create documents and users submit questions. GenAI platforms generate responses to those questions using selected documents and generative models (e.g., LLMs), while assigning attribution to publishers.
Following the literature on traditional search and recommendation systems, we model documents and questions as vectors in an embedding space; e.g., \citet{pmlr-v202-yao23b, 10.5555/3666122.3669383}.
Since publishers are interested in maximizing exposure, they may be incentivized to behave strategically.
We adopt \emph{better-response dynamics} \cite{repec:eee:gamebe:v:13:y:1996:i:1:p:111-124} to describe their learning dynamics. That is, we only assume that if there is a situation where a publisher may gain by deviating to any utility-improving action, then the publisher will indeed deviate. This is a minimal assumption, as it imposes neither optimality of deviations nor any order on deviators.

Although content modification through better-response learning dynamics was advocated in prior work \cite{10.1609/aaai.v39i13.33534, Omer_Madmon_2025}, we are the first to model, using game theory, the GenAI search ecosystem and its effect on the learning dynamics of publishers.
Specifically, publisher utility depends on both being attributed and conforming with her original content, and user utility is determined by both the relevance of generated responses to the question and the relevance of attributions to the response. 

To establish tight theoretical results, we appeal to the \emph{Nash equilibrium}\footnote{We restrict to pure Nash equilibrium (PNE), as it corresponds to stationary content and therefore provides a natural notion of stability in a search ecosystem.} (NE, \citealp{nash1950equilibrium}) solution concept, and convergence to equilibrium; a state which implies stability of the ecosystem. 
As in prior studies of learning dynamics --- e.g., \citet{NEURIPS2018_a9a1d531} --- a key question arises: "Can convergence to equilibrium be guaranteed in GenAI ecosystems with strategic publishers?"
Using the concept of \emph{potential games} \cite{monderer1996potential}, we revisit content selection mechanisms offered in the literature; interestingly, we show that in a relatively narrow set of scenarios, such convergence is guaranteed, although for some other mechanisms, convergence is empirically obtained.
Importantly, we show that better-response dynamics may diverge in retrieve-then-generate systems (e.g., RAG \cite{3495724.3496517}) and systems that are based on the winner-takes-all concept \cite{https://doi.org/10.1111/j.1468-0343.1989.tb00003.x}.

While stability is interesting in its own right, we also empirically study the welfare properties of overall dynamics regardless of whether convergence occurs.
Through experimental analysis, we show several intriguing results. Specifically, we introduce a novel mechanism design setting in which we show different mechanisms to trade off user welfare and publisher welfare.
In addition, we reveal the unexpected merits of unstable mechanisms in terms of social welfare. Specifically, we show that stability does not entail higher social welfare and, interestingly, in some cases, significantly lower. 
Finally, we empirically identify and analyze the mechanisms that maximize social welfare for various objectives. Specifically, we show that under different objectives, these mechanisms differ both in their family and in the choice of hyperparameters.

Our contributions can be summarized as follows:

\begin{itemize}
\item We introduce a novel game-theoretic model of the emerging GenAI search ecosystem with attribution-incentivized content creators. 
Our model generalizes prior models of search and recommendation ecosystems with ranking-incentivized content creators.

\item We provide a theoretical analysis of learning dynamics in terms of stability under mechanisms incorporating both attribution allocation functions and novel factors of generative platforms.
We prove that, among these mechanisms, only a relatively narrow family of mechanisms employing linear functions guarantees convergence.

\item We present an empirical analysis of better-response dynamics in GenAI ecosystems under the mechanisms stated above.
We consider both stability and welfare in the ecosystem, identifying novel phenomena and trade-offs. Additionally, we define social welfare measures and demonstrate an optimal mechanism design for generative platforms.


\end{itemize}

\section{Related Work}

\paragraph{Game Theory in Search and Recommendation Ecosystems.}
Game theory has been extensively employed to model strategic content creation in 
search \cite{10.1145/3477495.3532771} and recommendation \cite{NEURIPS2018_a9a1d531, Boutilier_Mladenov_Tennenholtz_2024} ecosystems.
With the growing adoption of dense retrieval methods \cite{xu2026a}, recent work has increasingly focused on models with continuous strategy spaces \cite{pmlr-v202-yao23b, 10.5555/3666122.3669383, madmon2025on, Omer_Madmon_2025, yu2025beyond}.
These game-theoretic formulations are commonly used to analyze stability \cite{reinman2026stabilitycompetitivesearchresults} and welfare \cite{10.1145/3637528.3672021} properties.
At the same time, the increasing use of GenAI for information-seeking \cite{NBERw34255} gives rise to models that capture strategic behavior in this setting \cite{taitler2025datasharinggenerativeai, 10.1145/3774904.3792475, ohayon2026contestsspilloversincentivizingcontent}.
For instance, \citet{wu2026aioverviewsbenefitsearch} introduce a model for AI-overview systems, and \citet{10.1145/3774904.3792308} propose a model in which creators strategically decide both their content and whether to share it with GenAI platforms.
To the best of our knowledge, we are the first to formally model the learning dynamics of strategic publishers in GenAI systems, specifically question-answering systems.

\paragraph{Strategic Content Modification in Ad Hoc Search and Recommendation.}
Strategic behavior of content creators has been widely addressed in prior work \cite{10.5555/3207692.3207717, hron2023modeling, prasad2023contentpromptingmodelingcontent, Boutilier_Mladenov_Tennenholtz_2024, 10.1609/aaai.v39i13.33534}.
When users observe a ranked list of content, publishers are ranking-incentivized; that is, they might modify their content to increase their rankings \cite{10.1145/3477495.3532771}. A substantial body of work has studied content modification of such publishers in search and recommendation systems \cite{10.1145/3664190.3672516, 10.1145/3726302.3730168,mordo2025rlrfcompetitivesearchagent,10.1145/3726302.3730026}.

The emergence of GenAI systems --- which often accompany their responses with citations \cite{DBLP:journals/corr/abs-2510-11560} --- has drastically changed the way users search for information \cite{NBERw34255}. Publishers may be attribution-incentivized; they might modify their content to increase attribution-based exposure in GenAI responses \cite{unknown}.
For instance, \citet{10.1145/3637528.3671900} propose a metric that jointly considers citations' position and the cited text to measure content exposure.
It has been broadly adopted and has become a common practice for estimating exposure of publishers in GenAI systems \cite{chen2025generativeengineoptimizationdominate, chen2026roleaugmentedintentdrivengenerativesearch, wu2026what, zhou2026ifgeoconflictawareinstructionfusion}.
Our work introduces, to the best of our knowledge, the first analysis of ecosystems with attribution-incentivized content creators.

\section{Preliminaries}\label{sec: preliminaries}

We begin by introducing the game-theoretic definitions and solution concepts used throughout the paper.

\paragraph{Game Theory Notation.}
An $n$-player \textbf{game} is a tuple
$G = (N, (X_i)_{i \in N}, (u_i)_{i \in N})$, where
$N := \{1,\ldots,n\}$ is the set of players, $X_i$ is the action set of
player $i$, and $u_i : X \to \mathbb{R}$ is player $i$'s utility function.
Here, $X := \times_{j \in N} X_j$ denotes the set of pure strategy profiles. For a profile $x \in X$, we write $x_i$ for the action of player $i$ and $x_{-i} \in X_{-i} := \times_{j \in N \backslash \{i\}}X_j$ for the actions of all players except $i$. 

\paragraph{Better-response, Equilibria, and Dynamics.}
Given a profile of opponents' actions $x_{-i} \in X_{-i}$ and two actions
$x_i,x'_i \in X_i$ of player $i$, we say that $x'_i$ is an
\textbf{$\epsilon$-better response} than $x_i$ with respect to $x_{-i}$ if
  $  u_i(x'_i,x_{-i}) > u_i(x_i,x_{-i}) + \epsilon .$
This notion naturally induces the equilibrium concept used in our analysis.
A pure strategy profile $x \in X$ is an \textbf{$\epsilon$-pure Nash equilibrium}
($\epsilon$-PNE) if, for every player $i \in N$, there is no action $x'_i$ that is an $\varepsilon$-better response than $x_i$, with respect to $x_{-i}$. 
When $\epsilon = 0$, we refer to a $0$-better response simply as a better
response, and to a $0$-PNE as a Pure Nash Equilibrium (PNE). A PNE represents
a stable outcome of a game \cite{NEURIPS2018_a9a1d531}, in the sense that no publisher has an incentive
to unilaterally change her strategy.

The notion of $\epsilon$-better response induces
\textbf{$\epsilon$-better-response dynamics}, defined as a sequence in which, at every timestep $t \geq 1$, either $x^{(t-1)}$ is an $\epsilon$-PNE and $x^{(t)} = x^{(t-1)}$, or there exists a player $i \in N$ such that $x_i^{(t)}$ is an $\epsilon$-better response than $x_i^{(t-1)}$ with respect to $x_{-i}^{(t-1)}$, and $x_j^{(t)} = x_j^{(t-1)}$ for every $j \neq i$.
We say that the dynamics \textit{converges} if there exists a timestep $T$ such that
$x^{(T)}$ is an $\epsilon$-PNE.

\paragraph{Potential games.}
A central class of games used in this paper is the class of potential
games. A game $G$ is a \textbf{potential game} if there exists a potential function
$\phi : X \to \mathbb{R}$ such that
for every player $i \in N$, every profile of opponents’ actions $x_{-i} \in X_{-i}$,
and every two actions $x'_i,x''_i \in X_i$, it holds that:
\[
    \phi(x'_i,x_{-i}) - \phi(x''_i,x_{-i})
    =
    u_i(x'_i,x_{-i}) - u_i(x''_i,x_{-i}) .
\]
Thus, in a potential game, the incentives of all players are represented by a single global function. We note a celebrated result for potential games with better-response dynamics:

\begin{theorem}[\citet{monderer1996potential}]
Let $G$ be a potential game and let $\epsilon > 0$. Then $G$ admits at least one $\epsilon$-PNE, and any $\epsilon$-better-response dynamics in $G$ converges.
\end{theorem}

\section{The Model}\label{sec: the model}
Our model is a formalism of a setting in which publishers compete for exposure by selecting documents to publish. In response to a user's question, a mediator chooses which content to present. We focus on a novel and relatively understudied type of mediator: GenAI platforms.

A \textbf{generative publishers' game} consists of a set of publishers $N := \{1,2,\ldots,n\}$, with $n \geq 2$, that strategically choose documents represented as \di-dimensional vectors in the bounded domain $\mathcal{V} := {[0,1]}^{\text{\di}}$.
Specifically, each publisher $i \in N$ selects an action $x_i \in X_i := \mathcal{V}$, interpreted as the embedding of the content provided by that publisher. Each publisher has an initial document $x^i_0 \in X_i$ and a cost factor $\lambda_i > 0$, which captures publisher $i$'s cost of providing content that differs from her initial document, i.e., content drift entails a cost in terms of conveying her original intent. We use $\lambda > 0$ to denote a uniform cost factor: $\lambda_i = \lambda$ for all $i \in N$.
Distances are measured by a continuous semi-metric $d : \mathcal{V} \times \mathcal{V} \to \mathbb{R}_+$\footnote{We use $\R_{+}$ to denote the set of non-negative real numbers and $\R_{++}$ to denote the set of strictly positive real numbers.}, such that $d(a,b) \leq 1$ for all $a,b \in \mathcal{V}$.
Given a strategy profile $x \in X$ and a user's question\footnote{The question serves as a proxy for the user’s information need.} (representation) $x^* \in \mathcal{V}$, the platform generates a response according to a generation function $y: X \times \mathcal{V} \to \mathcal{V}$.

In our model, aligned with modern real-world systems \cite{DBLP:journals/corr/abs-2510-11560}, users observe a generated response and attributed documents (e.g., in the form of citations). 
The platform employs an attribution mechanism that maps content and a response to (cited) content. Formally, an attribution allocation function maps a strategy profile and a response to a distribution over publishers $r: X \times \mathcal{V} \to \Delta_n$.
Let $r_i$, the i-th component of the distribution, be the attribution score of publisher $i$, which can be interpreted as her probability of being ranked highest in a citation list.

We abuse the notation and denote $y(x) \coloneqq y(x, x^*) $, $ r(x) \coloneqq r(x, y(x))$, $d^*(x_i) \coloneqq d(x_i, y(x))$ and $d^0_i(x_i) \coloneqq d(x_i, x_0^i)$.
The \textbf{utility }of publisher $i$ under the profile $x$ is:

\begin{equation}
    u_i(x) \coloneqq r_i(x) - \lambda_i \cdot d^0_i(x_i).
\end{equation}
That is, the utility of publisher $i$ is her attribution score minus the cost of providing content that differs from her initial document. 
We define the \textbf{publishers' welfare} $\mathcal{H}(x)$ and the \textbf{users' welfare} $\mathcal{U}(x)$.
The publishers' welfare is the sum of their utilities. The users' welfare is a linear interpolation of the relevance of the response to the question and the expected relevance of attributions to the response, where we estimate the relevance of $a \in \mathcal{V}$ to $b \in \mathcal{V}$ by $1-d(a, b)$\footnote{Below, we model the response as a centroid of estimated-relevant content. Prior literature on search has demonstrated that such centroids can serve as a proxy for relevance \cite{rocchio71relevance, lavrenko-2001-relevancebased, 10.1145/502585.502654}.}.  
\begin{equation}\label{eq:publishers' welfare}
    \mathcal{H}(x) \coloneqq \sumi u_i(x) = 1 -  \sumi \lambda_i  \cdot  d^0_i (x_i). 
    \end{equation}
    
    \begin{equation}\label{eq:users' welfare}
    \begin{aligned}
    \mathcal{U} (x) &\coloneqq \gamma \cdot (1 - d(x^*, y(x))) \;  \\ & + (1-\gamma) \cdot \sumi \bigl (1 - d^*(x_i) \bigr ) \cdot r_i(x) ,
    \end{aligned}
    \end{equation}
where $\gamma \in [0,1]$ is the interpolation parameter.
The users' welfare is a distinctive aspect of GenAI ecosystems, incorporating both \emph{responses' relevance} to the question and \emph{attributions' relevance}  to the response, which are the first and second terms in Equation~\ref{eq:users' welfare}, respectively.
Their weighting may differ between settings, domains, and users' underlying objectives. For example, in the medical domain, platforms may assign greater weight to attributions' relevance to ensure that responses are grounded with presumed-relevant external information.

We note that modern generation functions (e.g., LLMs) are complex, multi-layered, and generally lack a tractable analytical structure. To enable theoretical analysis, we therefore consider a family of centroid-based generation functions. Our approach is motivated by Rocchio's model \shortcite{rocchio71relevance}, a classical centroid-based method in information retrieval for relevance feedback in vector spaces\footnote{Centroids of document sets have also been shown to be effective for approximating desired text embeddings \cite{puduppully-etal-2023-multi}.}. Formally: 
\begin{equation}
y_{\alpha,z}(x)
:=
\alpha \cdot \frac{1}{\text{\K}}\sum_{j\in S_k} x_j
+
(1-\alpha)\cdot z ,
\label{eq:generation-function}
\end{equation}
where \(S_k\subseteq N\) denotes the set of \K{} publishers
whose documents are closest to \(x^*\), \(\alpha\in[0,1]\) is an interpolation parameter that controls the relative impact of documents on generation (henceforth \emph{retrieval weight}), \K{} ($\le n$) is a constant that accounts for the number of publishers' documents used to compute the centroid (henceforth \emph{context size}), and \(z\in\mathcal{V}\) represents the platform’s posterior belief, which, given a question, could be an LLM's parametric knowledge (representation) or a user's information need approximation; e.g., \cite{https://doi.org/10.1002/pra2.1245}. 
In addition, we use the normalized squared Euclidean distance
$
d(a,b) = \frac{1}{\text{\di}}\|a-b\|_2^2
$ as our distance semi-metric.

We discuss two fundamentally different settings: \K{} $< n$ and \K{} $= n$.
Although both settings can model GenAI systems, they represent different scenarios.
The former setting represents ecosystems with a large or diverse set of publishers (e.g., the Web), in which systems mainly employ RAG frameworks to retrieve top-\K{} documents presumed to be relevant. The latter setting represents niche domains in which a small number of publishers compete for user engagement, or systems that rely on a fixed set of relevant content creators. Each, as we demonstrate later, leads to substantially different outcomes.

Our model generalizes prior models of traditional systems \cite{10.5555/3666122.3669383, madmon2025on, Omer_Madmon_2025}, where publishers compete for ranking-based exposure.  
For instance, if we set $\alpha=0$ and $z=x^*$ in Equation~\ref{eq:generation-function}, then we arrive to Madmon et al.'s \shortcite{Omer_Madmon_2025} model.
We emphasize, however, a key distinction from these prior models. In ranking-based systems, publishers compete relative to an exogenous reference point (e.g., a query) and a ranking. In our model, publishers compete relative to a dynamic response (and a ranking), making the reference point endogenous (see Equation~\ref{eq:generation-function}). 

\section{Learning Dynamics of Strategic Publishers}
\label{sec: ranking functions}

In this section, we present a theoretical analysis of better-response dynamics under various content attribution functions using the generation function from Equation~\ref{eq:generation-function}.

\subsection{The Winner Takes All}\label{sec:PRP}

The first attribution function we consider is based on the widely used winner-takes-all concept\footnote{See, for example, the probability ranking principle (PRP) \cite{robertson1977probability}, all-pay auctions \cite{af85224f-dd40-3eea-93c3-38192284c231}, and TOP mechanisms \cite{NEURIPS2018_a9a1d531}.}.
Naturally, the \PRP{} function induces a deterministic attribution distribution that assigns all the mass to the presumably most relevant publisher. Recall, we estimate content's relevance using its distance from the response. Formally:

\begin{definition}
The \textbf{\PRP{} }attribution function $r^*$ of publisher $i$ is:
    \begin{equation}
        r_i^*(x) \coloneqq
        \frac{\mathds{I} \{ i\in\mu(x) \}}{|\mu(x)|} 
    \end{equation}
    where $\mu(x) \coloneqq \argmin_{j \in N} d^*(x_j)$ is the set of publishers providing the closest content (most relevant) to the response $y(x)$ among all publishers and $\mathds{I}$ denotes the indicator function.

\end{definition}

Informally, the \PRP{} attribution function assigns an attribution score of $1$ to the most relevant publisher (in the case of ties, the distribution is uniform over the relevant publishers), and a score of $0$ to others.
Note that the \PRP{} function yields optimal short-term attributions' relevance in the sense that it minimizes the attribution distance from the response for a fixed strategy profile $x$. However, it induces unstable ecosystems, in which dynamics may not converge, as we show in Appendix~\ref{apn:Additional_Learning_Dynamics_Results} (Observation~\ref{obs:2}).
In fact, as we show in the following observation (whose proof is given in Appendix~\ref{apn:omitted_proofs}), games may not even admit a PNE.

\begin{observation}\label{obs:1}
Consider generative publishers' games induced by the \PRP{} attribution function $r^*$. For both \K{} $< n$ and \K{} $=n$ cases, there exists a set of games in which any instance possesses no PNE.
\end{observation}

The proof is essentially based on the discontinuity of the \PRP{} function. 
This observation serves as a case in point for the shortcomings of the \PRP{} attribution function in terms of the ecosystem's stability. 
Specifically, the \PRP{} induces a game in which $\varepsilon$-better-response dynamics might not converge and a PNE might not exist.
We next study the dynamics induced by a common continuous function.

\subsection{The Softmax Function}\label{sec:Softmax}

The softmax function is widely used in machine learning to form a distribution. It therefore provides a natural way to smooth the discrete choice performed above. 

\begin{definition}
The \textbf{softmax} attribution function\footnote{The generative publishers’ game induced by the softmax attribution function extends Hirshleifer’s well-known model of rent-seeking competition \cite{6ea1d0d3-2a3d-3097-baf8-831d79887b68}.} $\tilde{r}$ of publisher $i$ is:
    \begin{equation}
            \tilde{r}_i(x) \coloneqq \frac{e^{-\beta \cdot d^*(x_i)}}{\sumj e^{-\beta \cdot d^*(x_j)}},
    \end{equation}
where $\beta > 0$ is an inverse temperature constant.
\end{definition}

We start by showing in the following theorem, the proof of which can be found in Appendix~\ref{apn:omitted_proofs}, that the generative game under the softmax function is not a potential game for the \K{} $=n$ case.
In Appendix~\ref{apn:Additional_Learning_Dynamics_Results} (Theorem~\ref{theorem: prop not potential}) we prove a stronger result for the \K{} $<n$ case. In particular, under any proportional function (as defined in \citealp{madmon2025on}) --- of which softmax is a special case --- the game is not a potential game.

\begin{theorem}\label{the:softmax-no-potential}
Consider generative publishers' games with context size \K{} \(=n\).
For every $\beta>0$, the softmax attribution function does not induce a potential game. Consequently, if the softmax induces a potential game, then
$\beta=0$. Equivalently:
\[
\tilde r_i^{0}(x)=\frac1n
\quad \text{for all } i\in[n],\ x\in X.
\]
\end{theorem}

Although Theorem~\ref{the:softmax-no-potential} and Theorem~\ref{theorem: prop not potential} show that the softmax attribution function does not induce a potential game, this does not imply the divergence of better-response dynamics.
In fact, empirical analysis in Section~\ref{sec:empirical_results} suggests that with high likelihood, better-response dynamics do converge under the softmax function.
However, the following observation, the proof of which is given in Appendix~\ref{apn:omitted_proofs}, demonstrates dynamics that do not converge. 

\begin{observation}\label{obs:3}
Consider generative publishers’ games induced by the softmax attribution function $\tilde r$. For both \K{} $< n$ and \K{} $=n$ cases, there exists a game set for which there exists $\beta_0>0$ such that for every $\beta\ge \beta_0$, there exists an infinite sequence of profitable deviations. In particular, for sufficiently small $\varepsilon$, $\varepsilon$-better-response dynamics may not converge.

\end{observation}

The result just stated indicates that under the softmax attribution function, stability is not guaranteed. Next, we discuss the linear attribution function, which leads to both stable and unstable ecosystems.

\subsection{The Linear Relative Relevance Function}
\label{linear_theory}

We adapt the notion of linear relative relevance from \citet{Omer_Madmon_2025}, which is the difference between the average distance of all other publishers from the generated response $y(x)$ and the distance of publisher $i$ from $y(x)$.

\begin{definition}
    Let $x \in X$ be a strategy profile in a generative publishers' game. The \textbf{linear relative relevance} of publisher $i \in N$ with respect to $x$ is:
    \begin{equation}
            \nu_i(x) \coloneqq \frac{1}{n-1} \Bigl( \sum_{j \in N \setminus \{i\}} d^*(x_j) \Bigr) - d^*(x_i).
    \end{equation}
\end{definition}

Notably, the relative relevance of publisher $i$ increases as other publishers become less relevant to $y(x)$ or she becomes more relevant to $y(x)$.
To induce a valid probability distribution, a linear function of the form \({r}_i(x) \coloneqq m \cdot \nu_i(x) + b\) entails a slope \(m \in (0, \frac{1}{n}]\) and a bias $b=\frac{1}{n}$.
Consequently, we can naturally define the linear attribution function:

\begin{definition}
The \textbf{linear }relative relevance attribution function $\hat{r}$ of publisher $i$ with slope \(m \in (0, \frac{1}{n}]\) is:
    \begin{equation}
        \hat{r}_i (x) \coloneqq m \cdot \nu_i(x) + \frac{1}{n} .
    \end{equation}
\end{definition}
Note that the slope $m$ controls the extent to which publishers are motivated to modify their content by varying the extent to which the function is uniform over publishers. 

\citet{Omer_Madmon_2025} showed that linear ranking entails stability in ecosystems with ranking-based exposure.
As discussed, this result does not naturally transfer to GenAI ecosystems. In fact, as we later show, this stability result does not hold for mechanisms with partial selection criteria (i.e., a context size of \K{} $<n$). 
However, in the following theorem, the proof of which is given in Appendix~\ref{apn:omitted_proofs}, we show that for mechanisms with context size \K{} $=n$ and the linear attribution function, the induced game is a potential game.

\begin{theorem}
\label{thm:linear-kn-potential}
Consider a generative publishers' game with the linear attribution function $\hat r$ and a generation using context size of \K{} $=n$.
Then, the induced game is an exact potential game with the potential function
\begin{align*}
\Psi_{}(x)
&=
-(1-\alpha)m\sum_{i=1}^n d(x_i,z) \\
&-
\frac{\alpha m(n-2)}{n(n-1)}
\sum_{1\le p<q\le n} d(x_p,x_q)
-
\sum_{i=1}^n \lambda_i \cdot d_i^0(x_i).
\end{align*}
Consequently, for every $\varepsilon>0$, the game admits an $\varepsilon$-PNE,
and every $\varepsilon$-better-response dynamics converges.
\end{theorem}

Since $X_i$ is compact for every $i$, and since both $y(\cdot)$ and $d(\cdot,\cdot)$ are continuous, each utility function $u_i$ is continuous in the strategy profile $x$. By Theorem~\ref{thm:linear-kn-potential}, the game is an exact potential game, hence it has acyclic dynamics. Therefore, the conditions of Kukushkin's \shortcite{kukushkin_nash_2011} theorem hold, thus:

\begin{corollary}
Under the conditions of Theorem~\ref{thm:linear-kn-potential}, the induced publishers' game admits a pure Nash equilibrium.
\end{corollary}

In contrast, as we show in the following theorem (whose proof is presented in Appendix~\ref{apn:omitted_proofs}), mechanisms with context size \K{} $<n$ do not induce potential games. 
Notably, this case represents, for example, retrieval-based systems in which only a subset of content creators' documents are used in generation \cite{3495724.3496517}.
Additionally, in Appendix~\ref{apn:Additional_Learning_Dynamics_Results} (Observation~\ref{obs:example_for_non_conv_with_linear_k<n}) we provide an example of a game in which better-response dynamics diverge.

\begin{theorem}\label{teo:linear_not_converge_for_K-smaller_than_n}
Consider the class of generative publishers’ games induced by the linear attribution function \(\hat r\).
For every \(\text{\K{}} < n\), this class contains games that are not potential games.
\end{theorem}

All in all, whether the linear attribution function induces a stable ecosystem is a question of mechanism design. Specifically, it depends on the context size \K{}.
In other words, the ability of platforms to ensure stability under better-response dynamics is subject to their ability to use the set of all publishers for generation and attribution.

\section{Empirical Results}\label{sec:empirical_results}
In this section, we present an empirical analysis of learning dynamics using simulated environments. The goal is to model realistic better-response dynamics.
We adopt a mechanism design perspective \cite{10.1093/acprof:oso/9780199734023.001.0001} and study the impact of our attribution mechanisms (\PRP{}, softmax, and linear) and novel design choices that emerged in the GenAI era (retrieval weight $\alpha$ and context size \K{}) on both the stability and welfare in the ecosystem. 

As noted by \citet{10.1145/3477495.3532771}, empirically evaluating competitive search settings poses substantial challenges.
For instance, the impact of textual modifications on a document’s (dense) representation is an active field of research \cite{tennenholtz2024embeddingaligned}. 
Therefore, following prior work on learning dynamics \cite{madmon2025on, Omer_Madmon_2025}, we use code-based simulations in the embedding space.
Specifically, we employ the \textit{discrete better-response dynamics} algorithm, intended to model realistic scenarios in which publishers iteratively modify their content to improve utility.

\subsection{Simulation Details}
In discrete better-response dynamics, publishers start with their initial documents, and at each timestep, one publisher modifies her document to improve her utility; a detailed algorithm for a single simulation can be found in Appendix~\ref{apn:sum_algorithm}. 
A simulation is said to \textit{converge} if it has reached an $\varepsilon$-PNE in no more than $T$ iterations.
In accordance with prior work \cite{Omer_Madmon_2025}, we set $T=1000$, $\varepsilon=10^{-6}$, and perform 500 simulations to construct a bootstrap confidence interval with a confidence level of 95\%.

We consider simulations with $n=5$ players (which is the common practice for the number of retrieved documents in RAG \cite{xu2024retrieval, yu2024rankrag}), a cost factor $\lambda =  0.5$ (to balance the trade offs induced by it, as we show in Appendix~\ref{apn:faith_effect}, where we study the impact of $\lambda$), a dimension $\text{\di}=2$, the generation function $y(x) = y_{\alpha,z}(x)$, and draw the question $x^*$ and initial documents $x^0_i$ uniformly from $\mathcal{V}$\footnote{We set a fixed random seed of $36$ to ensure reproducibility.}.
As in prior work \cite{Omer_Madmon_2025}, we set $\beta = 1$ for the softmax function and $m = \frac{1}{n}$ for the linear function. In Appendix~\ref{apn:Hyperparameters_Effect} we study the effect of $n$, $\text{\di}$, $\beta$, $m$, and the distribution of the question and the initial documents. The results demonstrate the generality of the findings reported below.

For the analysis of the retrieval weight $\alpha$, we draw $z$ --- the posterior on the user's information need --- from a normal distribution with expectation $\mu = x^*$\footnote{Such sampling might also capture the stochastic sampling from LLMs as part of the generation process \cite{Holtzman2020The}.} and set \K{} $=n$.
For the analysis of the context size \K{}, we set $\alpha=1$ to eliminate the effect of the platform's knowledge $z$ and focus on the dynamics among publishers.

\begin{figure}[t]
    \centering

    \begin{subfigure}[t]{\columnwidth}
    \centering

    \begin{subfigure}[t]
    {0.325\columnwidth}
        \centering
        \includegraphics[width=\linewidth]{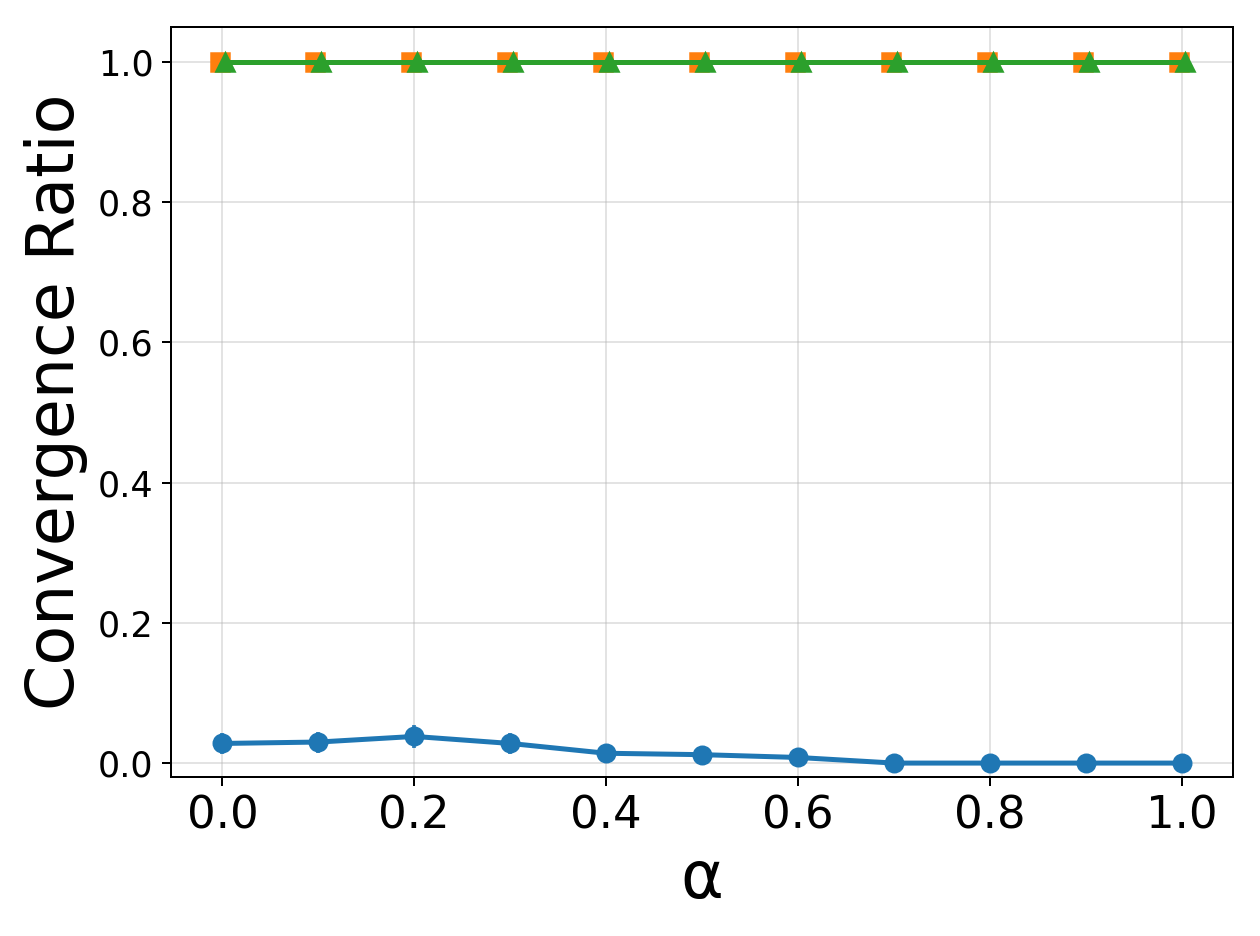}
        \caption{Convergence ratio}
        \label{fig:alpha_convergence_ratio}
    \end{subfigure}
    \hspace{0.06\columnwidth}
    \begin{subfigure}[t]{0.325\columnwidth}
        \centering
        \includegraphics[width=\linewidth]{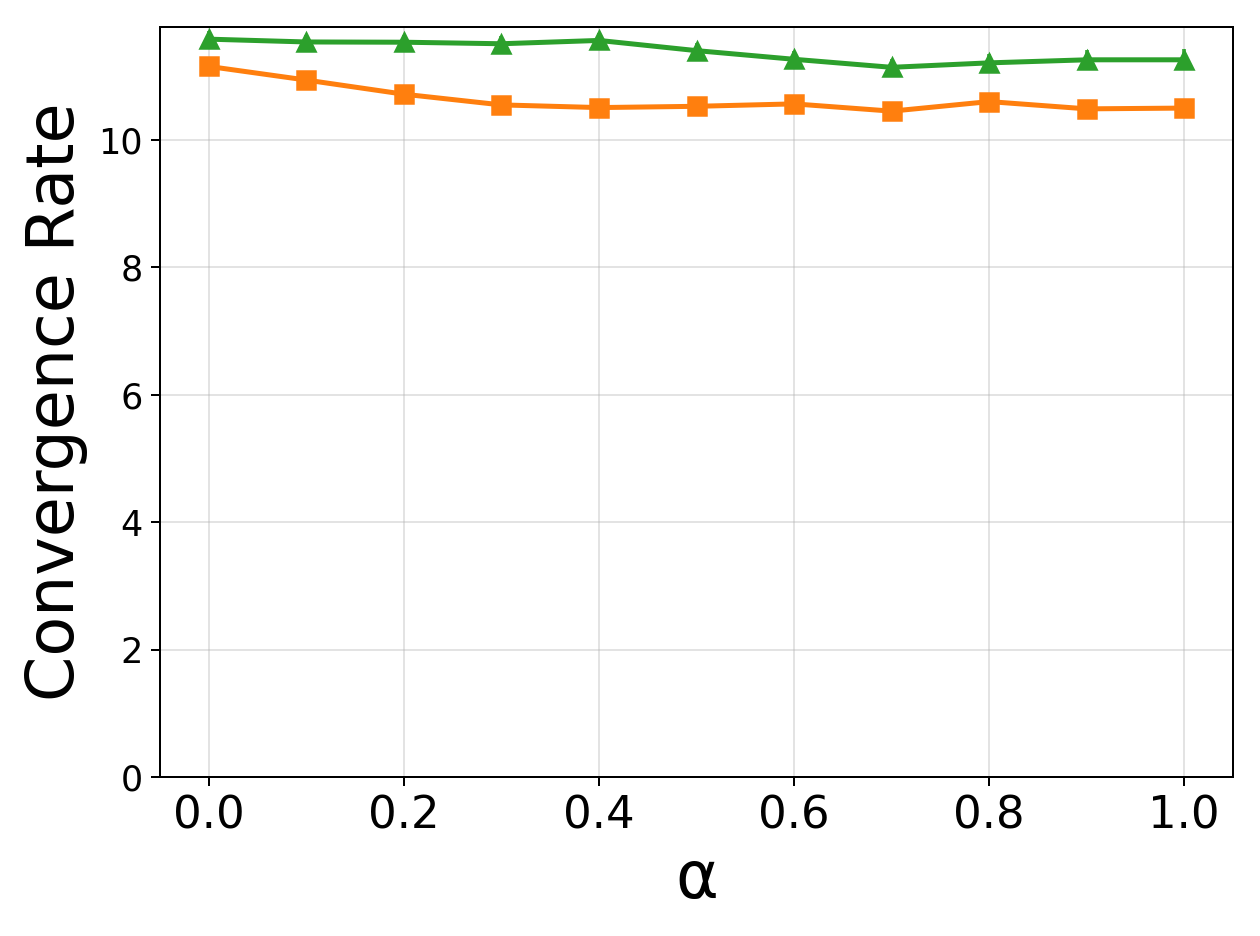}
        
        \caption{Convergence rate}\label{fig:alpha_convergence_rate}
    \end{subfigure}

    \caption*{\textbf{Effect of the value of} $\boldsymbol{\alpha}$}
    \label{fig:alpha_convergence}
\end{subfigure}

    \begin{subfigure}[t]{\columnwidth}
    \centering

    \begin{subfigure}[t]{0.325\columnwidth}
        \centering
        \includegraphics[width=\linewidth]{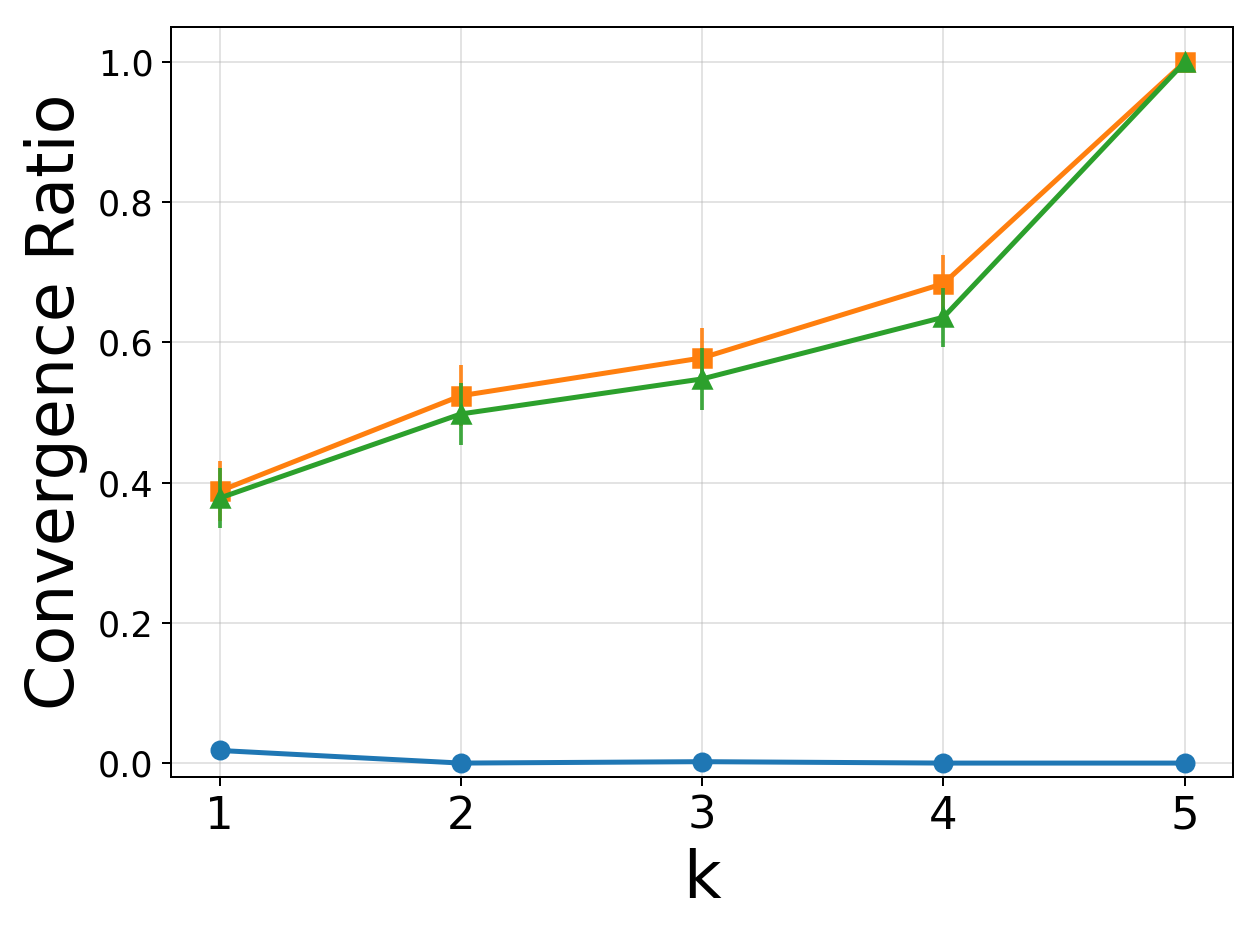}
        
        \caption{Convergence ratio}\label{fig:k_convergence_ratio}
    \end{subfigure}
    \hspace{0.06\columnwidth}
    \begin{subfigure}[t]{0.325\columnwidth}
        \centering
        \includegraphics[width=\linewidth]{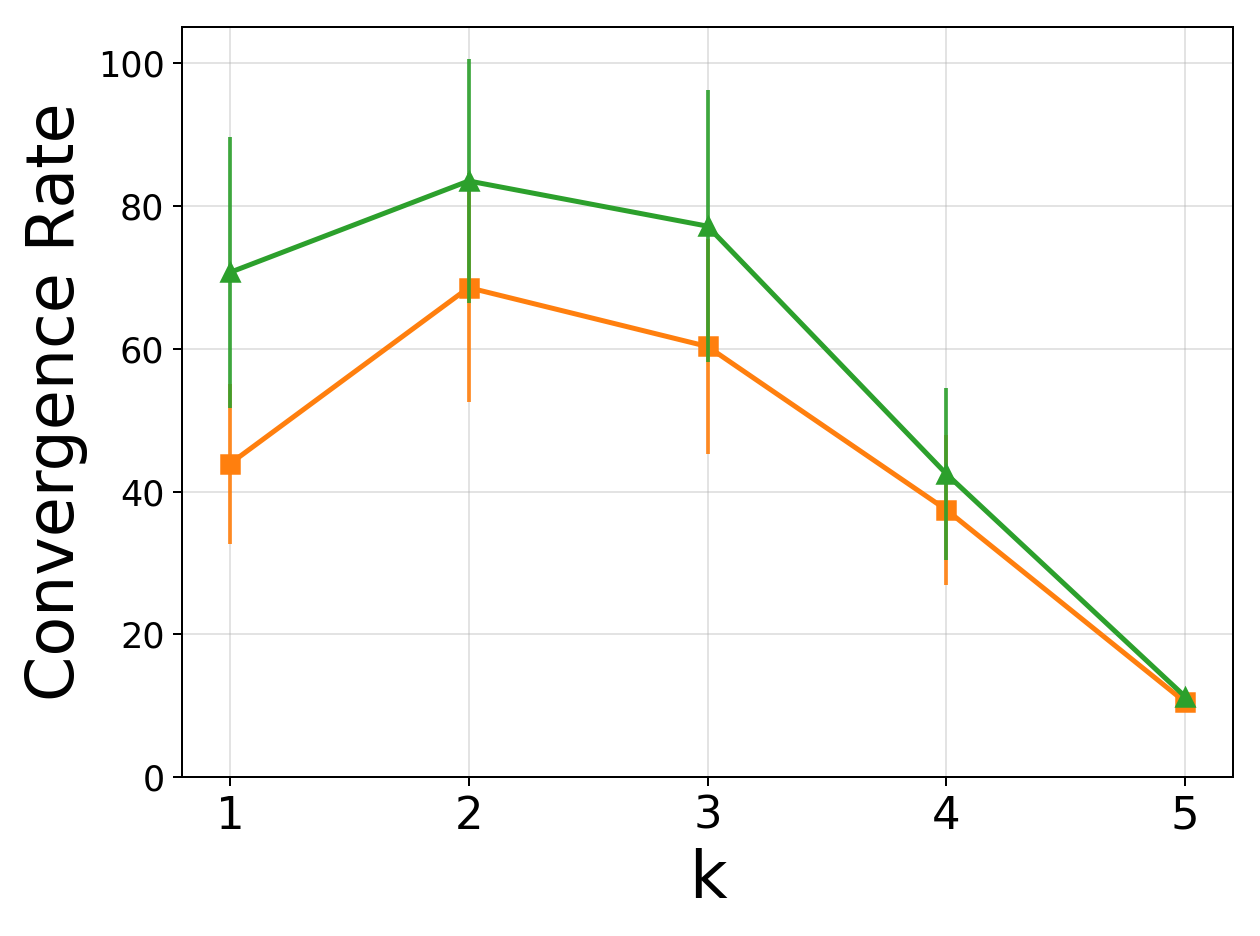}
        
        \caption{Convergence rate}\label{fig:k_convergence_rate}
    \end{subfigure}

    \caption*{\textbf{Effect of the value of} $\boldsymbol{k}$}
    \label{fig:k_convergence}
\end{subfigure}

    \caption{Stability analysis of the \PRP{}, softmax, and linear attribution functions for various values of $\alpha$ and \K{}, shown in blue circles, orange squares, and green triangles, respectively.
    }
    \label{fig:stability_analysis}
\end{figure}

\subsection{Stability Analysis}
We begin with the stability analysis of the GenAI ecosystem.
We report two stability measures: the \emph{convergence ratio}, which is the proportion of simulations that converged, and the \emph{convergence rate}, which is the average number of rounds until the simulation converged, among the converged instances. As we show below, since the dynamics under the \PRP{} function rarely converged, we report the latter only for the softmax and linear functions.
Figure~\ref{fig:stability_analysis} shows the stability results of our mechanisms.

Figure~\ref{fig:alpha_convergence_ratio} confirms the theoretical stability result for the linear function and highlights a substantial limitation of the \PRP{} function: the vast majority of the dynamics it induces fail to converge. In contrast, the dynamics induced by the softmax function converges in all cases. Informally, both the softmax and linear functions appear to induce a stable ecosystem.
We highlight a significant observation in Figure~\ref{fig:alpha_convergence_rate}; the convergence rate values remain relatively stable as the retrieval weight $\alpha$ varies. In other words, increasing the impact of publishers on the downstream response does not destabilize the ecosystem, neither in terms of the convergence ratio nor the convergence rate.
Notably, the softmax function induces faster convergence than the linear function.

We can see in Figure~\ref{fig:k_convergence_ratio} that the convergence ratio under the softmax and linear attribution functions increases with the value of the context size \K, whereas under the \PRP{} function it remains approximately $0$.
This suggests that utilizing more documents from the corpus --- specifically, all documents --- for response generation induces a more stable ecosystem.
This observation can likely be explained by the fact that as the value of \K{} increases, the number of discontinuity points in the utility function decreases.
Additionally, Figure~\ref{fig:k_convergence_rate} shows a parabolic trend for the convergence rate values. We point out that most of the mechanisms are statistically significantly indistinguishable, except for mechanisms with context size \K{} $= n$ (recall that it is the only stable mechanism), for which both the softmax and linear attribution functions exhibit statistically significantly faster convergence than that of mechanisms with context size \K{} $<n$.
This serves as another case in point for the merits of generation with \K{} $=n$.

\captionsetup[subfigure]{
    font=small,
    justification=centering,
    singlelinecheck=true,
    skip=3pt
}

\begin{figure}[t]
    \centering

    \begin{subfigure}[t]{\columnwidth}
    \centering

    \begin{subfigure}[t]{0.325\columnwidth}
        \centering
        \includegraphics[width=\linewidth]{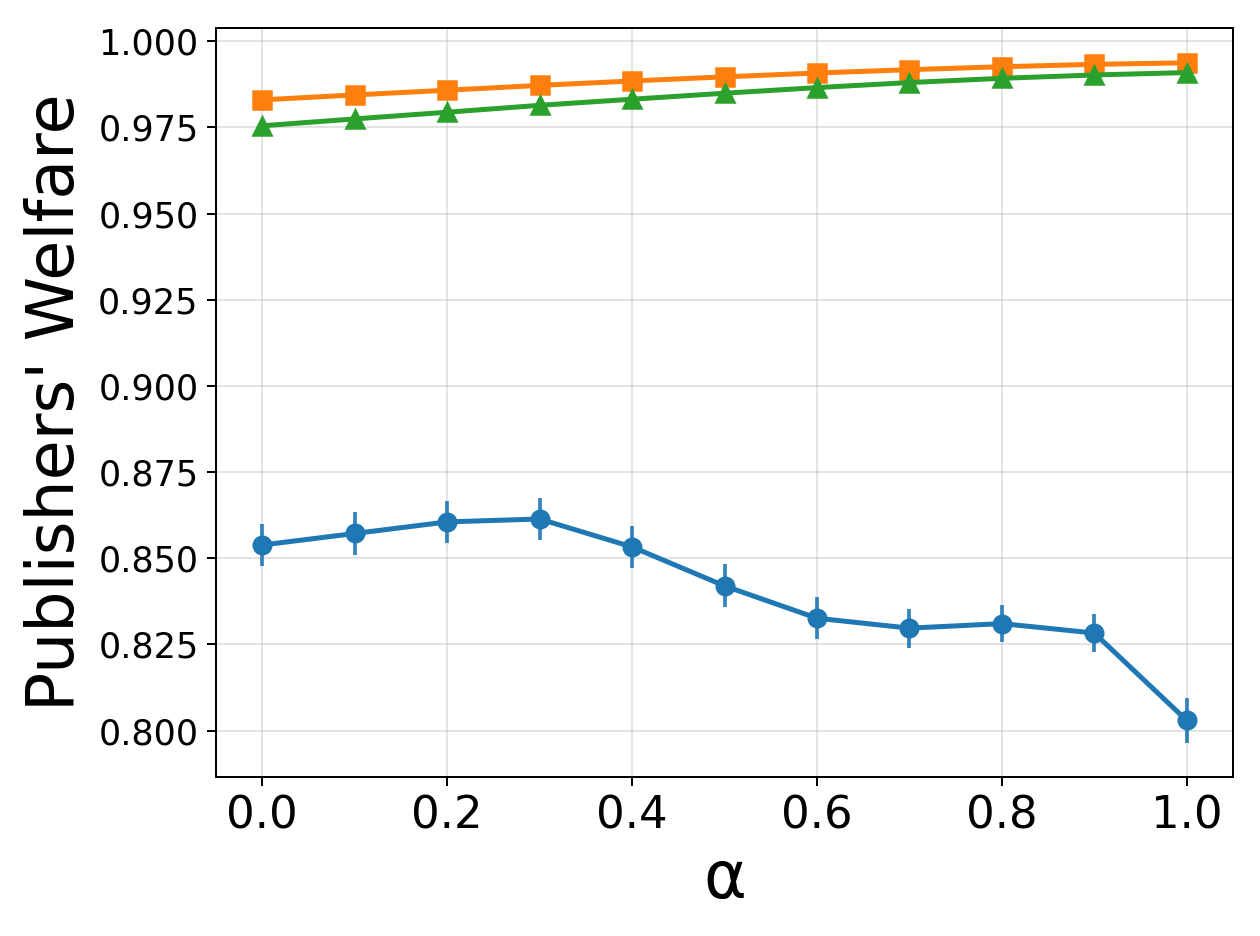}
        
        \caption{PW}\label{fig:welfare_analysis_alpha_PW}
    \end{subfigure}
    \hfill
    \begin{subfigure}[t]{0.325\columnwidth}
        \centering
        \includegraphics[width=\linewidth]{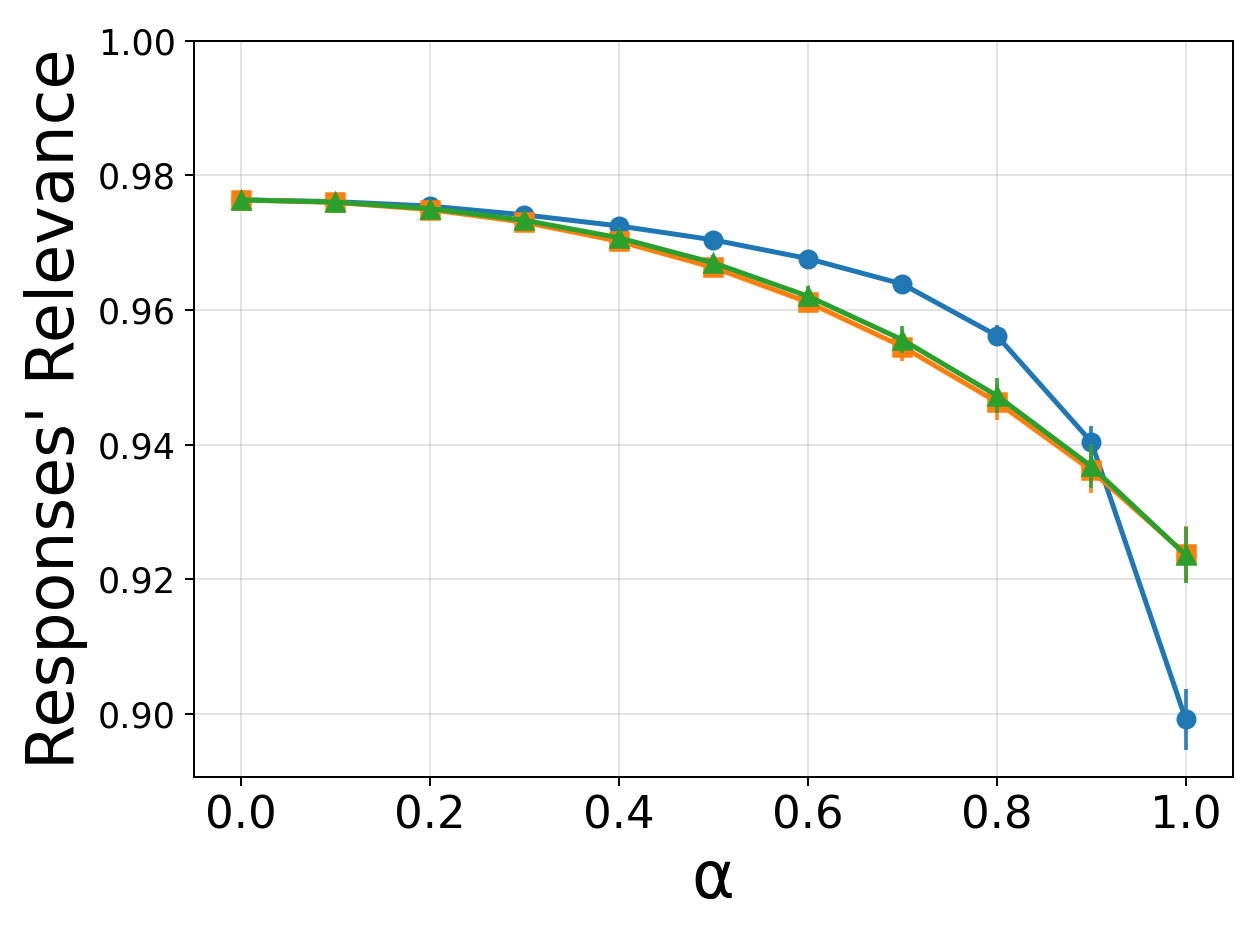}
        
        \caption{RR}\label{fig:welfare_analysis_alpha_RR}
    \end{subfigure}
    \hfill
    \begin{subfigure}[t]{0.325\columnwidth}
        \centering
        \includegraphics[width=\linewidth]{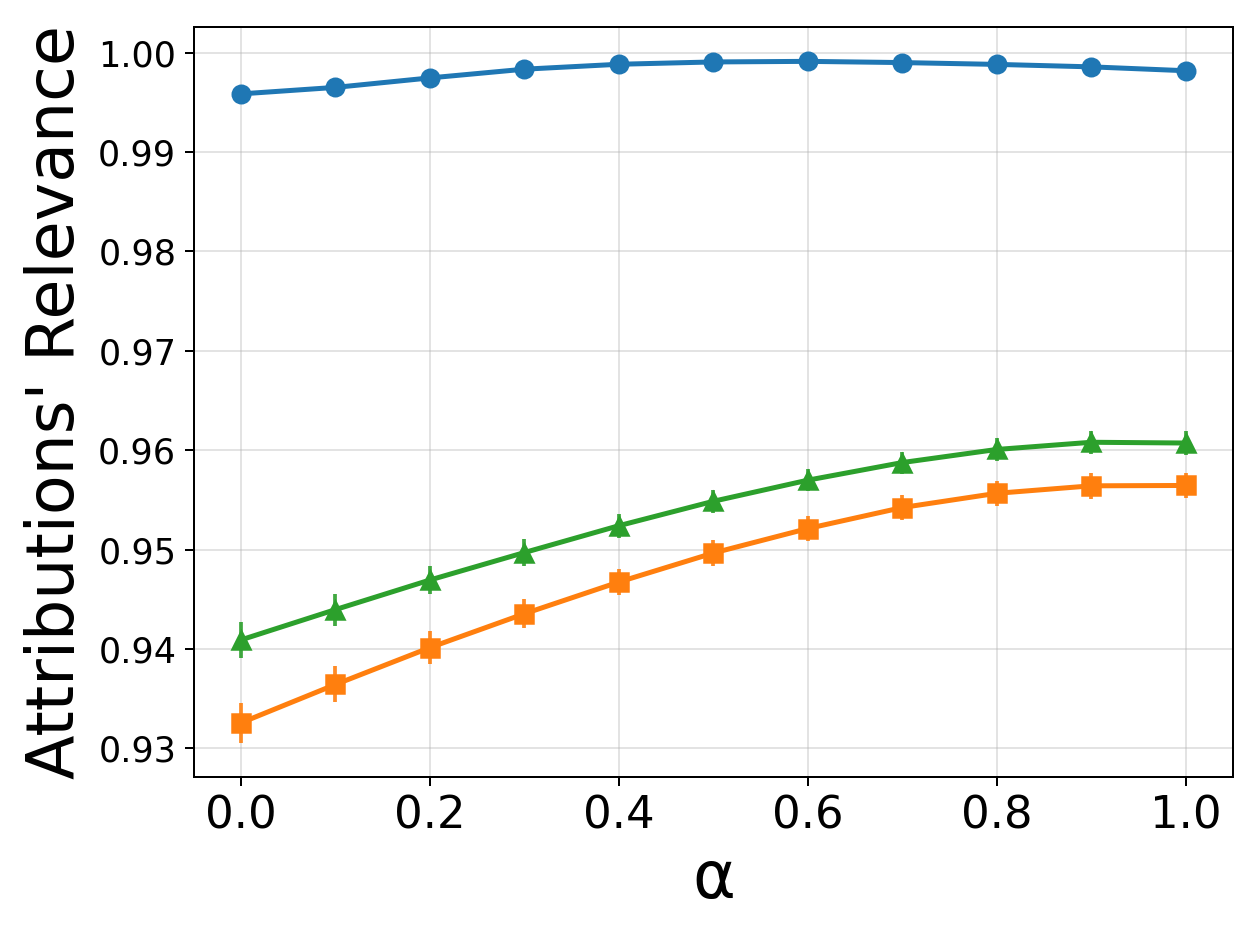}
        
        \caption{AR}\label{fig:welfare_analysis_alpha_AR}
    \end{subfigure}

    \caption*{\textbf{Effect of the value of} $\boldsymbol{\alpha}$}\label{fig:welfare_analysis_alpha}
\end{subfigure}

    \begin{subfigure}[t]{\columnwidth}
    \centering

    \begin{subfigure}[t]{0.325\columnwidth}
        \centering
        \includegraphics[width=\linewidth]{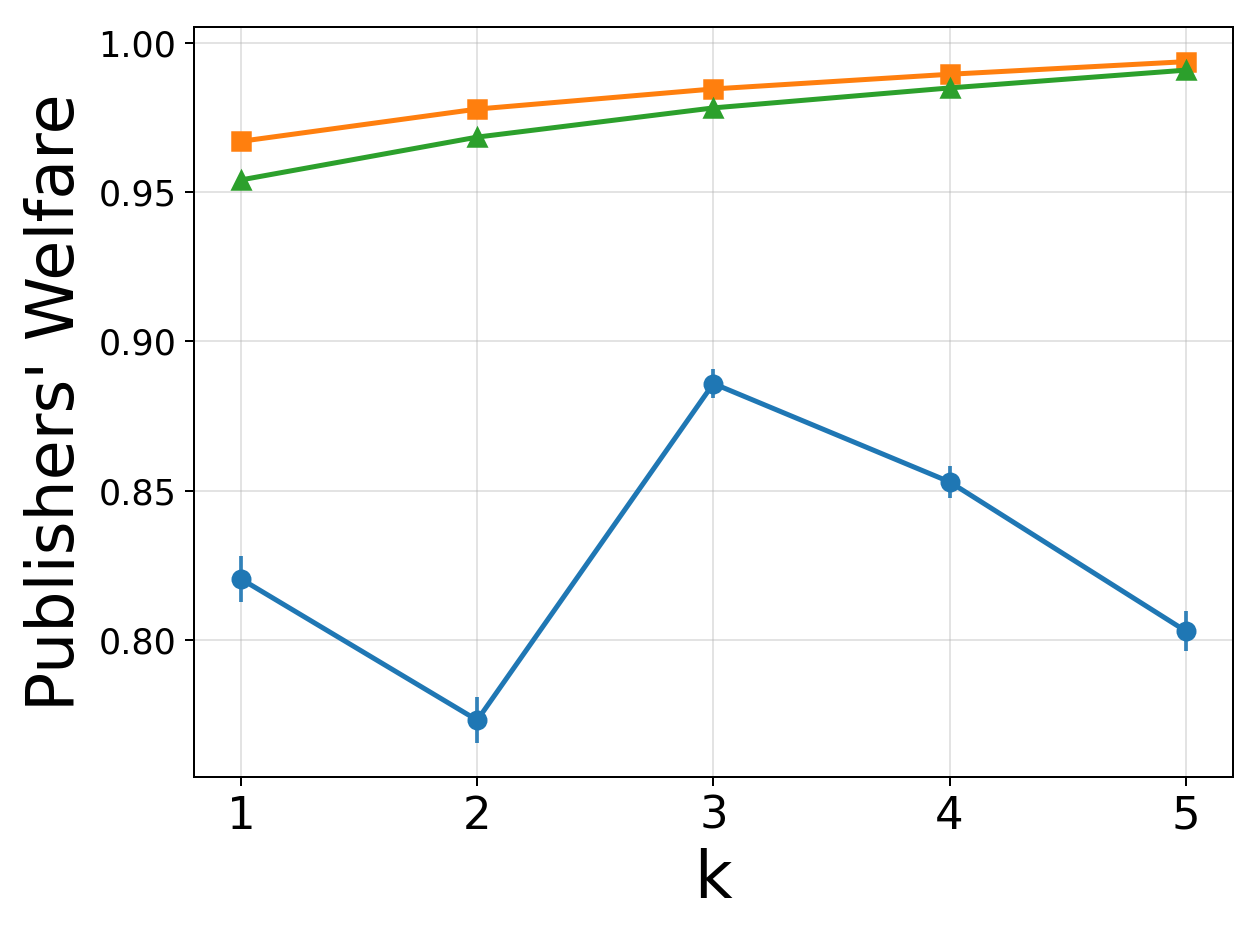}
        
        \caption{PW}\label{fig:welfare_analysis_k_PW}
    \end{subfigure}
    \hfill
    \begin{subfigure}[t]{0.325\columnwidth}
        \centering
        \includegraphics[width=\linewidth]{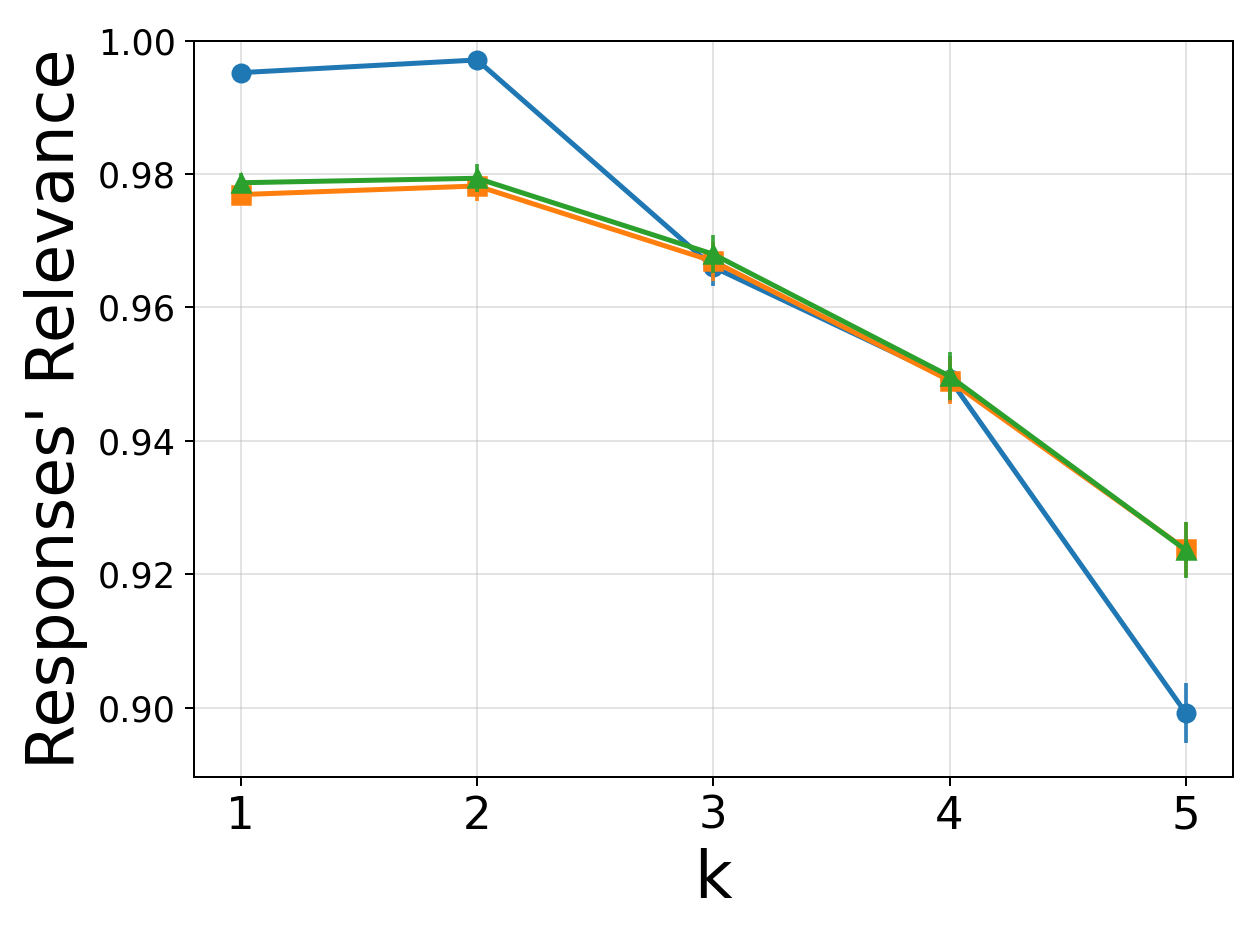}
        
        \caption{RR}\label{fig:welfare_analysis_k_RR}
    \end{subfigure}
    \hfill
    \begin{subfigure}[t]{0.325\columnwidth} 
        \centering
        \includegraphics[width=\linewidth]{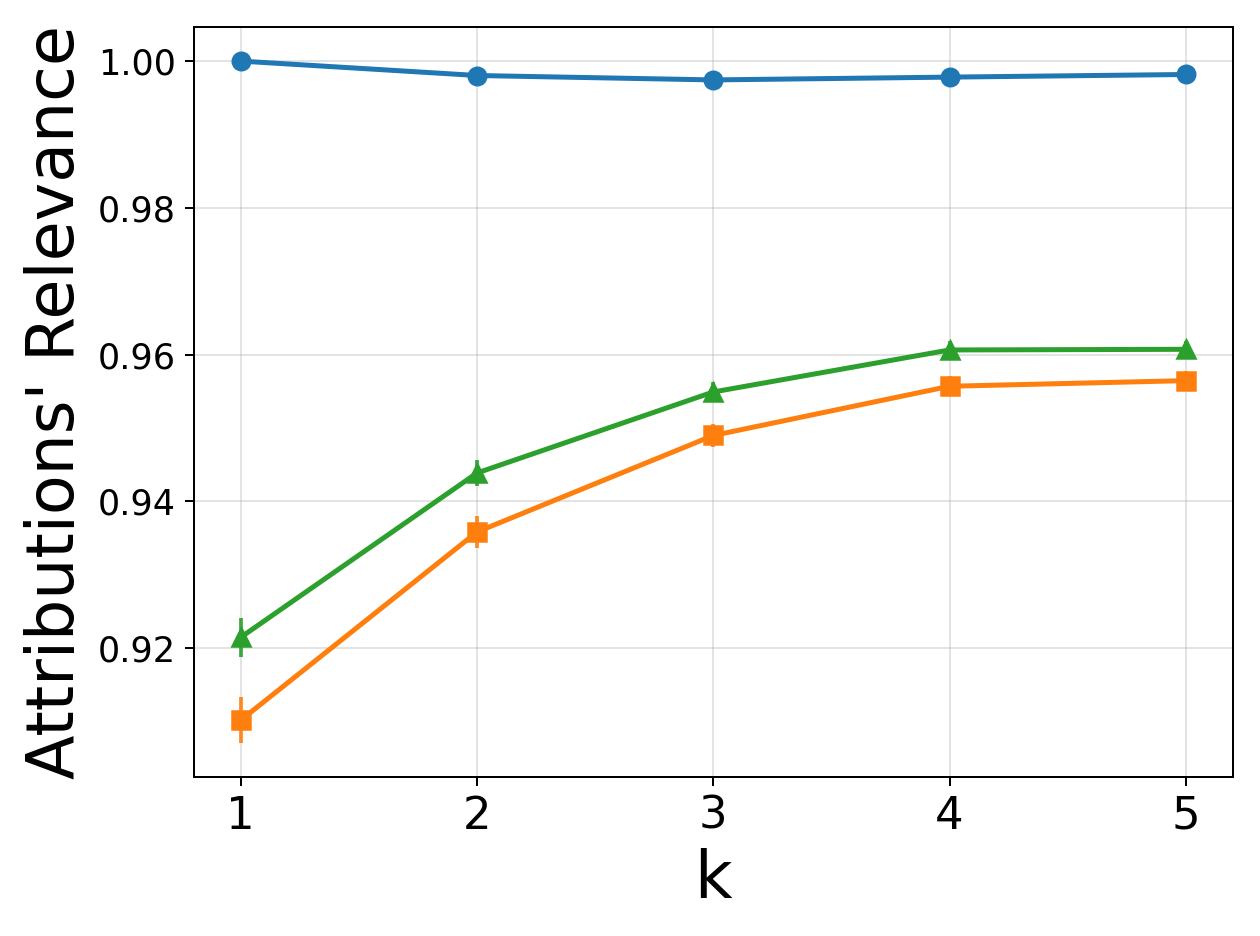}
        
        \caption{AR}\label{fig:welfare_analysis_k_AR}
    \end{subfigure}

    \caption*{\textbf{Effect of the value of} $\boldsymbol{k}$}
    \label{fig:welfare_analysis_k}
\end{subfigure}

    \caption{Welfare analysis of the \PRP{}, softmax, and linear attribution functions for various values of $\alpha$ and \K{}, shown in blue circles, orange squares, and green triangles, respectively.
    PW denotes the publishers' welfare, RR denotes the responses' relevance, and AR denotes the attributions' relevance.}
    \label{fig:welfare_analysis}
\end{figure}

\subsection{Welfare Analysis}
Since, as shown, many significant mechanisms --- e.g., the \PRP{} function and retrieval-based systems (\K{} $<n$) --- induce an unstable ecosystem, we turn to analyzing the ecosystem from the perspective of welfare. 
We focus on the measures defined in Section~\ref{sec: the model}: \emph{publishers’ welfare}, \emph{responses' relevance}, and \emph{attributions' relevance}.
If a simulation converged, we estimate the measures using the strategy profile to which the simulation converged. Otherwise, we average the measures over the last $M=900$ rounds \cite{Omer_Madmon_2025}.

We can see in Figure~\ref{fig:welfare_analysis_alpha_PW} that as $\alpha$ increases, the publishers' welfare under the softmax and linear functions increases. Informally, as publishers gain greater influence over the system's output, they modify their content less.
A likely explanation is that, for larger values of $\alpha$, changes to a document have a stronger effect on the generated response. Thus, in a sense, $\alpha$ may be tuned to motivate publishers to explore new content.
Moreover, the softmax function consistently yields higher publishers' welfare than the linear function, while both induce substantially higher values than those of the \PRP{} function. This observation is due to the discontinuity of the \PRP{} function, forcing publishers to apply extreme modifications to gain exposure. 

We observe in Figures~\ref{fig:welfare_analysis_alpha_RR} and \ref{fig:welfare_analysis_alpha_AR} that as the value of $\alpha$ increases, the relevance of the responses decreases, and the relevance of the attributions increases. That is, as publishers gain influence over the platform's outputs, the responses are less relevant, and, as expected, the attributions are more relevant.
We note that the decrease in responses' relevance is an unwarranted effect of the competition; thus, in a way, $\alpha$ can be used to trade off between the negative and the positive consequences of strategic dynamics.
Additionally, a key observation from Figures~\ref{fig:welfare_analysis_alpha_RR} and \ref{fig:welfare_analysis_alpha_AR} is that the \PRP{} function outperforms the softmax and linear functions in terms of both users' welfare measures. However, we point out that for $\alpha=1$ (i.e., generation is conducted only with publishers' documents), the \PRP{} induces lower responses' relevance than those of the softmax and linear functions.

Figure~\ref{fig:welfare_analysis_k_PW} shows that, under the softmax and linear functions, the publishers' welfare increases with \K{}, while under \PRP{} it exhibits non-monotonic behavior. 
Furthermore, we can see in Figure~\ref{fig:welfare_analysis_k_RR} that increasing \K{} decreases the responses' relevance under all three attribution functions, with the effect being especially pronounced under \PRP{}.
This trend can be attributed to the fact that, as the value of \K{} increases, less relevant documents (with respect to the question) affect generation.
Yet, we point out that increasing the value of \K{} from 1 to 2 increases the responses' relevance.
In contrast, Figure~\ref{fig:welfare_analysis_k_AR} shows that increasing \K{} increases the attributions' relevance under the softmax and linear functions, whereas under \PRP{} it is nearly constant and nearly maximal. 
These observations serve as evidence that unstable mechanisms can yield improved welfare, as we further demonstrate below.

\begin{figure}[t]
    \centering
    \includegraphics[width=0.7\linewidth]{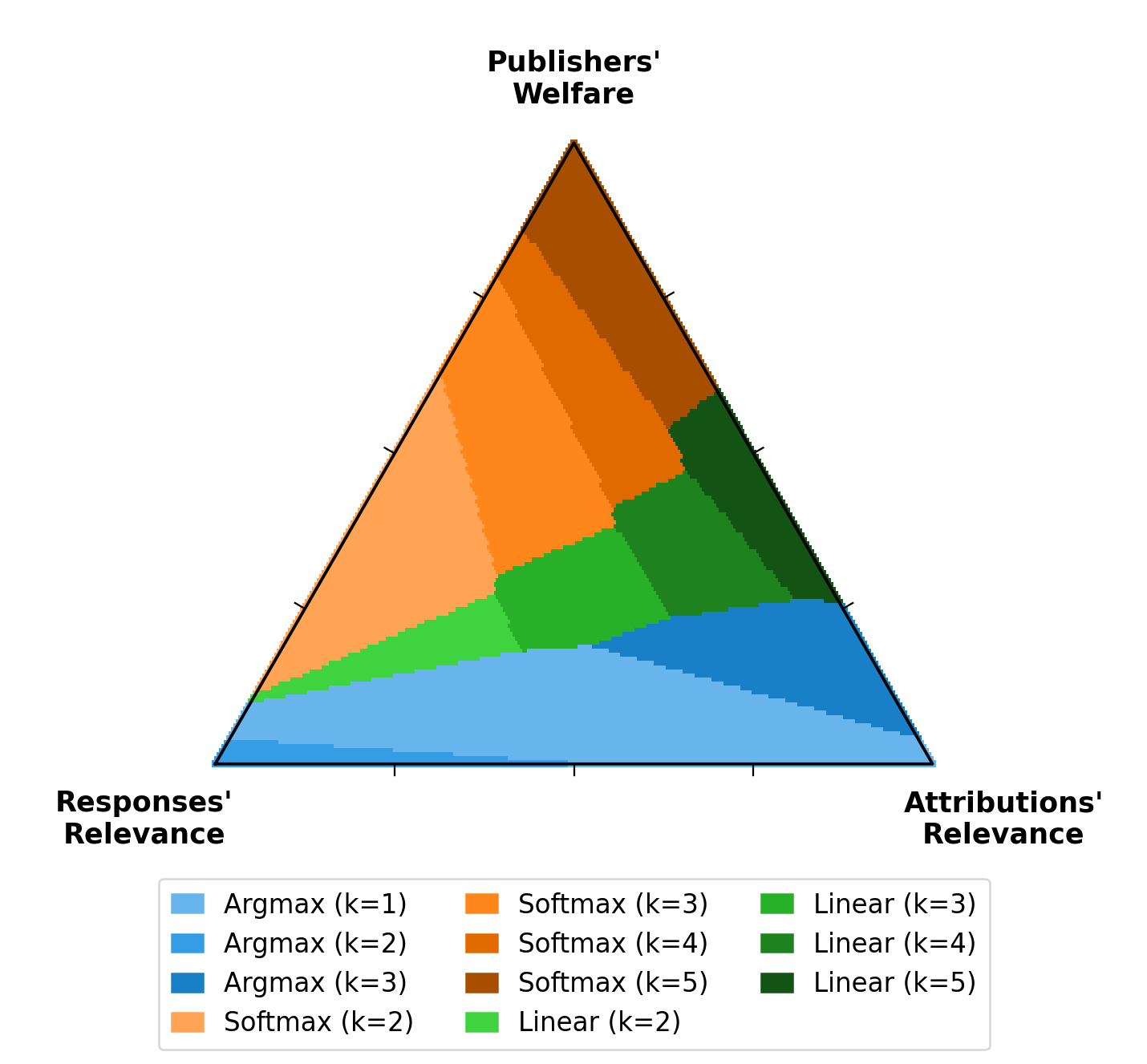}
    \caption{Optimal mechanism as a function of welfare measures distribution (weights). Each point in the ternary diagram represents a distribution over our welfare components: publishers' welfare, responses' relevance, and attributions' relevance. At each point, the color indicates which mechanism maximizes the social welfare. Different colors correspond to different attribution functions, and shades within the same color correspond to different values of \K{}.}
    \label{fig:social_welfare}
\end{figure}

\subsection{Mechanism Design}
Having analyzed welfare under various mechanisms, we now address the fundamental question of \emph{optimal mechanism design} \cite{10.1093/acprof:oso/9780199734023.001.0001}.
Let $f: X \times \Delta^3 \to \mathbb{R}$ be a \textbf{social welfare} function that, given a strategy profile and a distribution (weighting) over the welfare measures used above, assigns a welfare score; specifically, a linear combination of the measures.
We seek to characterize the mechanism that maximizes social welfare for every weighting. Notably, such weighting may vary with the setting, domain, and designer objectives.
As the mechanism space, we consider our three attribution functions and the context size \K{}\footnote{For brevity, we omit $\alpha$ from the mechanism design space.}.

Figure~\ref{fig:social_welfare} presents the optimal mechanism for every distribution over our welfare measures.
The first key observation is that the space is separated into clear regions of distinct optimal mechanism. Each attribution function has a substantial continuous region for which it is part of an optimal mechanism.
In general, mechanisms with the \PRP{} function are optimal when the publishers' welfare is of minor importance, with the softmax function when publishers' welfare is of greater importance, and with the linear function when the weights are more balanced between the welfare measures. 

Additionally, we can see in Figure~\ref{fig:social_welfare} that, in each attribution function region, the optimal value of \K{} varies. For example, as the importance of publishers' welfare decreases and the importance of responses' relevance increases, the \K{} of the optimal mechanism decreases.
We note that the optimal mechanism consists of higher values of \K{} when publishers' welfare is of greater importance and of smaller values of \K{} when responses' and attributions' relevance are of greater importance.
Such an analysis can help a designer to derive guidelines for optimality. For instance, in cases where publishers' welfare is of greater importance, one should include a large number of publishers in generation. 

\section{Concluding Remarks}

This work provides a first step toward understanding the learning dynamics of strategic, attribution-incentivized content creators in GenAI ecosystems. We acknowledge that these ecosystems are still evolving, and that the incentives of their stakeholders remain uncertain. Nevertheless, since GenAI systems increasingly include citations in their responses, a new exposure channel has emerged for content creators.

Our results show that even in complete-information games, convergence to a stable state is a nontrivial (yet achievable) property. This lays the groundwork for future research on learning dynamics in GenAI ecosystems, and in particular on whether less-informed publishers may lead to more stable ecosystems.
In addition, our welfare analysis shows that unstable mechanisms may still have advantages in terms of social welfare. Indeed, the optimal mechanism depends on how social welfare is defined and weighted in the ecosystem. This enables future mechanism-design research aimed at jointly optimizing stability and social welfare.

\bibliography{references}

\clearpage

\appendix

\section{Additional Learning
Dynamics Results} \label{apn:Additional_Learning_Dynamics_Results}
In this section, we provide additional theoretical results regarding stability and learning dynamics.
We consider the model defined in Section~\ref{sec: the model} and the attribution functions described in Section~\ref{sec: ranking functions} (\PRP{}, softmax, and linear). 

\subsection{The Winner Takes All}
As discussed in Section~\ref{sec:PRP}, we first show the shortcomings of the \PRP{} attribution function in terms of ecosystem stability. 
In the following observation, we show that better-response dynamics may not converge in some games under the \PRP{} function.

\begin{observation}\label{obs:2}
For every $0 < \epsilon < \frac{3}{4}$, there exists a generative publishers' game for which there exists an infinite sequence of profitable deviations. That is,
$\varepsilon$-better-response dynamics do not converge.
Specifically, consider the class of games with:
\[
n=3, \quad \text{\di}=1, \quad \lambda=1, \quad x_1^0=x_2^0=x_3^0=\frac12, \quad r=r*, \]
and the generation function $y(x) = y_{1,z}(x)= \frac{1}{n}\sum_{j=1}^n x_j
$ (\K{} $=n$).
Fix any $\rho\in\left(0,\frac14\right)$ and let
$
x^{(0)}=\left(\frac12,\frac12+\rho,\frac12+2\rho\right).
$
Define a sequence $(x^{(t)})_{t\ge 0}$ recursively as follows: at step $t$, the deviating publisher is
$i_t\in\{1,2,3\}$ with cyclic order
$
i_t = 1,2,3,1,2,3,\ldots,
$
and publisher $i_t$ deviates to the midpoint of the other two publishers, i.e.,
\[
x_{i_t}^{(t+1)}=\frac{x_j^{(t)}+x_\ell^{(t)}}{2},
\qquad \{i_t,j,\ell\}=\{1,2,3\}.
\]

Then every deviation in this sequence is profitable. Hence $(x^{(t)})_{t\ge 0}$ is an infinite
better-response path, and in particular the dynamics do not converge.
\end{observation}

\begin{proof}
Let $t\ge 0$, and suppose publisher $i=i_t$ deviates while the other two publishers are $j,\ell$.
By construction we have
$
x_i^{(t+1)}=\frac{x_j^{(t)}+x_\ell^{(t)}}{2}.
$
Therefore,
\begin{align*}
y(x^{(t+1)})
=
\frac{x_i^{(t+1)}+x_j^{(t)}+x_\ell^{(t)}}{3}
=
\frac{\frac{x_j^{(t)}+x_\ell^{(t)}}{2}+x_j^{(t)}+x_\ell^{(t)}}{3}
\\
=
\frac{x_j^{(t)}+x_\ell^{(t)}}{2}
=
x_i^{(t+1)}.
\end{align*}
Hence, after the deviation, publisher $i$ is located exactly at the generated response:
\[
d_i^*(x^{(t+1)}) = d\!\left(x_i^{(t+1)},y(x^{(t+1)})\right)=0.
\]
Since the other two publishers are distinct, publisher $i$ is the unique closest publisher to
$y(x^{(t+1)})$, and thus under \PRP{} receives full attribution:
\[
r_i(x^{(t+1)})=1.
\]

Next, note that all coordinates throughout the process remain in the interval
$
\left[\frac12,\frac12+2\rho\right],
$
because each deviation replaces one coordinate by the midpoint of the other two. Since
$\rho<\frac14$, it follows that for every $t$ and every publisher $m$,
$
\left|x_m^{(t)}-\frac12\right|\le 2\rho < \frac12.
$
Therefore, after publisher $i$ deviates, her movement cost is strictly smaller than $1$:
\[
d_i^0(x_i^{(t+1)}) = \left(x_i^{(t+1)}-\frac12\right)^2 \le (2\rho)^2 < 0.25.
\]
Since $\lambda=1$, her post-deviation utility is:
\begin{align*}
u_i(x^{(t+1)})
&=
r_i(x^{(t+1)})-\lambda d_i^0(x_i^{(t+1)}) \\
&=
1-\left(x_i^{(t+1)}-\frac12\right)^2
>0.75.
\end{align*}

On the other hand, before deviating, publisher $i$ is not the winner in the cyclic construction, so
$r_i(x^{(t)})=0$, and therefore
$
u_i(x^{(t)}) = -\left(x_i^{(t)}-\frac12\right)^2 \le 0.
$
Thus:
\[
u_i(x^{(t+1)}) > 0.75 \ge u_i(x^{(t)}),
\]
so the deviation is profitable.

Since this holds at every step, we obtain an infinite sequence of profitable deviations with utility gain of over $0.75$. In particular, for any $0.75>\epsilon>0$, the above is an example of an $\varepsilon$-better-response dynamics that do not converge.
\end{proof}

Observation~\ref{obs:2} shows that, under the \PRP{} function, better-response dynamics may not converge even in a relatively simple publishers' game in the \K{} $=n$ case. In the empirical results in Section~\ref{sec:empirical_results}, we observed the instability of the ecosystem with better-response dynamics for the \K{} $<n$ case.

\subsection{The Softmax Function}
We now show that in the case where \K{} $<n$, the softmax attribution function described in Section~\ref{sec:Softmax} does not induce a potential game. In fact, we show this for a broader family of attribution functions, the class of proportional attribution functions \cite{madmon2025on}, defined as:

\begin{definition}
    We say an attribution function $r$ is \textbf{proportional} if there exists a twice continuously differentiable function $g:[0,1] \to \R_{++}$ such that:
    \begin{equation}
        r_i(x) = \frac{g \bigl (d^*(x_i) \bigr )}{
        \sumj g \bigl (d^*(x_j) \bigr )}
    \end{equation}
\end{definition}

Notably, the softmax function is proportional.

\begin{theorem}\label{theorem: prop not potential}
Consider the class of generative publishers' games with a proportional attribution function $r_i(x)$ and $n> \text{\K} \ge 1$. 

If this attribution function induces a potential game,
then $g$ is constant on $[0,1]$. Hence the function is uniform:
\[
r_i(x)\equiv \frac{1}{n}.
\]
\end{theorem}

\begin{proof}
Assume by contradiction that there exists a non-constant twice continuously
differentiable function $g:[0,1]\to\mathbb{R}_{++}$ such that the
proportional attribution rule
\[
r_i(x)
=
\frac{g(d_i^\ast(x))}
{\sum_{j=1}^n g(d_j^\ast(x))}
\]
induces an exact potential game for every instance in the class.

We derive a contradiction by considering specific instances. Fix
$\text{\di}=1$, $x^\ast=0$, and
$y(x)
=
y_{1,z}(x)$. 
Fix also an arbitrary cost factor $\lambda>0$ and arbitrary initial documents.

On every open set in which the set of publishers' documents used in generation is fixed --- i.e., no equal distances from $x^*$ for the \K{}-th and \K{}$+1$-th distant documents --- the utility functions are twice continuously differentiable.
Therefore, by Theorem 4.5 of \citet{monderer1996potential}, a necessary condition for the induced game to be an exact potential game is
\[
\frac{\partial^2 u_i(x)}
{\partial x_j\partial x_i}
=
\frac{\partial^2 u_j(x)}
{\partial x_i\partial x_j}.
\]
for every pair of distinct publishers $i,j\in N$.

The content-modification cost of publisher $i$ depends only on $x_i$.
Therefore, it does not contribute to mixed partial derivatives with
respect to the actions of two distinct publishers.

We distinguish between two cases.

\paragraph{Case 1: $1\leq \text{\K{}}\leq n-2$.}

Fix $\varepsilon,\rho>0$ such that
\[
0<\varepsilon<1
\qquad\text{and}\qquad
0<\rho<1-\varepsilon.
\]
Fix arbitrary constants
$\bar{x}_{\text{\K{}}+3},\ldots,\bar{x}_n
\in(\varepsilon+\rho,1),$ 
and consider profiles of the form:
\[
x_1=\cdots=x_{\text{\text{\K{}}}}=\varepsilon,
\qquad
x_{\text{\text{\K{}}}+1},x_{\text{\text{\K{}}}+2}
\in(\varepsilon+\rho,1),
\]
\[
x_\ell=\bar{x}_\ell
\qquad
\text{for every }
\ell\in\{\text{\K{}}+3,\ldots,n\}.
\]
When $n=\text{\K{}}+2$, the last condition is vacuous.

Since $x^\ast=0$, publishers $1,\ldots,\text{\K{}}$ are exactly the
$\text{\K{}}$ publishers closest to the user's question. Therefore,
\[
y(x)
=
\frac{1}{\text{\K{}}}
\sum_{j=1}^{\text{\K{}}}x_j
=
\varepsilon.
\]
Consequently,
$d_j^\ast(x)=0 
\text{ for every }j\in\{1,\ldots,\text{\K{}}\}.$
Define
\[
s
:=
(x_{\text{\text{\K{}}}+1}-\varepsilon)^2
\qquad\text{and}\qquad
t
:=
(x_{\text{\text{\K{}}}+2}-\varepsilon)^2.
\]
For every $\ell\in\{\text{\K{}}+3,\ldots,n\}$, we have
\[
d_\ell^\ast(x)
=
(\bar{x}_\ell-\varepsilon)^2,
\]
which is constant with respect to $x_{\text{\K{}}+1}$ and
$x_{\text{\K{}}+2}$.

Define
\[
A
:=
\text{\K{}}g(0)
+
\sum_{\ell=\text{\K{}}+3}^n
g\left((\bar{x}_\ell-\varepsilon)^2\right),
\]
where the sum is interpreted as zero when $n=\text{\K{}}+2$.
Since $g$ is strictly positive, $A>0$.

The denominator of the proportional attribution rule is
\[
D(s,t)
:=
A+g(s)+g(t).
\]
Hence,
\[
r_{\text{\K{}}+1}(x)
=
\frac{g(s)}{D(s,t)}
\qquad\text{and}\qquad
r_{\text{\K{}}+2}(x)
=
\frac{g(t)}{D(s,t)}.
\]

Direct differentiation yields
\begin{align*}
\frac{\partial^2 u_{\text{\K{}}+1}(x)}
{\partial x_{\text{\K{}}+2}\partial x_{\text{\K{}}+1}}
&=
\frac{
4(x_{\text{\K{}}+1}-\varepsilon)
(x_{\text{\K{}}+2}-\varepsilon)
g'(s)g'(t)
}{
D(s,t)^3
} \\
&\cdot
\left(
g(s)-A-g(t)
\right),
\end{align*}
and
\begin{align*}
\frac{\partial^2 u_{\text{\K{}}+2}(x)}
{\partial x_{\text{\K{}}+1}\partial x_{\text{\K{}}+2}}
&=
\frac{
4(x_{\text{\K{}}+1}-\varepsilon)
(x_{\text{\K{}}+2}-\varepsilon)
g'(s)g'(t)
}{
D(s,t)^3
} \\
&\cdot
\left(
g(t)-A-g(s)
\right).
\end{align*}

By the mixed-partial condition, these expressions must be equal. Subtracting the second expression from the first gives
\[
\frac{
8(x_{\text{\K{}}+1}-\varepsilon)
(x_{\text{\K{}}+2}-\varepsilon)
g'(s)g'(t)
\left(g(s)-g(t)\right)
}{
D(s,t)^3
}
=
0.
\]
Since
\[
x_{\text{\K{}}+1}>\varepsilon,
\qquad
x_{\text{\K{}}+2}>\varepsilon,
\qquad
D(s,t)>0,
\]
we obtain
\[
g'(s)g'(t)\left(g(s)-g(t)\right)=0.
\]

We next show that this identity holds
for every $s,t\in(0,1)$. Fix arbitrary $s,t\in(0,1)$. Choose
$\varepsilon>0$ sufficiently small such that
\[
\varepsilon+\max\{\sqrt{s},\sqrt{t}\}<1,
\]
and choose
\[
0<\rho<\min\{\sqrt{s},\sqrt{t}\}.
\]
Setting
\[
x_{\text{\K{}}+1}=\varepsilon+\sqrt{s}
\qquad\text{and}\qquad
x_{\text{\K{}}+2}=\varepsilon+\sqrt{t}
\]
produces the desired values of $s$ and $t$. Thus,
the identity holds for every
$s,t\in(0,1)$.

Suppose that there exists $s_0\in(0,1)$ such that $g'(s_0)\neq0$.
Since $g'$ is continuous, there exists an open interval
$I\subseteq(0,1)$ containing $s_0$ such that
\[
g'(s)\neq0
\qquad
\text{for every }s\in I.
\]
For every $s,t\in I$, the identity implies
\[
g(s)=g(t).
\]
Therefore, $g$ is constant on $I$, contradicting $g'(s_0)\neq0$.
It follows that
\[
g'(s)=0
\qquad
\text{for every }s\in(0,1).
\]
Hence, $g$ is constant on $[0,1]$, contradicting the assumption that
$g$ is non-constant.

\paragraph{Case 2: $\text{\K{}}=n-1$.}

In this case, there is only one publisher's document that is not used in generation, so the
construction used in the previous case cannot be applied. 

Let publisher $n$ be the most distant publisher from the question. Fix an
arbitrary $b_0\in(0,1)$ and choose
$\varepsilon\in(0,1-b_0)$. There exist an open interval
$B\subseteq(0,1-\varepsilon)$ containing $b_0$ and a constant
$\delta>0$ such that, for every $b\in B$ and
$q\in(-\delta,\delta)$, the profile defined below belongs to
$(0,1)^n$ and publishers $1,\ldots,\text{\K{}}$ are exactly the
$\text{\K{}}$ publishers closest to $x^\ast=0$.

For every $(q,b)\in(-\delta,\delta)\times B$, define the strategy
profile $x(q,b)\in X$ by
\[
x_1(q,b)
=
\varepsilon+\text{\K{}}q, \qquad
x_n(q,b)
=
\varepsilon+q+b,
\]
\[
x_j(q,b)
=
\varepsilon
\qquad
\text{for every }
j\in\{2,\ldots,\text{\K{}}\}.
\]
When $\text{\K{}}=1$, the condition concerning publishers
$2,\ldots,\text{\K{}}$ is vacuous.

Indeed, for sufficiently small $\delta$, we have
\[
x_n(q,b)>x_j(q,b)
\qquad
\text{for every }
j\in\{1,\ldots,\text{\K{}}\}.
\]
Since all coordinates are positive and $x^\ast=0$, publisher $n$ is
farther from the question than publishers $1,\ldots,\text{\K{}}$.
Consequently, the set of publishers used in generation is fixed
throughout this neighborhood.

For every such profile, the generated response is
\[
y(x(q,b))
=
\frac{1}{\text{\K{}}}
\sum_{j=1}^{\text{\K{}}}x_j(q,b).
\]
Substituting the definition of $x(q,b)$ gives
\begin{align*}
y(x(q,b))
&=
\frac{1}{\text{\K{}}}
\left(
\varepsilon+\text{\K{}}q
+
(\text{\K{}}-1)\varepsilon
\right)
=
\varepsilon+q.
\end{align*}

Therefore, the distances from the generated response are
\[
d_1^\ast(x(q,b))
=
\left(
x_1(q,b)-y(x(q,b))
\right)^2
=
(\text{\K{}}-1)^2q^2,
\]
\[
d_j^\ast(x(q,b))
=
q^2
\qquad
\text{for every }
j\in\{2,\ldots,\text{\K{}}\},
\]
\[
d_n^\ast(x(q,b))
=
\left(
x_n(q,b)-y(x(q,b))
\right)^2
=
b^2.
\]

To simplify notation, define the restrictions of the attribution
functions to this family of profiles by
\[
R_i(q,b)
:=
r_i(x(q,b))
\qquad
\text{for every }i\in N.
\]
Notice that $r_i$ remains a function of strategy profiles. The function
$R_i$ is merely the composition of $r_i$ with the profile-valued map
$(q,b)\mapsto x(q,b)$.

Define
\[
D(q,b)
:=
g\left((\text{\K{}}-1)^2q^2\right)
+
(\text{\K{}}-1)g(q^2)
+
g(b^2).
\]
It follows that
\[
R_1(q,b)
=
\frac{
g\left((\text{\K{}}-1)^2q^2\right)
}{
D(q,b)
}
\qquad
R_n(q,b)
=
\frac{
g(b^2)
}{
D(q,b)
}.
\]

At $q=0$, define
\[
H(b)
:=
R_1(0,b)
=
\frac{
g(0)
}{
\text{\K{}}g(0)+g(b^2)
}.
\]
At the profile $x(0,b)$, every publisher used in the generation process has a distance of zero from the generated response. Hence, all
$\text{\K{}}$ publishers included in the context receive the same
attribution score $H(b)$. Since the attribution scores sum to one,
\[
R_n(0,b)
=
1-\text{\K{}}H(b).
\]

Both $R_1(q,b)$ and $R_n(q,b)$ are even functions of $q$, because
$q$ appears only through $q^2$. Therefore,
\[
\frac{\partial R_1}{\partial q}(0,b)
=
0
\qquad\text{and}\qquad
\frac{\partial R_n}{\partial q}(0,b)
=
0.
\]
Differentiating these identities with respect to $b$ yields
\[
\frac{\partial^2 R_1}
{\partial b\partial q}(0,b)
=
0
\qquad\text{and}\qquad
\frac{\partial^2 R_n}
{\partial b\partial q}(0,b)
=
0.
\]

We now relate derivatives with respect to $q$ and $b$ to derivatives
with respect to the original actions $x_1$ and $x_n$. On the
two-dimensional family of profiles defined above we get
\[
q
=
\frac{x_1-\varepsilon}{\text{\K{}}},
\qquad
b
=
x_n-\varepsilon
-
\frac{x_1-\varepsilon}{\text{\K{}}}.
\]
Consequently, the chain rule gives
\[
\frac{\partial}{\partial x_1}
=
\frac{1}{\text{\K{}}}
\left(
\frac{\partial}{\partial q}
-
\frac{\partial}{\partial b}
\right),
\qquad
\frac{\partial}{\partial x_n}
=
\frac{\partial}{\partial b}.
\]

The content-modification cost of each publisher depends only on that
publisher's own action. Therefore, it does not contribute to mixed
partial derivatives with respect to $x_1$ and $x_n$.

Evaluating the mixed partial derivative of publisher $1$'s utility at
the profile $x(0,b)$ gives
\begin{align*}
\frac{\partial^2 u_1}
{\partial x_n\partial x_1}(x(0,b))
&=
\frac{1}{\text{\K{}}}
\left(
\frac{\partial^2 R_1}
{\partial b\partial q}(0,b)
-
\frac{\partial^2 R_1}
{\partial b^2}(0,b)
\right)
\\
&=
-\frac{1}{\text{\K{}}}H''(b).
\end{align*}

Similarly, the mixed partial derivative of publisher $n$'s utility is
\begin{align*}
\frac{\partial^2 u_n}
{\partial x_1\partial x_n}(x(0,b))
&=
\frac{1}{\text{\K{}}}
\left(
\frac{\partial^2 R_n}
{\partial q\partial b}(0,b)
-
\frac{\partial^2 R_n}
{\partial b^2}(0,b)
\right)
\\
&=
-\frac{1}{\text{\K{}}}
\frac{\mathrm{d}^2}{\mathrm{d}b^2}
\left(
1-\text{\K{}}H(b)
\right)
\\
&=
H''(b).
\end{align*}

As stated above, since the induced game is assumed to be an exact potential game, the
mixed-partial condition implies
\[
\frac{\partial^2 u_1}
{\partial x_n\partial x_1}(x(0,b))
=
\frac{\partial^2 u_n}
{\partial x_1\partial x_n}(x(0,b)).
\]
Therefore,
\[
-\frac{1}{\text{\K{}}}H''(b)
=
H''(b).
\]
Since $\K{}\geq1$, we obtain
$H''(b)=0.$
The point $b_0\in(0,1)$ was arbitrary. Thus,
\[
H''(b)=0
\qquad
\text{for every }b\in(0,1),
\]
and $H$ is affine on $(0,1)$.

Moreover,
\[
H'(b)
=
-
\frac{
2b\,g(0)g'(b^2)
}{
\left(
\text{\K{}}g(0)+g(b^2)
\right)^2
}.
\]
Since $g'$ is continuous on $[0,1]$,
$\lim_{b\to0^+}H'(b)=0.$
Because $H$ is affine on $(0,1)$, its derivative is constant.
Consequently,
\[
H'(b)=0
\qquad
\text{for every }b\in(0,1).
\]
Hence, $H$ is constant on $(0,1)$.

Finally,
\[
H(b)
=
\frac{
g(0)
}{
\text{\K{}}g(0)+g(b^2)
}.
\]
Since $g(0)>0$, the constancy of $H$ implies that $g(b^2)$ is
constant for every $b\in(0,1)$. Since $b^2$ ranges over $(0,1)$,
$g$ is constant on $(0,1)$. By continuity, $g$ is constant on
$[0,1]$, contradicting the assumption that $g$ is non-constant.
\end{proof}

Theorem \ref{theorem: prop not potential} implies, in particular, that any proportional attribution function --- e.g., the softmax function --- does not induce a potential game.
One might naturally question whether the convergence of better-response dynamics under the softmax attribution function is still guaranteed, despite the game not being a potential game. However, as we showed in the main paper (Observation~\ref{obs:3}), this is not the case.

\subsection{The Linear Relative Relevance Function}
Here, we show an example of the non-convergence of better-response dynamics under the linear function in the \K{} $<n$ case. Recall that in Section~\ref{linear_theory} we proved that, in this setting, the game is potentially not a potential game.
The example below establishes the instability of mechanisms with the linear attribution function and generation based on \K{} $<n$ publishers.

\begin{observation}\label{obs:example_for_non_conv_with_linear_k<n}
There exists a generative publishers' game with \K{} $<n$ and the linear
attribution function for which better-response dynamics need not converge.
In particular, there exists an infinite $\varepsilon$-better-response path
for some $\varepsilon>0$.
\end{observation}

\begin{proof}
Consider a generative publishers' game instance with:
\[
    n=3,\quad \text{\di{}}=1,\quad x^*=\frac12, \quad x_i^0=\frac12,\quad
    \lambda=\frac{1}{1000},
\]
the generation function $y(x) = y_{1,z}(x)= \frac{1}{2}\sum_{j=1}^2 x_j$ (\K{} $=2$), and the linear attribution function $r=\hat{r}$ with slope $m=\frac13$.

Consider the following eight profiles:
\[
\begin{aligned}
A&=\left(0,\frac15,\frac15\right),&
B&=\left(0,\frac25,\frac15\right),&
C&=\left(0,\frac35,\frac15\right),\\
D&=\left(0,\frac45,\frac15\right), &
E&=\left(0,\frac45,\frac25\right),&
F&=\left(1,\frac45,\frac25\right),\\
G&=\left(1,\frac15,\frac25\right),&
H&=\left(1,\frac15,\frac15\right).
\end{aligned}
\]
We claim that the following cycle is a better-response path:
\[
    A \xrightarrow{2} B
    \xrightarrow{2} C
    \xrightarrow{2} D
    \xrightarrow{3} E
    \xrightarrow{1} F
    \xrightarrow{2} G
    \xrightarrow{3} H
    \xrightarrow{1} A .
\]

For each step, the deviating publisher's increase in relative relevance
and utility is as follows:
\[
\begin{array}{c|c|c|c}
\text{Step} & \text{Deviator} &
\Delta \nu_i &
\Delta u_i \\ \hline
A\to B & 2 & \frac{1}{50} & \frac{253}{37500} \\
B\to C & 2 & \frac{1}{50} & \frac{1}{150} \\
C\to D & 2 & \frac{1}{50} & \frac{247}{37500} \\
D\to E & 3 & \frac{2}{25} & \frac{1003}{37500} \\
E\to F & 1 & \frac{1}{5} & \frac{1}{15} \\
F\to G & 2 & \frac{9}{50} & \frac{3}{50} \\
G\to H & 3 & \frac{2}{25} & \frac{997}{37500} \\
H\to A & 1 & \frac{3}{5} & \frac{1}{5}
\end{array}
\]
All utility gains are strictly positive. Therefore each step is a better
response. Moreover,
\[
    \min \Delta u_i=\frac{247}{37500}>0.
\]
Hence, for every
\[
    0<\varepsilon<\frac{247}{37500},
\]
each step in the cycle is an $\varepsilon$-better response. Repeating the
cycle indefinitely yields an infinite $\varepsilon$-better-response path,
so the dynamics do not converge.
\end{proof}

\section{Omitted Proofs}\label{apn:omitted_proofs}

\begin{proof}[Proof of Observation \ref{obs:1}]
We divide the proof into two parts: the \K{} $<n$ case and the \K{} $=n$ case. We start with the \K{} $=n$ case.

Let $n \ge 5$, let $\lambda \in \left(\frac{4}{n},1\right]$, and consider a generative publishers' game $G$ in which:
\[
\text{\di}=1, \quad
x_1^0=1,\quad x_2^0=\cdots=x_n^0=0, 
\]
the generation function $y(x) = y_{1,z}(x)= \frac{1}{n}\sum_{j=1}^n x_j$ (\K{} $=n$), an arbitrary $x^*\in[0,1]$, and the \PRP{} attribution function $r=r^*$.

Assume by contradiction that $x^{eq}$ is a PNE, and write
$
y:=y(x^{eq})=\frac1n\sum_{j=1}^n x_j^{eq}.
$

We first rule out the case where all strategies are equal. If
$
x_1^{eq}=\cdots=x_n^{eq}=a,
$
then under the \PRP{}, each publisher gets attribution probability of $\frac1n$. Hence:
\[
u_1(x^{eq})=\frac1n-\lambda(1-a)^2,
\qquad
u_i(x^{eq})=\frac1n-\lambda a^2 \quad (i\ge 2).
\]
For this profile to be a PNE, all these utilities must be nonnegative (otherwise, at least one publisher has a profitable deviation to her initial document), so we would need
\[
a \le \frac1{\sqrt{n\lambda}}
\qquad\text{and}\qquad
1-a \le \frac1{\sqrt{n\lambda}}.
\]
Thus,
\[
1 \le \frac{2}{\sqrt{n\lambda}},
\]
i.e.,
\[
\lambda \le \frac{4}{n},
\]
contrary to $\lambda>\frac4n$. Therefore $x^{eq}$ is not an equal strategies profile.

Next, in PNE under the \PRP{}, every non-winning publisher\footnote{Here, a winning publisher is any publisher receiving a positive attribution score.} must play her initial document. For example, if publisher $1$ is not a winner
and $x_1^{eq}\neq 1$, then deviating to $1$ yields a utility of at least $0$, which is strictly better than her current
utility $-\lambda(1-x_1^{eq})^2<0$. 

We now claim that publisher $1$ cannot be a winner. Suppose otherwise. Since the profile is not all equal,
there is at least one non-winner, and every non-winner is at her initial document. In particular, every losing
publisher is at $0$.

Also, a winning publisher cannot lie above $y$: if some winner $i\ge 2$ satisfies
$x_i^{eq}>y$, then by moving slightly toward $y$, publisher $i$ becomes strictly closer to the new average
than every other publisher, while strictly decreasing her movement cost. Hence this would be a profitable
deviation. By the same argument, if publisher $1$ is a winner then necessarily $x_1^{eq}\ge y$.

Therefore, if publisher $1$ is a winner, then all winning publishers $i \ge 2$ lie at some point $y-\rho$
and publisher $1$ lies at $y+\rho$, for some $\rho\ge 0$. Let $t\ge 1$ be the number of winning publishers. Since all losers are at $0$, the generation function gives:
\[
ny = (y+\rho) + (t-1)(y-\rho),
\]
hence
\[
(n-t)y = (2-t)\rho.
\]
Notably, if $t\ge3$ the generation response would be closer to $y-\rho$ (i.e., publisher 1 is not the winner). Therefore, it holds that $t\le2$.

Suppose first that \(t=1\). Let publisher \(1\) be the unique winning
publisher. Since all other publishers are located at \(0\), the generated
response satisfies
$y=\frac{x_1}{n}.$
Hence, the distance of publisher \(1\) from the generated response is
\[
\rho = |x_1-y| = (n-1)y,
\]
whereas the distance of every other publisher from the generated response is
\(y\). Because publisher \(1\) is the unique winner, it must be at least as
close to the generated response as every other publisher. Therefore,
\[
\rho \leq y.
\]
Combining the two relations gives
\[
(n-1)y \leq y.
\]
Since \(n\geq 5\), this implies \(y=0\), and consequently \(x_1=0\).
Thus, all publishers choose the same strategy, contradicting the case under
consideration. Therefore, \(t\neq 1\).

If $t=2$, the left-hand side is nonnegative ($n-t>0$) and the right-hand side is $0$, thus $y=0$. It holds that $y-\rho \in [0,1]$ so $\rho=0$, contrary to the non-equal strategy profile.
So publisher $1$ is not a winner.

Hence, publisher $1$ is a non-winner (i.e., all winners are high-index publishers), therefore,
\[
x_1^{eq}=1.
\]

We next claim that every winner must in fact be located exactly at $y$. Suppose some winning high-index publisher is located at $a<y$.

If there is at least one non-winner high-index publisher, then that player is located at $0$. Let $i$ be such a player,
and let
\[
m:=\frac{1+\sum_{j\neq 1,i}^n x_j^{eq}}{n-1},
\]
be the average of the other $n-1$ publishers. Since the current winners are all at the same point $a<y$, we have
\[
m=\frac{ny}{n-1}>y>a.
\]
If publisher $i$ deviates to $m$, then the new generated answer is exactly $m$, so publisher $i$ becomes the
unique \PRP{} winner. Her new utility is:
\[
u_i(m,x_{-i}^{eq})=1-\lambda m^2>0,
\]
because $m<1$ and $\lambda\le 1$. But currently $u_i(x^{eq})=0$, contradiction.

If, on the other hand, all high-index publishers are winners, then there are at least two winners at the same
point $a<y$. Any one of them can move slightly toward $y$, become the unique winner, and improve her
ranking payoff from $\frac{1}{n-1}$ to $1$, while changing her cost by an arbitrarily small amount.
This is again a profitable deviation, contradiction.

Therefore every winner must be exactly at $y$.

So $x^{eq}$ must be of the form:
\[
(1,\underbrace{y,\ldots,y}_{t\text{ times}},\underbrace{0,\ldots,0}_{n-t-1\text{ times}})
\]
for some $t\in\{1,\ldots,n-2\}$. The case $t=n-1$ would imply
\[
y=\frac{1+(n-1)y}{n}=1,
\]
which would make the equal strategies profile, already ruled out.

Now pick a losing high-index publisher $i$, so $x_i^{eq}=0$. Let
$
m:=\frac{1+t y}{n-1}
$ the average of the other $n-1$ publishers.
We get
\[
m=\frac{ny}{n-1}.
\]
In particular,
$
m>y
$
and
$
0<m<1.
$
If publisher $i$ deviates to $m$, then the generated response becomes exactly $m$, so publisher $i$ becomes
the unique \PRP{} winner. Hence,
\[
u_i(m,x_{-i}^{eq})=1-\lambda m^2>0=u_i(x^{eq}),
\]
a profitable deviation.

This contradiction shows that no PNE exists and completes the proof for the case of \K{} $=n$. We now turn to the \K{} $<n$ case. 

Let $\lambda > \frac{1}{3}$ and consider a generative publishers' game $G$ in which:
\[
n=3, \quad \text{\di}=1, \quad
x_1^0=1,\quad x_2^0=x_3^0=0,\quad x^*=1, 
\]
the generation function $y(x) = y_{1,z}(x)= \frac{1}{2}\sum_{j=1}^2 x_j$ (\K{} $=2<n$), and the \PRP{} attribution function $r=r^*$.

Assume by contradiction that $x^{eq}=(x_1^{eq},x_2^{eq},x_3^{eq})$ is a PNE of $G$.

We first claim that $x_1^{eq}=1$. Indeed, if publisher $1$ is not among the two publishers closest to $x^*=1$,
then $r_1(x^{eq})=0$, whereas by deviating to $\hat x_1=1$, publisher $1$ becomes one of the (at least) two closest
publishers to $x^*=1$, obtains an attribution probability of at least $\frac13$, and incurs no cost. If publisher $1$ is already
among the two publishers closest to $x^*=1$, then deviating to $\hat x_1=1$ does not decrease her attribution probability.
In all cases, the deviation strictly decreases publisher $1$'s cost unless already $x_1^{eq}=1$. Hence
$x_1^{eq}=1$.

We now consider the remaining two coordinates.

\begin{itemize}
    \item $x_2^{eq}=x_3^{eq}=a$. If $a<1$, we obtain
    \[
    y(x^{eq})=\frac{1+a}{2}.
    \]
    All three publishers are at the same distance from $y(x^{eq})$, so each receives an attribution probability of $\frac13$.
    Thus
    \[
    u_2(x^{eq})=\frac13-\lambda a^2.
    \]
    By deviating to $\hat x_2=a+\varepsilon$ for sufficiently small $\varepsilon>0$, publisher $2$ becomes the
    unique non-player-1 member of the generating set, and the attributed publishers are exactly publishers $1$
    and $2$. Hence
    \[
    u_2(\hat x_2,x^{eq}_{-2})=\frac12-\lambda(a+\varepsilon)^2>\frac13-\lambda a^2=u_2(x^{eq}),
    \]
    a contradiction. If $a=1$, then
    \[
    u_2(x^{eq})=\frac13-\lambda<0,
    \]
    whereas by deviating to $\hat x_2=0$, publisher $2$ gets utility $0$. Hence, this case is not a PNE either.

    \item $x_2^{eq}\neq x_3^{eq}$. Without loss of generality, assume
    \[
    x_2^{eq}>x_3^{eq}.
    \]
    Since $x_1^{eq}=1$, publisher $2$ is the unique second-closest publisher to $x^*=1$, and therefore
    \[
    y(x^{eq})=\frac{1+x_2^{eq}}{2}.
    \]
    We now compare distances to the generated response:
    \[
    \left|1-y(x^{eq})\right|=\left|x_2^{eq}-y(x^{eq})\right|=\frac{1-x_2^{eq}}{2},
    \]
    while
    \[
    \left|x_3^{eq}-y(x^{eq})\right|
    =\frac{1+x_2^{eq}-2x_3^{eq}}{2}
    >\frac{1-x_2^{eq}}{2},
    \]
    because $x_2^{eq}>x_3^{eq}$. Hence, publishers $1$ and $2$ are the unique closest publishers to the generated response, and
    \[
    u_2(x^{eq})=\frac12-\lambda(x_2^{eq})^2,\qquad
    u_3(x^{eq})=-\lambda(x_3^{eq})^2.
    \]

    If $x_3^{eq}>0$, then publisher $3$ can deviate to $\hat x_3=0$, therefore,
    \[
    u_3(\hat x_3,x_{-3}^{eq})=0>-\lambda(x_3^{eq})^2=u_3(x^{eq}),
    \]
    contradiction. Hence necessarily $x_3^{eq}=0$.

    We are therefore left with a candidate equilibrium of the form $(1,x_2^{eq},0)$ with $x_2^{eq}>0$. At this profile, publisher $2$ gets
    \[
    u_2(x^{eq})=\frac12-\lambda(x_2^{eq})^2.
    \]
    If publisher $2$ deviates to $\hat x_2=0$, the new profile is $(1,0,0)$. Regardless of how ties are broken, the generated response is
    \[
    y(\hat x_2,x_{-2}^{eq})=\frac12,
    \]
    because the second selected generator is located at $0$. All three publishers are then at a distance of $\frac12$ from the generated response, so each receives an attribution probability of $\frac13$. Hence, the deviation yields utility
    \[
    u_2(\hat x_2,x_{-2}^{eq})=\frac13.
    \]
    Therefore, if $x^{eq}$ were a PNE, we would have
    \[
    \frac12-\lambda(x_2^{eq})^2\ge \frac13,
    \]
    i.e.,
    \[
    \lambda(x_2^{eq})^2\le \frac16.
    \]
    In particular,
    \[
    (x_2^{eq})^2\le \frac{1}{6\lambda}<\frac{1}{2\lambda},
    \]
    so
    \[
    x_2^{eq}<\frac{1}{\sqrt{2\lambda}}.
    \]
    Choose any
    \[
    b\in\Bigl(x_2^{eq},\,\min\Bigl\{1,\frac{1}{\sqrt{2\lambda}}\Bigr\}\Bigr).
    \]
    Such a $b$ exists by the strict inequality above. If publisher $3$ deviates to $\hat x_3=b$, then publisher $3$ becomes the unique second-closest publisher to $x^*=1$, so the generating set becomes $\{1,3\}$ and
    \[
    y(x_1^{eq},x_2^{eq},\hat x_3)=\frac{1+b}{2}.
    \]
    Now publishers $1$ and $3$ are the unique closest publishers to this generated response, because $b>x_2^{eq}$. Hence
    \[
    u_3(x_1^{eq},x_2^{eq},\hat x_3)=\frac12-\lambda b^2>0.
    \]
    Since currently $u_3(x^{eq})=0$, this is a profitable deviation, contradiction.
\end{itemize}

All possible cases lead to a contradiction. Therefore, $G$ possesses no PNE.
\end{proof}

\begin{proof}[Proof of Theorem \ref{the:softmax-no-potential}]
Assume that there exists $\beta>0$ such that the
softmax attribution rule $\tilde r$ induces an exact potential game.
In particular, consider a generative publishers' game with the following
parameters:
\[
n=3,\quad \text{\di}=1, \quad y(x) = y_{1,z}(x):= \frac{1}{n}\sum_{j=1}^n x_j,
\]
an arbitrary $\lambda>0$, and arbitrary initial documents $x_1^0,x_2^0,x_3^0\in[0,1]$.
By assumption, this is a potential game and, thus, admits an exact potential function, which we denote by
$\Phi^\beta$.

Since the utilities are twice continuously differentiable, by Theorem 4.5 of \citet{monderer1996potential}, the existence of an
exact potential implies that
\[
\frac{\partial^2 u_i}{\partial x_j \partial x_i}
=
\frac{\partial^2 u_j}{\partial x_i \partial x_j}
\qquad \text{for every } i, j \in N.
\]

Now fix $s\in(0,1)$ and choose $t>0$ such that:
\[
0<t<\min\{s,1-s\}.
\]
For the profile $x=(s+t,s,s-t)$, denote
\[
g:=e^{-\beta t^2},\qquad D:=1+2g.
\]
Note that $y(x)=s$, so we get
\[
d_1^*=t^2,\qquad d_2^*=0,\qquad d_3^*=t^2,
\]
and hence
\[
\tilde r_1(x)=\tilde r_3(x)=\frac{g}{D},\qquad \tilde r_2(x)=\frac{1}{D}.
\]
Since the cost term $d_i^0(x_i)$ depends only on $x_i$, its mixed cross-partials vanish, and therefore:
\[
\frac{\partial^2 u_1}{\partial x_2\partial x_1}
-
\frac{\partial^2 u_2}{\partial x_1\partial x_2}
=
\frac{\partial^2 \tilde r_1}{\partial x_2\partial x_1}
-
\frac{\partial^2 \tilde r_2}{\partial x_1\partial x_2}.
\]

Now
\[
d_1^*=\left(\frac{2x_1-x_2-x_3}{3}\right)^2,\quad
d_2^*=\left(\frac{2x_2-x_1-x_3}{3}\right)^2,
\]
\[
d_3^*=\left(\frac{2x_3-x_1-x_2}{3}\right)^2.
\]
Evaluating at $x=(s+t,s,s-t)$ gives
\[
\partial_{x_1} d_1^*=\frac{4t}{3},\qquad
\partial_{x_2} d_1^*=-\frac{2t}{3},\qquad
\partial_{x_1x_2}^2 d_1^*=-\frac49,
\]
\[
\partial_{x_1} d_2^*=0,\qquad
\partial_{x_2} d_2^*=0,\qquad
\partial_{x_1x_2}^2 d_2^*=-\frac49,
\]
\[
\partial_{x_1} d_3^*=\frac{2t}{3},\qquad
\partial_{x_2} d_3^*=\frac{2t}{3},\qquad
\partial_{x_1x_2}^2 d_3^*=\frac29.
\]

For the softmax rule,
\[
\partial_{x_j}\tilde r_i
=
-\beta \tilde r_i
\left(
\partial_{x_j}d_i^*
-
\sum_{m=1}^3 \tilde r_m\,\partial_{x_j}d_m^*
\right).
\]
Thus, if we set
\[
B_1:=\sum_{m=1}^3 \tilde r_m\,\partial_{x_1}d_m^*,
\qquad
B_2:=\sum_{m=1}^3 \tilde r_m\,\partial_{x_2}d_m^*,
\]
then at $x=(s+t,s,s-t)$,
\[
B_1=\frac{g}{D}\cdot \frac{4t}{3}+\frac{1}{D}\cdot 0+\frac{g}{D}\cdot \frac{2t}{3}
=\frac{2gt}{D},
\]
\[
B_2=\frac{g}{D}\cdot\left(-\frac{2t}{3}\right)+\frac{1}{D}\cdot 0+\frac{g}{D}\cdot \frac{2t}{3}
=0.
\]
Hence,
\[
\partial_{x_2}\tilde r_1
=
-\beta\frac{g}{D}\left(-\frac{2t}{3}-0\right)
=
\frac{2\beta gt}{3D},
\]
\[
\partial_{x_1}\tilde r_2
=
-\beta\frac{1}{D}\left(0-\frac{2gt}{D}\right)
=
\frac{2\beta gt}{D^2}.
\]

Differentiating once more,
\[
\partial_{x_2x_1}^2 \tilde r_1
=
-\beta\left[
\partial_{x_2}\tilde r_1\bigl(\partial_{x_1}d_1^*-B_1\bigr)
+
\tilde r_1\bigl(\partial_{x_2x_1}^2 d_1^*-\partial_{x_2}B_1\bigr)
\right],
\]
\[
\partial_{x_1x_2}^2 \tilde r_2
=
-\beta\left[
\partial_{x_1}\tilde r_2\bigl(\partial_{x_2}d_2^*-B_2\bigr)
+
\tilde r_2\bigl(\partial_{x_1x_2}^2 d_2^*-\partial_{x_1}B_2\bigr)
\right].
\]
Using
\[
\partial_{x_2}B_1=\partial_{x_1}B_2
=
\frac{4\beta g t^2}{9D}-\frac{2(2+g)}{9D},
\]
one obtains:
\[
\partial_{x_2x_1}^2 \tilde r_1
=
\frac{2\beta g(3g-4\beta t^2)}{9D^2},
\qquad
\partial_{x_1x_2}^2 \tilde r_2
=
\frac{2\beta g(3+2\beta t^2)}{9D^2}.
\]
Therefore,
\begin{align*}
\left(
\frac{\partial^2 u_1}{\partial x_2\partial x_1}
-
\frac{\partial^2 u_2}{\partial x_1\partial x_2}
\right)\Bigg|_{x=(s+t,s,s-t)}
\\
=
-\frac{2\beta g}{3(1+2g)^2}\bigl(1+2\beta t^2-g\bigr).
\end{align*}

Since $\beta>0$ and $t>0$, we have
$
0<g(\beta)=e^{-\beta t^2}<1,
$
and therefore
$
1+2\beta t^2-g(\beta)>0.
$
It follows that
\[
-\frac{2\beta\, g(\beta)}{3(1+2g(\beta))^2}
\Bigl(1+2\beta t^2-g(\beta)\Bigr)\neq 0.
\]
Hence,
\[
\frac{\partial^2 u_1}{\partial x_2\partial x_1}
\neq
\frac{\partial^2 u_2}{\partial x_1\partial x_2}
\]
at the profile $x=(s+t,s,s-t)$, contradicting a necessary condition for the
existence of an exact potential.
We conclude that no exact potential can exist when $\beta>0$. Therefore, if the softmax rule induces an exact potential game, it must be that $\beta=0$.

Finally, when $\beta=0$, the attribution function becomes uniform across publishers:
\[
\tilde r_i^{0}(x)
=
\frac{1}{\sum_{j=1}^n 1}
=
\frac1n
\qquad \text{for all } i\in N.
\]
\end{proof}

\begin{proof}[Proof of Observation \ref{obs:3}]
We divide the proof into two parts: the \K{} $<n$ case and the \K{} $=n$ case. We start with the \K{} $=n$ case.

Fix any
$
t\in\left(0,\frac12\right)
$
and 
$
\lambda>0
$
such that
$
\lambda t^2<\frac12
$
and 
$
\lambda(1-t)^2<\frac12.
$
Consider a generative publishers' game in which:
\[
n=3, \quad \text{\di}=1, \quad x_1^0=1,\quad x_2^0=x_3^0=0, \quad r=\tilde r,
\]
and the generation function $y(x) = y_{1,z}(x)= \frac{1}{n}\sum_{j=1}^n x_j$ (\K{} $=n$).

Consider the following strategy profiles:
\[
A=(t,0,0),\;
B=(1,0,0),\; 
C=(1,0,t),\;
D=(t,0,t).
\]
These profiles form a
cyclic better-response dynamics, with deviators
$
A \xrightarrow{1} B \xrightarrow{3} C \xrightarrow{1} D \xrightarrow{3} A.
$

Let
\[
\Delta_{A\to B}(\beta):=u_1(B)-u_1(A),\quad
\Delta_{B\to C}(\beta):=u_3(C)-u_3(B),
\]
\[
\Delta_{C\to D}(\beta):=u_1(D)-u_1(C),\quad
\]
\[
\Delta_{D\to A}(\beta):=u_3(A)-u_3(D).
\]

For \(A=(t,0,0)\), we have \(y(A)=t/3\), so
\[
d_1^*(A)=\left(t-\frac t3\right)^2=\frac{4t^2}{9},
\]
\[
d_2^*(A)=d_3^*(A)=\left(0-\frac t3\right)^2=\frac{t^2}{9}.
\]
Hence,
\[
\tilde r_1(A)
=
\frac{e^{-4\beta t^2/9}}{e^{-4\beta t^2/9}+2e^{-\beta t^2/9}}
=
\frac{1}{1+2e^{\beta t^2/3}}.
\]

For \(B=(1,0,0)\), we have \(y(B)=1/3\), so:
\[
d_1^*(B)=\left(1-\frac13\right)^2=\frac49,
\]
\[
d_2^*(B)=d_3^*(B)=\left(0-\frac13\right)^2=\frac19.
\]
Hence,
\[
\tilde r_1(B)
=
\frac{e^{-4\beta/9}}{e^{-4\beta/9}+2e^{-\beta/9}}
=
\frac{1}{1+2e^{\beta/3}}.
\]
Using the above we get
\begin{align*}
\Delta_{A\to B}(\beta)
=
\frac{1}{1+2e^{\beta/3}}
-
\frac{1}{1+2e^{\beta t^2/3}}
+
\lambda(1-t)^2
\\
\xrightarrow[\beta\to\infty]{}\;
\lambda(1-t)^2>0.
\end{align*}

At $B$, publisher $3$ has softmax weight
\[
\tilde r_3(B)=\frac{1}{2+e^{-\beta/3}}
\;\xrightarrow[\beta\to\infty]{}\;
\frac12.
\]
At $C=(1,0,t)$, since
\[
d_3^*(C)<d_2^*(C)<d_1^*(C)
\qquad\text{for every }t\in\left(0,\frac12\right),
\]
publisher $3$ is the unique closest publisher to the generated response, and therefore
\[
\tilde r_3(C)\xrightarrow[\beta\to\infty]{}1.
\]
Thus,
\[
\Delta_{B\to C}(\beta)
\;\xrightarrow[\beta\to\infty]{}\;
1-\lambda t^2-\frac12
=
\frac12-\lambda t^2
>0.
\]

At $C$, publisher $1$ is not a closest publisher, so
\[
\tilde r_1(C)\xrightarrow[\beta\to\infty]{}0.
\]
At $D=(t,0,t)$, publishers $1$ and $3$ are tied for being closest to the generated response, while publisher
$2$ is farther, hence
\[
\tilde r_1(D)=\frac{1}{2+e^{-\beta t^2/3}}
\;\xrightarrow[\beta\to\infty]{}\;
\frac12.
\]
Therefore,
\[
\Delta_{C\to D}(\beta)
\;\xrightarrow[\beta\to\infty]{}\;
\frac12-\lambda(1-t)^2
>0.
\]

Finally, at both $A$ and $D$, publisher $3$ has the same attribution probability
\[
\tilde r_3(A)=\tilde r_3(D)=\frac{1}{2+e^{-\beta t^2/3}},
\]
while her movement cost equals $0$ at $A$ and $\lambda t^2$ at $D$. Hence,
\[
\Delta_{D\to A}(\beta)=\lambda t^2>0
\qquad\text{for every }\beta>0.
\]

Since the first three utility gaps converge to strictly positive limits, there exists \(\beta_0>0\) such that
all four deviations are profitable for every \(\beta\ge\beta_0\). Fix any such \(\beta\), and define:
{\small
\[
\epsilon
:=
\frac12
\min\Bigl\{
\Delta_{A\to B}(\beta),
\Delta_{B\to C}(\beta),
\Delta_{C\to D}(\beta),
\Delta_{D\to A}(\beta)
\Bigr\}
>0.
\]
}
Then every step in the cycle is an \(\epsilon\)-better response. Repeating the cycle yields an infinite
\(\epsilon\)-better-response dynamic.
This completes the proof for the case of \K{} $=n$. We now turn to the \K{} $<n$ case.

Fix $\lambda > \frac{1}{6}$, and choose $b,c \in (0,1)$ such that:
$$
0 < b < c < 1, \qquad \lambda b^2 < \frac{1}{6}, \qquad \frac{1}{6} < \lambda c^2 < \frac{1}{2}.
$$
Consider a generative publishers' game with:
$$
n=3, \quad \text{\di}=1, \quad x^*=1, \quad x_1^0=1, \quad x_2^0=x_3^0=0,
$$
the generation function $y(x) = y_{1,z}(x):= \frac{1}{2}\sum_{j=1}^2 x_j$ ($\text{\K}=2$), and the softmax attribution function $r=\tilde r_\beta$.
Then there exists $\beta_0>0$ such that for every $\beta \geq \beta_0$, the following four profiles form a cyclic better-response path:
$$
A=(1,0,0), \; B=(1,b,0), \; C=(1,b,c), \; D=(1,0,c),
$$
with deviators:
$$
A \xrightarrow{\,2\,} B \xrightarrow{\,3\,} C \xrightarrow{\,2\,} D \xrightarrow{\,3\,} A.
$$
Hence, there exists an infinite sequence of profitable deviations, and in particular, for sufficiently small $\varepsilon>0$, $\varepsilon$-better-response dynamics may not converge.

We analyze the four profiles one by one.
At the profile
$
A=(1,0,0),
$
publisher $1$ is certainly among the two publishers closest to $x^*=1$, and the second selected publisher is one of $2,3$. In either case, the generated response is:
$$
y(A)=\frac{1+0}{2}=\frac{1}{2}.
$$
Therefore, all three publishers are at the same distance from the generated response,
$$
d_1^*(A)=d_2^*(A)=d_3^*(A)=\left(\frac{1}{2}\right)^2=\frac{1}{4}.
$$
Hence, under the softmax attribution rule, we have
$$
\tilde r_1(A)=\tilde r_2(A)=\tilde r_3(A)=\frac{1}{3},
$$
for every $\beta>0$. In particular,
$
u_2(A)=u_3(A)=\frac{1}{3}.
$

Next, consider
$
B=(1,b,0).
$
Since $0<b<1$, the two publishers closest to $x^*=1$ are $1$ and $2$; therefore
$$
y(B)=\frac{1+b}{2}.
$$
The distances to the generated response are:
$$
d_1^*(B)=d_2^*(B)=\left(\frac{1-b}{2}\right)^2,
\qquad
d_3^*(B)=\left(\frac{1+b}{2}\right)^2.
$$
Since $d_3^*(B)>d_1^*(B)=d_2^*(B)$, it follows that
$$
\lim_{\beta\to\infty} \tilde r_2(B)=\frac{1}{2},
\qquad
\lim_{\beta\to\infty} \tilde r_3(B)=0.
$$
Hence,
$$
\lim_{\beta\to\infty} \bigl(u_2(B)-u_2(A)\bigr)
=
\left(\frac{1}{2}-\lambda b^2\right)-\frac{1}{3}
=
\frac{1}{6}-\lambda b^2
>0,
$$
that is, for sufficiently large $\beta$ publisher $2$ has a profitable deviation from profile $A$ to profile $B$.

Now consider
$
C=(1,b,c).
$
Since $0<b<c<1$, the two publishers closest to $x^*=1$ are $1$ and $3$, and thus:
$$
y(C)=\frac{1+c}{2}.
$$
The distances to the generated response are
$$
d_1^*(C)=d_3^*(C)=\left(\frac{1-c}{2}\right)^2, \quad
d_2^*(C)=\left(\frac{1+c}{2}-b\right)^2.
$$
Because $c>b$, we have
$$
\frac{1+c}{2}-b > \frac{1-c}{2},
$$
and therefore
$
d_2^*(C) > d_1^*(C)=d_3^*(C).
$
Thus, it follows that
$$
\lim_{\beta\to\infty} \tilde r_3(C)=\frac{1}{2},
\qquad
\lim_{\beta\to\infty} \tilde r_3(B)=0,
$$
and hence
$$
\lim_{\beta\to\infty} \bigl(u_3(C)-u_3(B)\bigr)
=
\left(\frac{1}{2}-\lambda c^2\right)-0
=
\frac{1}{2}-\lambda c^2
>0,
$$
that is, for sufficiently large $\beta$ publisher $3$ has a profitable deviation from profile $B$ to profile $C$.

Next, consider
$
D=(1,0,c).
$
Again the two publishers closest to $x^*=1$ are $1$ and $3$, so
$$
y(D)=\frac{1+c}{2}.
$$
The distances to the generated response are
$$
d_1^*(D)=d_3^*(D)=\left(\frac{1-c}{2}\right)^2,
\qquad
d_2^*(D)=\left(\frac{1+c}{2}\right)^2.
$$
Hence,
$$
\lim_{\beta\to\infty} \tilde r_2(C)=0,
\qquad
\lim_{\beta\to\infty} \tilde r_2(D)=0,
$$
and therefore
$$
\lim_{\beta\to\infty} \bigl(u_2(D)-u_2(C)\bigr)
=
0-(-\lambda b^2)
=
\lambda b^2
>0,
$$
that is, for sufficiently large $\beta$ publisher $2$ has a profitable deviation from profile $C$ to profile $D$.

Finally, comparing $D$ and $A$, we have already seen that
$
u_3(A)=\frac{1}{3}
$
for every $\beta>0$, while:
$$
\lim_{\beta\to\infty} u_3(D)
=
\frac{1}{2}-\lambda c^2.
$$
Thus, we get:
$$
\lim_{\beta\to\infty} \bigl(u_3(A)-u_3(D)\bigr)
=
\frac{1}{3}-\left(\frac{1}{2}-\lambda c^2\right)
=
\lambda c^2-\frac{1}{6}
>0,
$$
that is, for sufficiently large $\beta$ publisher $3$ has a profitable deviation from profile $D$ to profile $A$.

We have shown that the four utility differences
$$
u_2(B)-u_2(A), \:
u_3(C)-u_3(B), \:
u_2(D)-u_2(C), \:
u_3(A)-u_3(D)
$$
all converge, as $\beta\to\infty$, to strictly positive quantities. Therefore, there exists $\beta_0>0$ such that for every $\beta \geq \beta_0$, all four deviations remain profitable. In other words,
$
A \xrightarrow{\,2\,} B \xrightarrow{\,3\,} C \xrightarrow{\,2\,} D \xrightarrow{\,3\,} A
$
is a cyclic better-response path.
Hence, there exists an infinite sequence of profitable deviations, and in particular, for sufficiently small $\varepsilon>0$, $\varepsilon$-better-response dynamics may not converge.
\end{proof}

\begin{proof}[Proof of Theorem \ref{teo:linear_not_converge_for_K-smaller_than_n}]
Assume by contradiction that the linear attribution function \(\hat r_a\) induces a potential game for the \K{} $<n$ case. Thus, it induces a potential game for a generative publishers' game with: 
\[
\text{\di}=1,
\qquad x^*=0, \qquad \lambda>0,
\]
the generation function $y(x) = y_{1,z}(x)= \frac{1}{k}\sum_{j=1}^k x_j$ where \K{} $<n$, and arbitrary initial documents.

Fix parameters \(\varepsilon,\rho\)
such that
\[
0<\varepsilon<1,\qquad 0<\rho<1-\varepsilon.
\]
Now consider the open set:
{\small
\[
U:=\Bigl\{x\in [0,1]^n:\ x_1,\ldots,x_\text{\K}\in (0,\varepsilon),\ x_{\text{\K}+1},\ldots,x_n\in (\varepsilon+\rho,1)\Bigr\}.
\]}
Since \(x^*=0\), for every \(x\in U\) the \K{} closest publishers to the user's question $x^*$ are exactly publishers \(1,\ldots,\)\K. Hence, on \(U\) it holds that:
\[
y(x)=\frac{1}{\text{\K}}\sum_{j=1}^\text{\K} x_j.
\]
In particular, on \(U\), the function \(y\) is twice continuously differentiable, it does not depend on
\(x_{\text{\K}+1}\), and
\[
\frac{\partial y}{\partial x_1}(x)=\frac{1}{\text{\K}}\qquad \forall x\in U.
\]

By assumption, this game is a potential game. Since the utilities are twice continuously differentiable on
\(U\), Theorem 4.5 of \citet{monderer1996potential} implies that for every \(x\in U\) we have:
\[
\frac{\partial^2 u_1}{\partial x_{\text{\K}+1}\partial x_1}(x)
=
\frac{\partial^2 u_{\text{\K}+1}}{\partial x_1\partial x_{\text{\K}+1}}(x).
\]
We now compute these two cross-derivatives.

Because \(y\) is independent of \(x_{\text{\K}+1}\) on \(U\), for every \(j\neq \text{\K}+1\) the term:
$
d_j^*(x)=(x_j-y(x))^2
$
is independent of \(x_{\text{\K}+1}\). Therefore,
\[
\frac{\partial^2 d_j^*}{\partial x_{\text{\K}+1}\partial x_1}(x)=0
\qquad \forall j\neq \text{\K}+1.
\]

For publisher \(\text{\K}+1\),
$
d_{\text{\K}+1}^*(x)=(x_{\text{\K}+1}-y(x))^2.
$
Differentiating first with respect to \(x_1\), we obtain
\[
\frac{\partial d_{\text{\K}+1}^*}{\partial x_1}(x)
=
-2(x_{\text{\K}+1}-y(x))\frac{\partial y}{\partial x_1}(x)
=
-\frac{2}{\text{\K}}(x_{\text{\K}+1}-y(x)).
\]
Differentiating now with respect to \(x_{\text{\K}+1}\), and using again that \(y\) does not depend on \(x_{\text{\K}+1}\),
gives
\[
\frac{\partial^2 d_{\text{\K}+1}^*}{\partial x_{\text{\K}+1}\partial x_1}(x)
=
-\frac{2}{\text{\K}}.
\]

Recall that
\[
\nu_i(x)=\frac{1}{n-1}\sum_{j\neq i} d_j^*(x)-d_i^*(x).
\]
Hence,
\[
\frac{\partial^2 \nu_1}{\partial x_{\text{\K}+1}\partial x_1}(x)
=
\frac{1}{n-1}\frac{\partial^2 d_{\text{\K}+1}^*}{\partial x_{\text{\K}+1}\partial x_1}(x)
=
-\frac{2}{\text{\K}(n-1)},
\]
whereas
\[
\frac{\partial^2 \nu_{\text{\K}+1}}{\partial x_{\text{\K}+1}\partial x_1}(x)
=
-\frac{\partial^2 d_{\text{\K}+1}^*}{\partial x_{\text{\K}+1}\partial x_1}(x)
=
\frac{2}{\text{\K}}.
\]

The cost terms depend only on a player’s own action, so their mixed partial derivatives vanish. Therefore,
\[
\frac{\partial^2 u_1}{\partial x_{\text{\K}+1}\partial x_1}(x)
=
m\frac{\partial^2 \nu_1}{\partial x_{\text{\K}+1}\partial x_1}(x)
=
-\frac{2m}{\text{\K}(n-1)},
\]
\[
\frac{\partial^2 u_{\text{\K}+1}}{\partial x_1\partial x_{\text{\K}+1}}(x)
=
m\frac{\partial^2 \nu_{\text{\K}+1}}{\partial x_{\text{\K}+1}\partial x_1}(x)
=
\frac{2m}{\text{\K}}.
\]

Since \(m>0\), \K{} $\ge1$, $n>$ \K{} these two expressions are not equal. This contradicts the necessary cross-partial equality
for potential games.

Therefore, the game is not a potential game. Since this game belongs to the above class, we conclude
that the linear attribution function does not induce a potential game for generative publishers' games.

\end{proof}

An interesting question that arises is whether better-response dynamics in generative publishers' games with context size of \K{} $<n$ converge, even though such games are not potential games. As we showed in Section~\ref{sec:empirical_results}, that is not the case. 
To conclude, generative publishers' games with the linear attribution function converge under any $\epsilon$-better-response dynamic in the case of \K{} $=n$; moreover, a PNE exists.

\begin{proof}[Proof of Theorem \ref{thm:linear-kn-potential}]
Let \(i\in N\), \(x_{-i}\in X_{-i}\), \(x_i,x_i'\in X_i\), and denote
$
x:=(x_i,x_{-i}),
$$
x':=(x_i',x_{-i}).
$
Our goal is to show that:
\[
\Psi_{}(x')-\Psi_{}(x)
=
u_i(x')-u_i(x).
\]
Recall that
$
u_i(x)
=
m\nu_i(x)+\frac{1}{n}
-
\lambda d_i^0(x_i).
$

Let
\[
S_{-i}:=\sum_{j\ne i}x_j .
\]
For fixed \(x_{-i}\), we can write the generated response as
\[
y_{\alpha,z}(x)
=
\frac{\alpha}{n}x_i
+
\frac{\alpha}{n}S_{-i}
+
(1-\alpha)z .
\]
For brevity, define
\[
\eta:=\frac{\alpha}{n},
\qquad
h_{-i}:=\eta S_{-i}+(1-\alpha)z.
\]
Then
$
y_{\alpha,z}(x)=\eta x_i+h_{-i}.
$

By definition of the linear relative relevance and the notations above, we obtain
\[
\nu_i(x)
=
\frac{1}{\text{\di}(n-1)}
\sum_{j\ne i}\|x_j-\eta x_i-h_{-i}\|_2^2
-
\frac{1}{\text{\di}}\|(1-\eta)x_i-h_{-i}\|_2^2 .
\]
We now isolate the terms that depend on \(x_i\). Expanding the first
term gives
\begin{align*}
&\frac{1}{\text{\di}(n-1)}
\sum_{j\ne i}\|x_j-\eta x_i-h_{-i}\|_2^2 \\
&=
C_1(x_{-i})
-
\frac{2\eta}{\text{\di}(n-1)}
x_i^\top
\sum_{j\ne i}(x_j-h_{-i})
+
\frac{\eta^2}{\text{\di}}\|x_i\|_2^2,
\end{align*}
where \(C_1(x_{-i})\) is independent of \(x_i\). Similarly for the second term:
{\small
\[
\frac{1}{\text{\di}}\|(1-\eta)x_i-h_{-i}\|_2^2
=
\frac{(1-\eta)^2}{\text{\di}}\|x_i\|_2^2
-
\frac{2(1-\eta)}{\text{\di}}x_i^\top h_{-i}
+
C_2(x_{-i}),
\]
}
where \(C_2(x_{-i})\) is independent of \(x_i\). Hence,
\begin{align*}
\nu_i(x)
&=
C_i(x_{-i})
-
\frac{1-2\eta}{\text{\di}}\|x_i\|_2^2 \\
&-
\frac{2\eta}{\text{\di}(n-1)}
x_i^\top
\sum_{j\ne i}(x_j-h_{-i})
+
\frac{2(1-\eta)}{\text{\di}}x_i^\top h_{-i},
\end{align*}
for some term \(C_i(x_{-i})\) that is independent of \(x_i\).

Since
\[
\sum_{j\ne i}(x_j-h_{-i})
=
S_{-i}-(n-1)h_{-i},
\]
we get:
\[
\nu_i(x)
=
C_i(x_{-i})
-
\frac{1-2\eta}{\text{\di}}\|x_i\|_2^2
-
\frac{2\eta}{\text{\di}(n-1)}x_i^\top S_{-i}
+
\frac{2}{\text{\di}}x_i^\top h_{-i}.
\]
Substituting
\[
\eta=\frac{\alpha}{n},
\qquad
h_{-i}=\frac{\alpha}{n}S_{-i}+(1-\alpha)z,
\]
yields
\begin{equation}\label{v_i}
\begin{aligned}
\nu_i(x)
&=
C_i(x_{-i})
-
\frac{1-\frac{2\alpha}{n}}{\text{\di}}\|x_i\|_2^2 \\
&+
\frac{2\alpha(n-2)}{kn(n-1)}x_i^\top S_{-i} 
+
\frac{2(1-\alpha)}{\text{\di}}x_i^\top z .
\end{aligned}
\end{equation}

Now define
\[
Q_i(x_i;x_{-i})
:=
(1-\alpha)d(x_i,z)
+
\frac{\alpha(n-2)}{n(n-1)}
\sum_{j\ne i}d(x_i,x_j).
\]
Expanding \(Q_i\), we obtain
\begin{align*}
Q_i(x_i;x_{-i})
&=
\frac{1-\alpha}{\text{\di}}
\left(
\|x_i\|_2^2-2x_i^\top z+\|z\|_2^2
\right) \\
&+
\frac{\alpha(n-2)}{kn(n-1)}
\sum_{j\ne i}
\left(
\|x_i\|_2^2-2x_i^\top x_j+\|x_j\|_2^2
\right).
\end{align*}
Collecting only the terms that depend on \(x_i\), we get
\begin{equation}\label{Q_i}
\begin{aligned}
Q_i(x_i;x_{-i})
&=
C_i'(x_{-i})
+
\frac{1-\frac{2\alpha}{n}}{\text{\di}}\|x_i\|_2^2 \\
&-
\frac{2\alpha(n-2)}{kn(n-1)}x_i^\top S_{-i}
-
\frac{2(1-\alpha)}{\text{\di}}x_i^\top z,
\end{aligned}
\end{equation}
where \(C_i'(x_{-i})\) is independent of \(x_i\). Comparing (\ref{v_i}) with (\ref{Q_i}), we obtain
\[
\nu_i(x)
=
C_i''(x_{-i})
-
Q_i(x_i;x_{-i}),
\]
for some term \(C_i''(x_{-i})\) independent of \(x_i\). Therefore,
\[
\nu_i(x')-\nu_i(x)
=
-
\left[
Q_i(x_i';x_{-i})-Q_i(x_i;x_{-i})
\right].
\]
That is,
\begin{align*}
\nu_i(x')-\nu_i(x)
&=
-(1-\alpha)
\left[
d(x_i',z)-d(x_i,z)
\right] \\
&-
\frac{\alpha(n-2)}{n(n-1)}
\sum_{j\ne i}
\left[
d(x_i',x_j)-d(x_i,x_j)
\right].
\end{align*}

Now consider the proposed potential.
Since only player \(i\) changes her action, all terms that do not involve
\(x_i\) cancel in the following difference:
\[
\begin{aligned}
\Psi_{}(x')-\Psi_{}(x)
&=
-(1-\alpha)m
\left[
d(x_i',z)-d(x_i,z)
\right] \\
&-
\frac{\alpha m(n-2)}{n(n-1)}
\sum_{j\ne i}
\left[
d(x_i',x_j)-d(x_i,x_j)
\right] \\
&-
\lambda_i
\left[
d_i^0(x_i')-d_i^0(x_i)
\right].
\end{aligned}
\]
Using the identity for \(\nu_i(x')-\nu_i(x)\) derived above, this becomes
\[
\Psi_{}(x')-\Psi_{}(x)
=
m\left[
\nu_i(x')-\nu_i(x)
\right]
-
\lambda_i
\left[
d_i^0(x_i')-d_i^0(x_i)
\right].
\]
Finally, since
\[
u_i(x)=m\nu_i(x)+\frac{1}{n}-\lambda_i d_i^0(x_i),
\]
we conclude that
\[
\Psi_{}(x')-\Psi_{}(x)
=
u_i(x')-u_i(x).
\]
Thus, \(\Psi_{}\) is an exact potential function.

The existence of an \(\varepsilon\)-PNE and the convergence of every
\(\varepsilon\)-better-response dynamics follow from Theorem 1.
\end{proof}

\section{Simulation Algorithm}\label{apn:sum_algorithm}

We use the discrete better response dynamics simulation described in detail by \citet{Omer_Madmon_2025}. 
The simulation starts from an initial strategy profile (i.e., initial documents) and proceeds in rounds. 
In each round, we identify the publishers who can improve their utility by more than a threshold \(\varepsilon\) through a single document modification, where the modification is restricted to a predefined finite set of directions \(D\) and step sizes \(S\). 
If at least one such publisher exists, one improving publisher is selected uniformly at random. 
Among all feasible modifications for that publisher, the simulation selects a utility-maximizing modification; if several modifications are tied for maximal utility, one of them is chosen uniformly at random. 
The selected publisher's document is then updated accordingly, while all other publishers' documents remain unchanged.

The simulation stops early if no publisher can improve her utility by more than \(\varepsilon\) using any feasible modification (i.e., the simulation converges). 
In this case, the utilities at the final strategy profile are returned. 
If the simulation reaches the maximal number of rounds \(T\), we return the averages of the utilities over the final \(M\) rounds.

\begin{figure}[t]
    \centering

\begin{subfigure}[t]{\columnwidth}
    \centering

    \begin{subfigure}[t]{0.325\columnwidth}
        \centering
        \includegraphics[width=\linewidth]{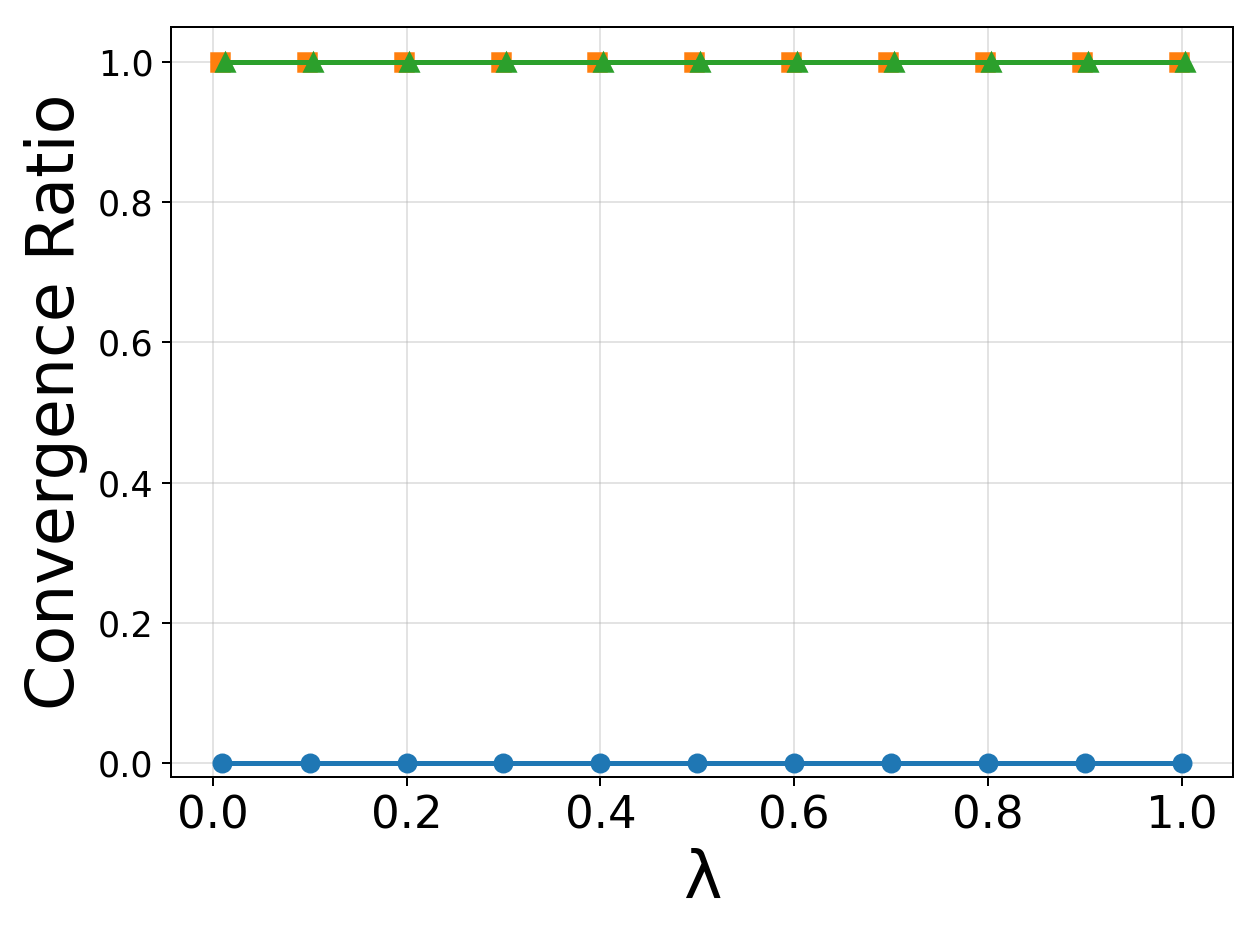}
        \caption{Convergence ratio}
        \label{fig:faith_effect_convergence_ratio}
    \end{subfigure}
    \hspace{0.06\columnwidth}
    \begin{subfigure}[t]{0.325\columnwidth}
        \centering
        \includegraphics[width=\linewidth]{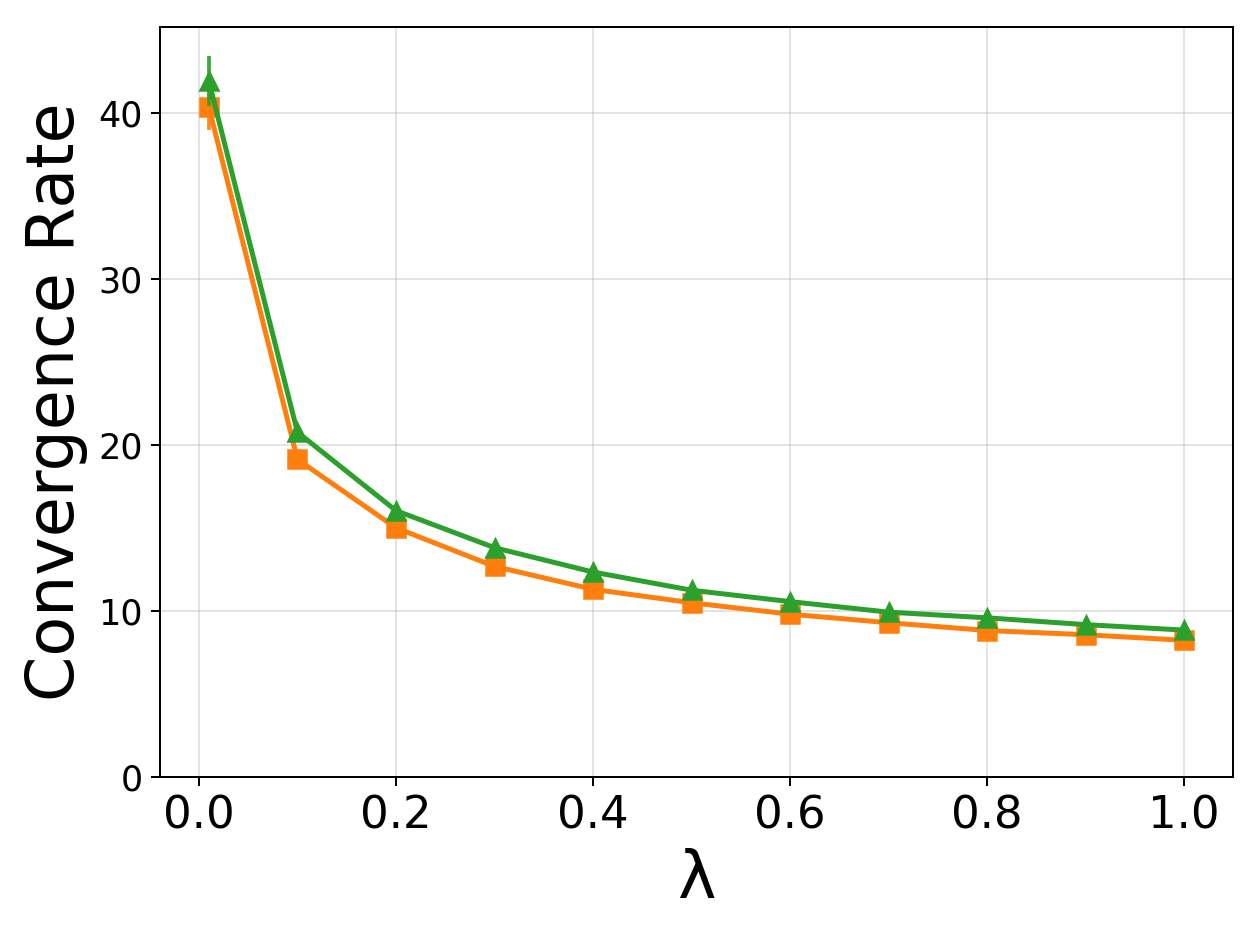}
        
        \caption{Convergence rate}\label{fig:faith_effect_convergence_rate}
    \end{subfigure}

\end{subfigure}

\begin{subfigure}[t]{\columnwidth}
    \centering

    \begin{subfigure}[t]{0.325\columnwidth}
        \centering
        \includegraphics[width=\linewidth]{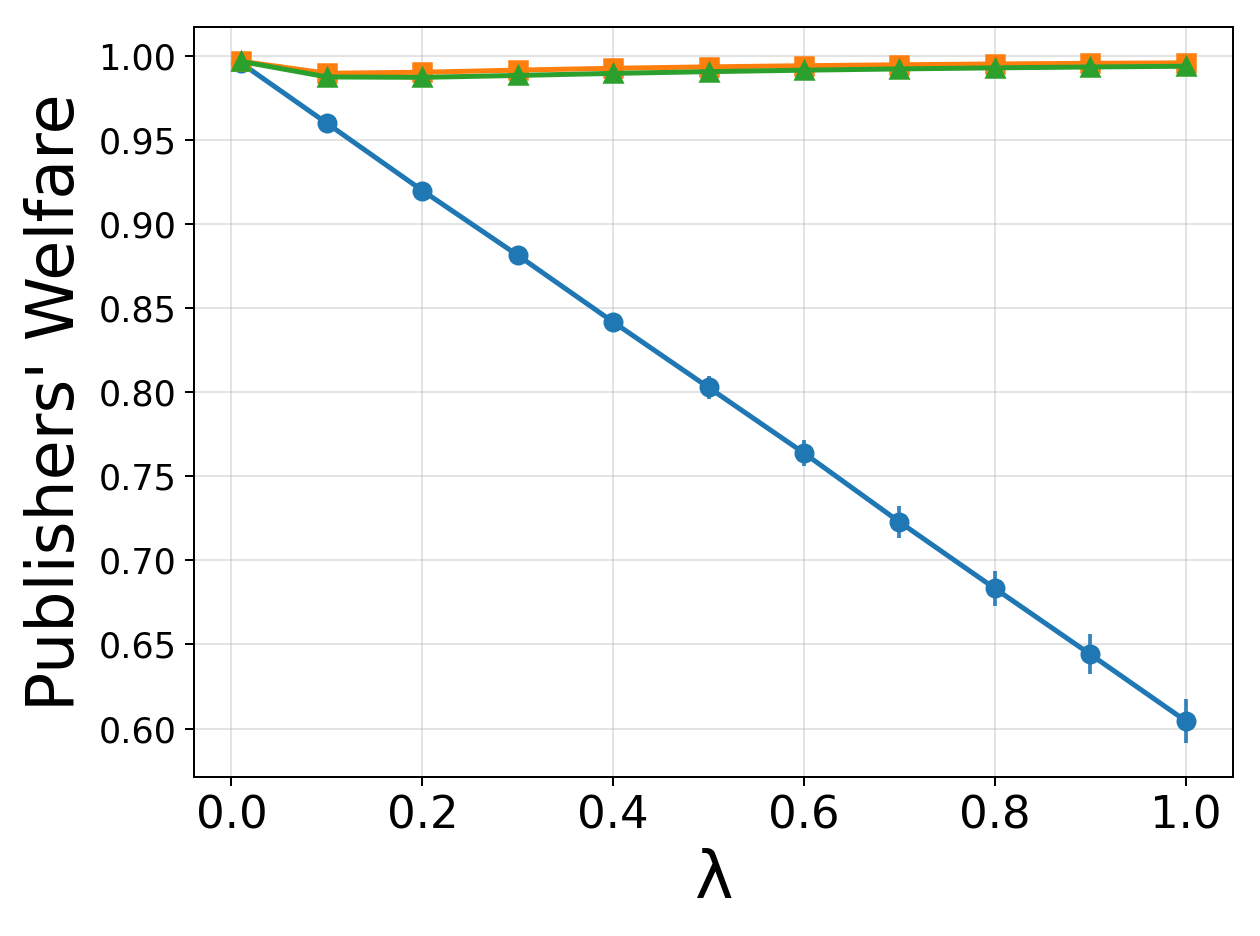}
        
        \caption{PW}\label{fig:faith_effect_PW}
    \end{subfigure}
    \hfill
    \begin{subfigure}[t]{0.325\columnwidth}
        \centering
        \includegraphics[width=\linewidth]{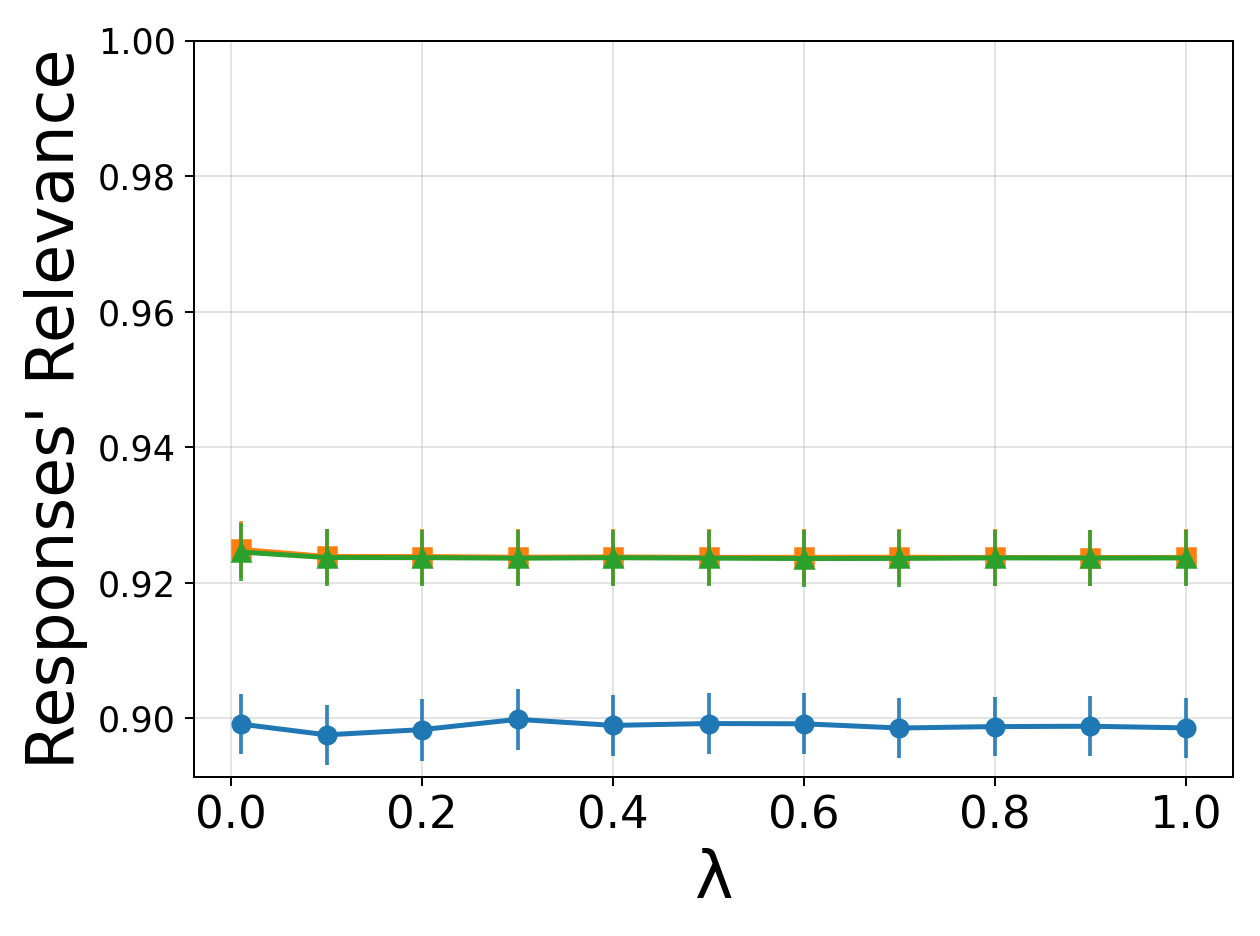}
        
        \caption{RR}\label{fig:faith_effect_RR}
    \end{subfigure}
    \hfill
    \begin{subfigure}[t]{0.325\columnwidth} 
        \centering
        \includegraphics[width=\linewidth]{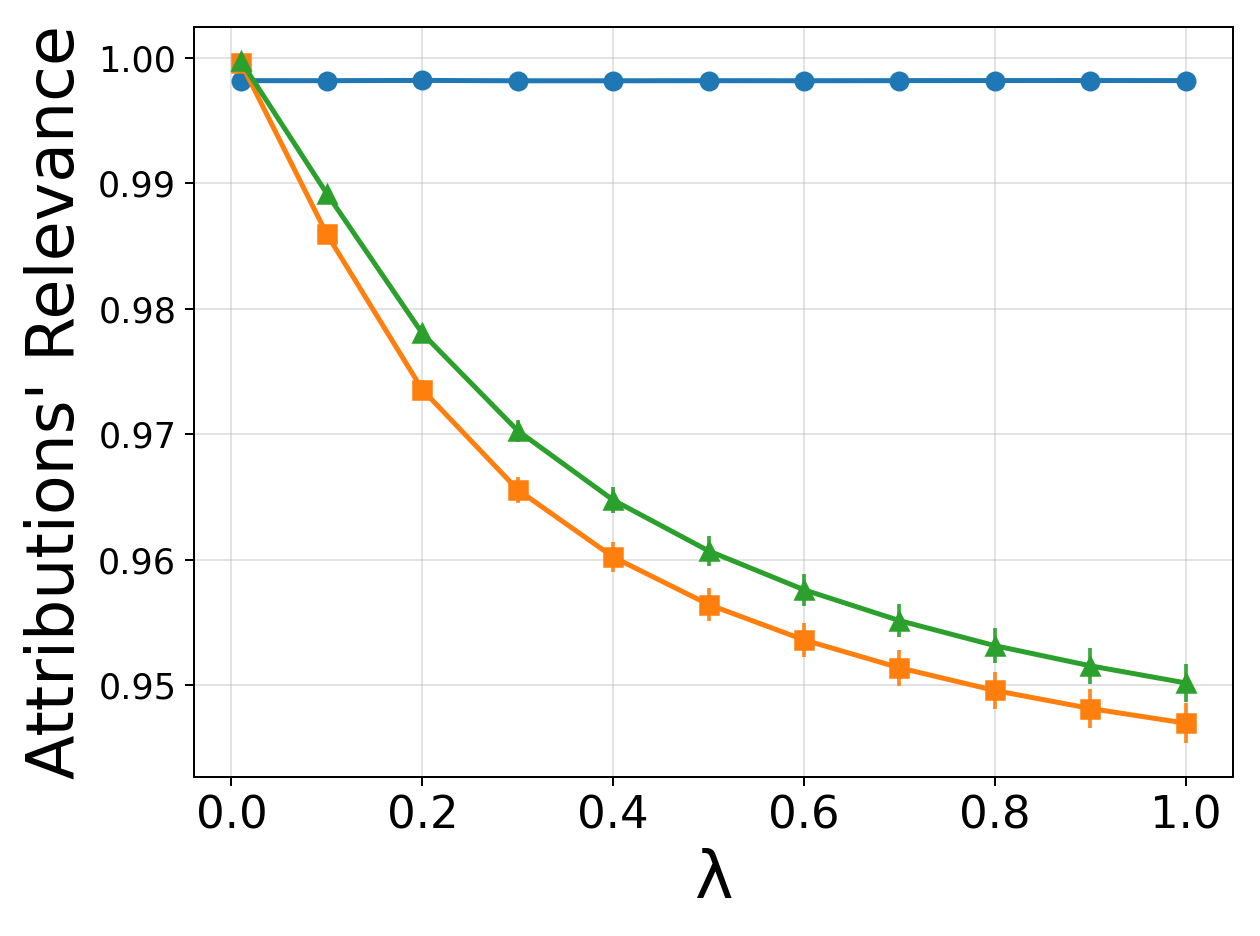}
        
        \caption{AR}\label{fig:faith_effect_AR}
    \end{subfigure}

\end{subfigure}

    \caption{Stability and welfare analysis of the \PRP{}, softmax, and linear attribution functions for various values of cost factors $\lambda$, shown in blue circles, orange squares, and green triangles, respectively.
    The top charts report stability measures: convergence ratio and convergence rate. The bottom charts report welfare measures: publishers' welfare (PW), responses' relevance (RR), and attributions' relevance (AR).}
    \label{fig:faith_effect}
\end{figure}

\section{Cost Factor Analysis}\label{apn:faith_effect}

In this section, we empirically study the effect of the cost factor $\lambda$ on the stability of the ecosystem and the social welfare, as described in Section~\ref{sec:empirical_results}. Specifically, we use the same discrete better response dynamics simulation and hyperparameters we used for $\alpha$ analysis in Section~\ref{sec:empirical_results}, while varying the value of the cost factor $\lambda$.
Figure~\ref{fig:faith_effect} reports the resulting outcomes as a function of $\lambda$, starting from $0.01$ and increasing in increments of $0.1$. We observe that, under our three attribution functions, both the convergence ratio values (Figure~\ref{fig:faith_effect_convergence_ratio}) and responses' relevance values (Figure~\ref{fig:faith_effect_RR}) remain stable across all $\lambda$ values. In contrast, the convergence rate (Figure~\ref{fig:faith_effect_convergence_rate}), the publishers' welfare (Figure~\ref{fig:faith_effect_PW}), and the attributions' relevance (Figure~\ref{fig:faith_effect_AR}) exhibit more notable trends.

We emphasize that the observed stability in response relevance is not a trivial outcome. One might have expected relevance to vary as publishers face different incentives to modify their content. Instead, our results suggest that generative platforms can preserve response relevance across heterogeneous publisher types, including publishers that differ in the cost of deviating from their original content. This capability has previously been shown to be non-trivial \cite{ohayon2026contestsspilloversincentivizingcontent}, where creators' effort can be interpreted as a manageable cost factor.

The convergence rate decreases as $\lambda$ increases, likely because for small values of $\lambda$, publishers are more incentivized to modify their initial content, resulting in a longer interaction. 
The publishers' welfare decreases consistently under the \PRP{} attribution function, while it remains relatively stable under the softmax and linear functions. In other words, higher content-modification costs (i.e., larger values of $\lambda$) lead to increased modifications under the \PRP{} function. This may be attributed to the deviation that publishers have to perform from their original content to gain exposure under the \PRP{} function.
As expected, when publishers are less incentivized to modify their initial content (large $\lambda$), the attributions' relevance decreases consistently under the softmax and linear attribution functions. Surprisingly, it remains stable under the \PRP{} function. This stability is likely due to the discontinuity of \PRP{}, which allows a publisher to obtain the entire exposure allocation by exceeding its competitors by only a small margin.

Additionally, Figure~\ref{fig:faith_effect} shows that, for the vast majority of measures and configurations, the relative performance between the three attribution functions remains consistent across all values of $\lambda$. For example, for any value of $\lambda$, the publishers' welfare is highest under the softmax function and lowest under the \PRP{} function.

\begin{figure}[t]
    \centering

\begin{subfigure}[t]{\columnwidth}
    \centering

    \begin{subfigure}[t]{0.325\columnwidth}
        \centering
        \includegraphics[width=\linewidth]{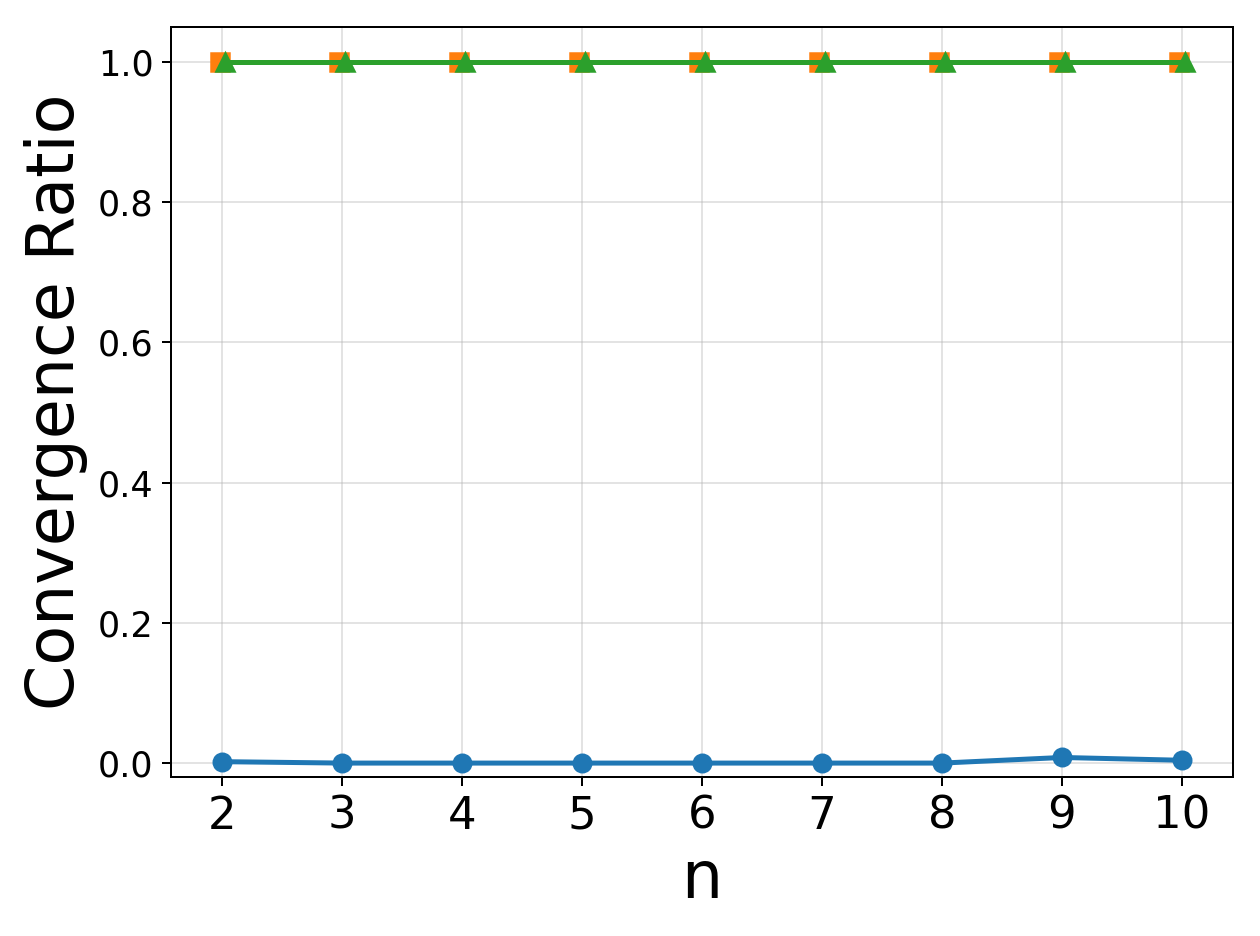}
        \caption{Convergence ratio}
        \label{fig:number_of_players_convergence_ratio}
    \end{subfigure}
    \hspace{0.06\columnwidth}
    \begin{subfigure}[t]{0.325\columnwidth}
        \centering
        \includegraphics[width=\linewidth]{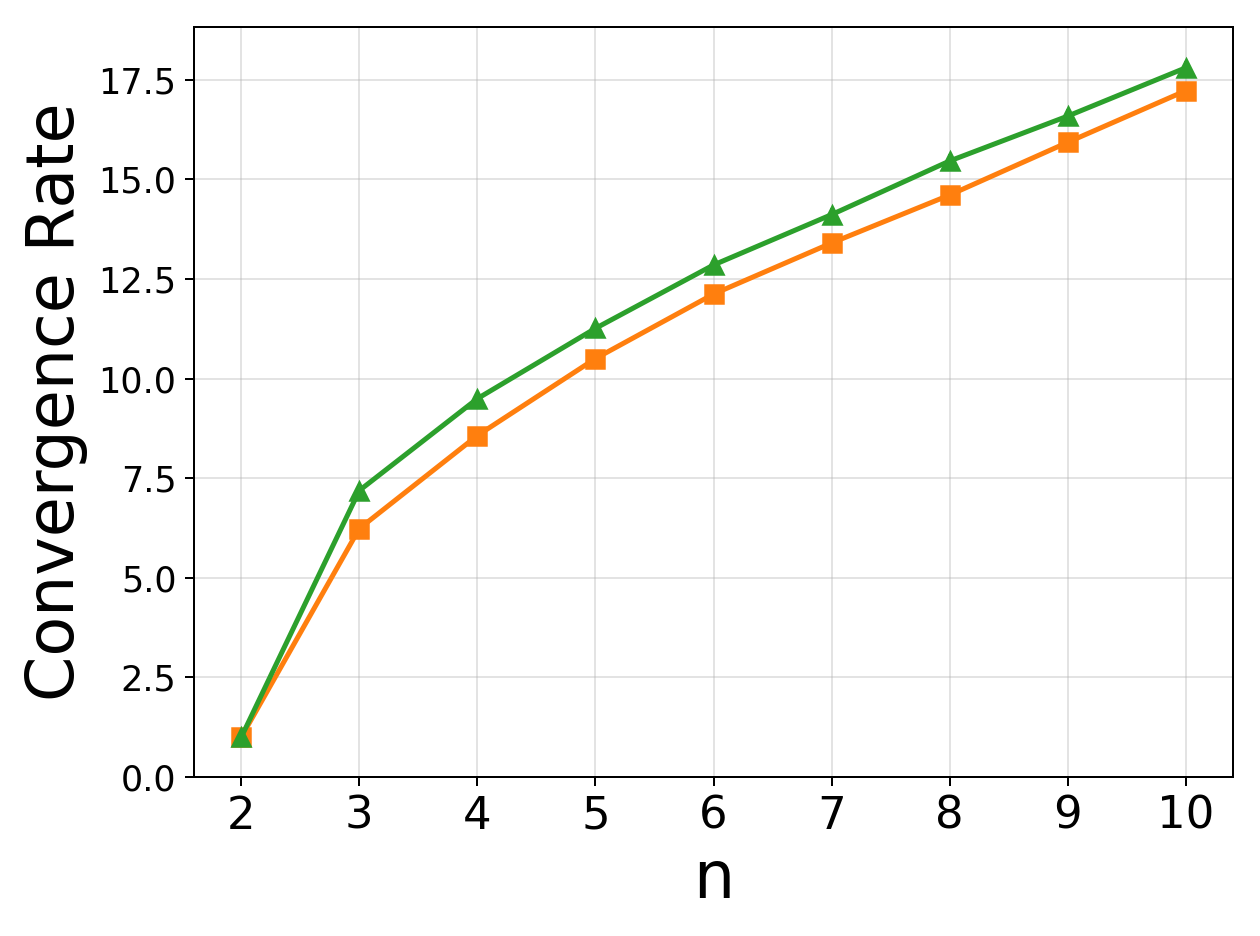}
        
        \caption{Convergence rate}\label{fig:number_of_players_convergence_rate}
    \end{subfigure}

\end{subfigure}

\begin{subfigure}[t]{\columnwidth}
    \centering

    \begin{subfigure}[t]{0.325\columnwidth}
        \centering
        \includegraphics[width=\linewidth]{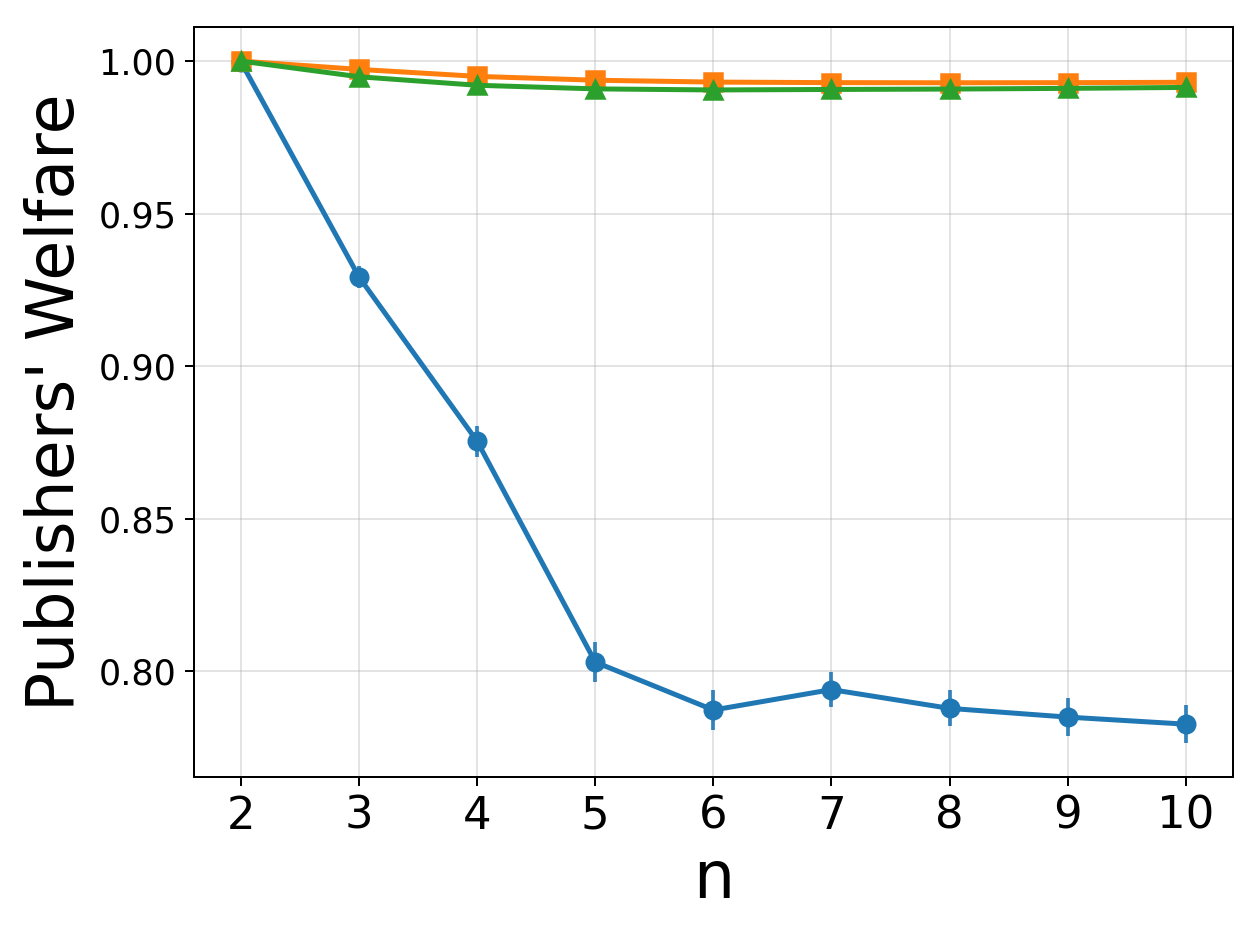}
        
        \caption{PW}\label{fig:number_of_players_PW}
    \end{subfigure}
    \hfill
    \begin{subfigure}[t]{0.325\columnwidth}
        \centering
        \includegraphics[width=\linewidth]{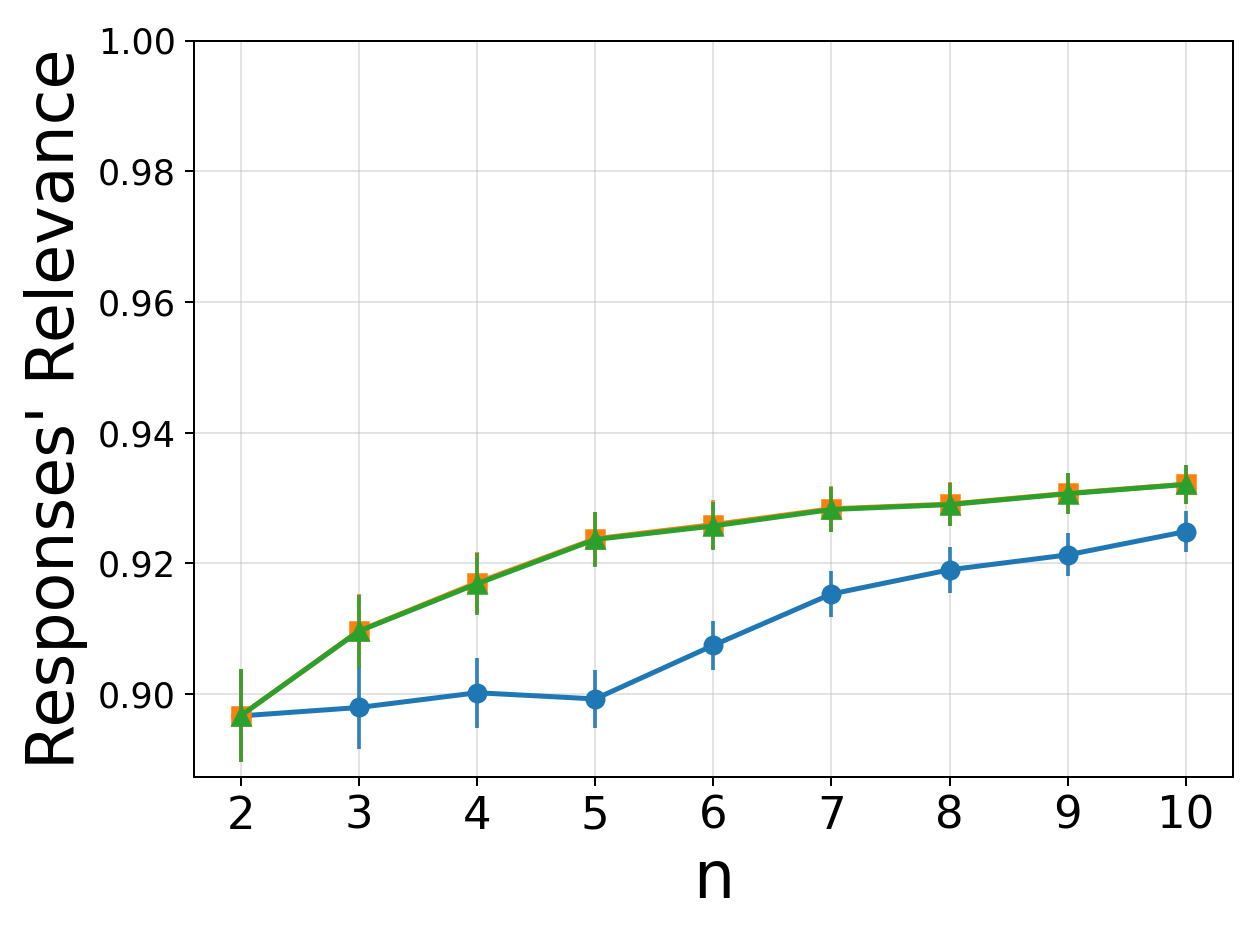}
        
        \caption{RR}\label{fig:number_of_players_RR}
    \end{subfigure}
    \hfill
    \begin{subfigure}[t]{0.325\columnwidth} 
        \centering
        \includegraphics[width=\linewidth]{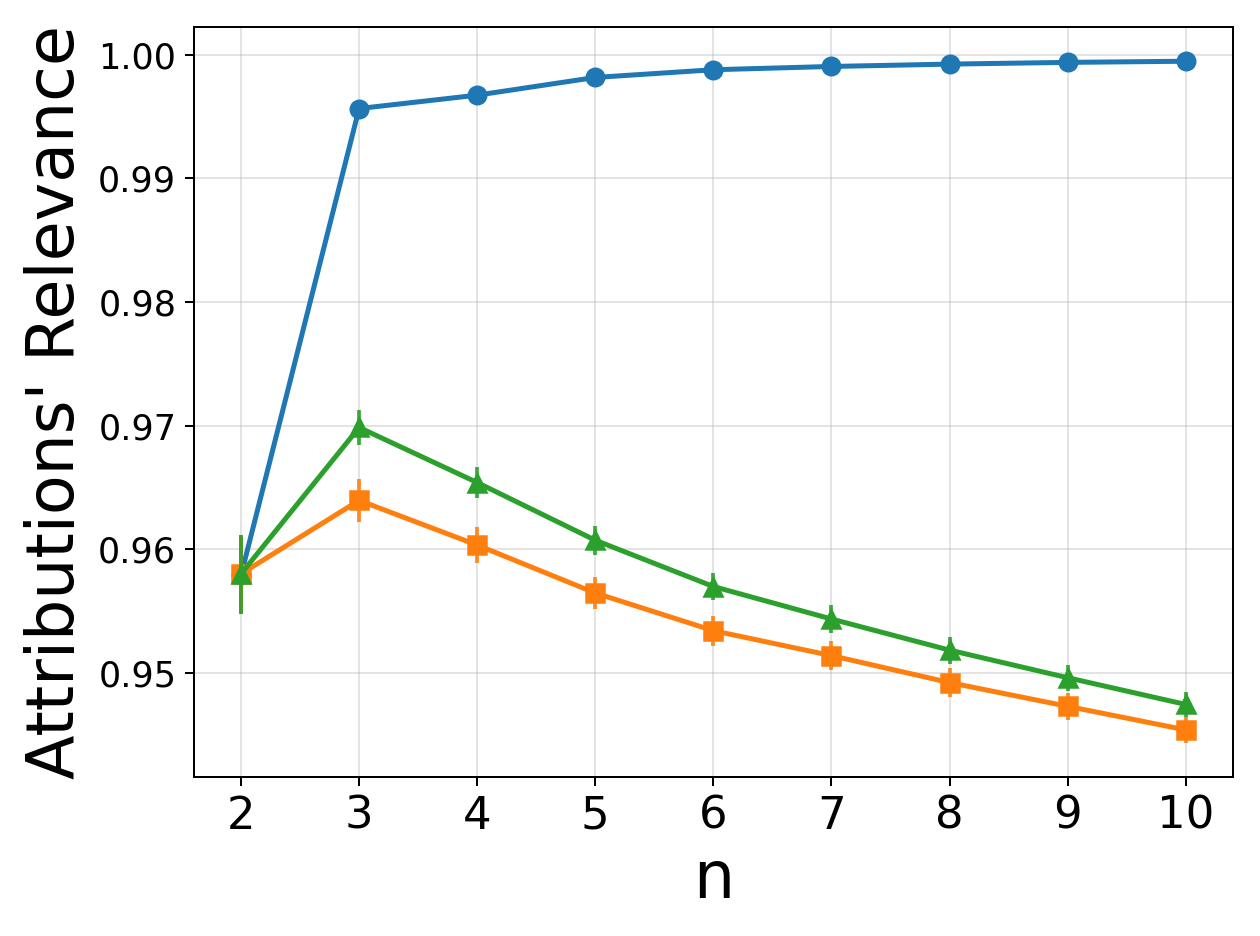}
        
        \caption{AR}\label{fig:number_of_players_AR}
    \end{subfigure}
\end{subfigure}

    \caption{Stability and welfare analysis of the \PRP{}, softmax, and linear attribution functions for various values of players $n$, shown in blue circles, orange squares, and green triangles, respectively.
    The top charts report stability measures: convergence ratio and convergence rate. The bottom charts report welfare measures: publishers' welfare (PW), responses' relevance (RR), and attributions' relevance (AR).}
    \label{fig:number_of_players}
\end{figure}

\section{Sensitivity Analysis of Our Experimental Results} \label{apn:Hyperparameters_Effect}

This section provides supplementary empirical analyses that extend the results reported in Section~\ref{sec:empirical_results}.
We assess the sensitivity of the result patterns reported in Section~\ref{sec:empirical_results} to the hyperparameter values chosen and explore novel trends.
We study the effect of the number of publishers $n$, the documents' representation dimension $\text{\di}$, the hyperparameters $\beta$ and $m$ of the softmax and linear distributions, and the distribution of the user's question and initial documents on the stability and welfare in the ecosystem.
We use the algorithm, evaluation measures, and hyperparameters as described in Section~\ref{sec:empirical_results} for $\alpha$ analysis.

\paragraph{The Number of Players.} Figure~\ref{fig:number_of_players} reports the stability and welfare measures for various numbers of players $n$ (from 2 to 10).
We highlight that in all measures and for all $n\neq 2$, the relative performance (as observed in Section~\ref{sec:empirical_results}) between our three attribution functions is preserved. 
For $n=2$, all functions perform similarly in terms of convergence rate and welfare measures, since both publishers always have identical distances from the generated response, resulting in a degenerate game in that sense.

In terms of the convergence ratio, the results presented in Section~\ref{sec:empirical_results} hold for all values of $n$. We can see in Figure~\ref{fig:number_of_players_convergence_ratio} that the softmax and linear functions induce stability (i.e., convergence ratio of 1). The \PRP{} function induces an unstable dynamics, with a convergence ratio of approximately $0$. 
As expected, in Figure~\ref{fig:number_of_players_convergence_rate} we observe that more players induce slower convergence. In fact, the convergence rate increases linearly with the number of players.

In terms of publishers' welfare (Figure~\ref{fig:number_of_players_PW}), the numbers under the softmax and linear functions remain relatively stable, whereas under the \PRP{} function, they decline and stabilize for $n>5$. This indicates that while the softmax and linear functions maintain the welfare of publishers as more publishers participate in the competition, the \PRP{} function maintains a substantially lower welfare for the higher values of $n$.

As the number of content creators increases, we observe increased performance in terms of the responses' relevance (Figure~\ref{fig:number_of_players_RR}). This suggests that GenAI ecosystems with a large number of content creators produce more relevant responses to users than those with a small number.
Moreover, for $n>3$, the attributions' relevance (Figure~\ref{fig:number_of_players_AR}) under the softmax and linear functions decreases and, under the \PRP{} function, increases. In fact, there is a substantially increased gap in relevance between the \PRP{} function and the softmax and the linear functions. Notably, the attributions' relevance is almost optimal under the \PRP{} function.
These observations are important for mechanism designers aiming to ground responses in external sources or to refer users to external sources for additional information (e.g., using citations).

\begin{figure}[t]
    \centering

\begin{subfigure}[t]{\columnwidth}
    \centering

    \begin{subfigure}[t]{0.325\columnwidth}
        \centering
        \includegraphics[width=\linewidth]{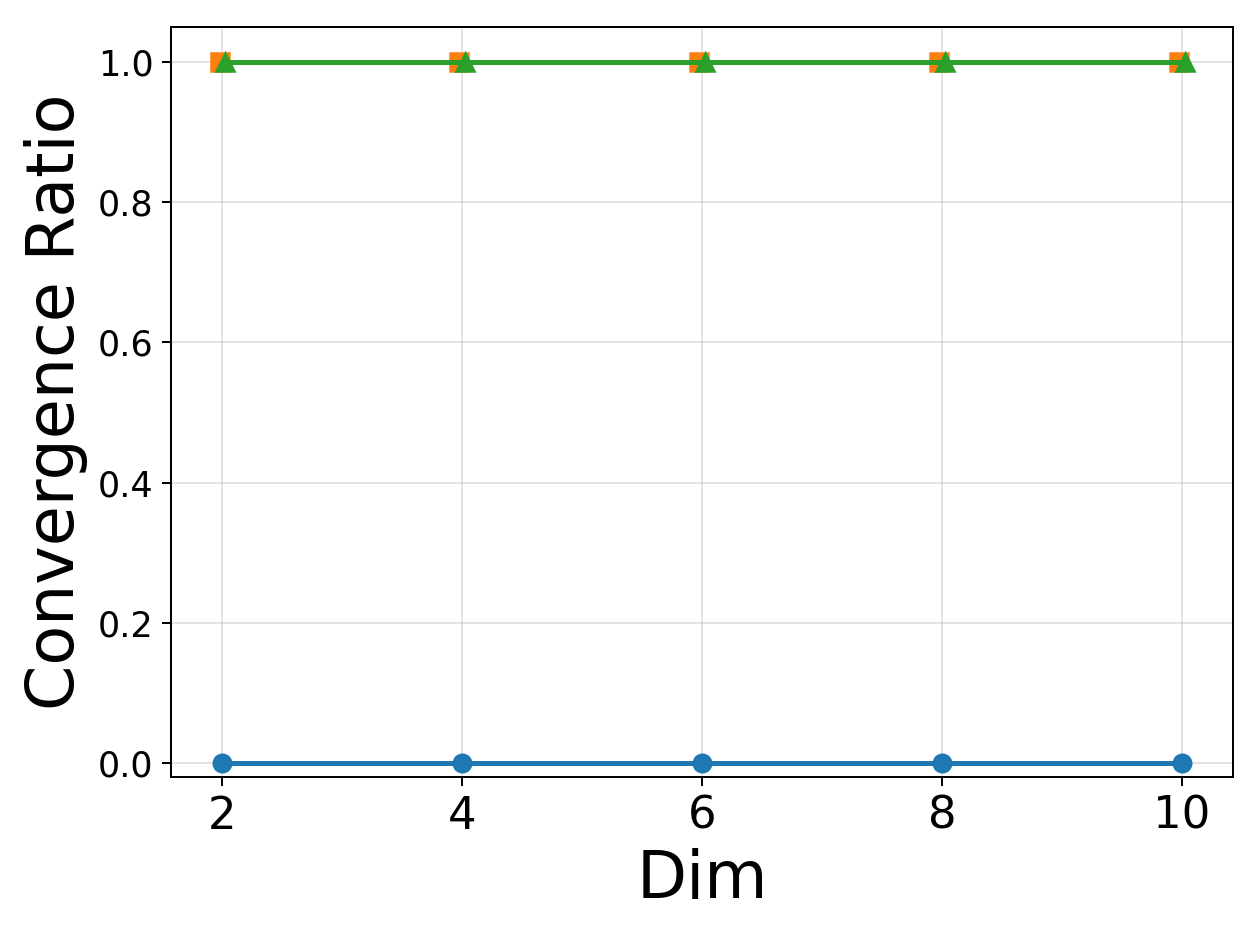}
        \caption{Convergence ratio}
        \label{fig:Dimension_convergence_ratio}
    \end{subfigure}
    \hspace{0.06\columnwidth}
    \begin{subfigure}[t]{0.325\columnwidth}
        \centering
        \includegraphics[width=\linewidth]{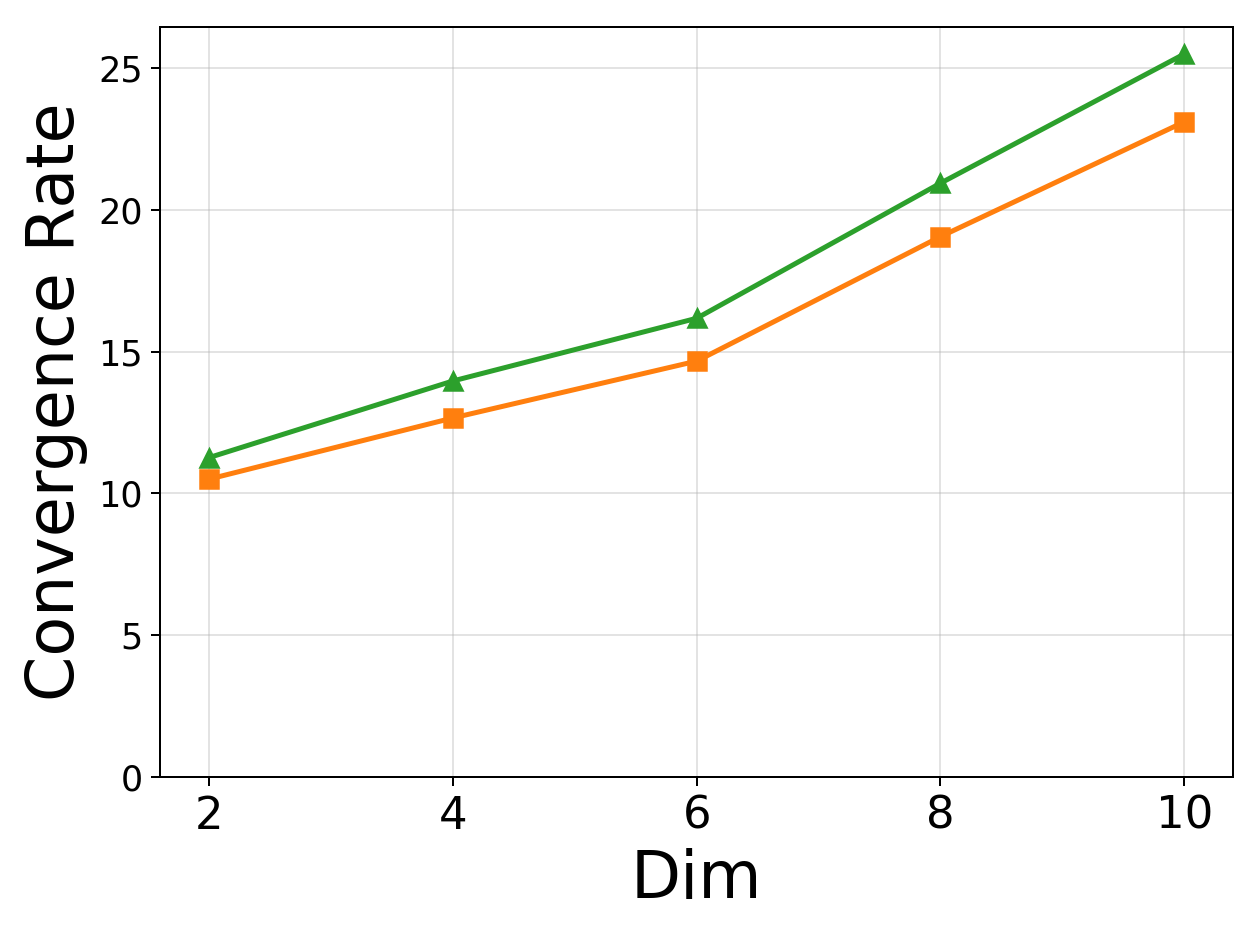}
        
        \caption{Convergence rate}\label{fig:Dimension_convergence_rate}
    \end{subfigure}

\end{subfigure}

\begin{subfigure}[t]{\columnwidth}
    \centering

    \begin{subfigure}[t]{0.325\columnwidth}
        \centering
        \includegraphics[width=\linewidth]{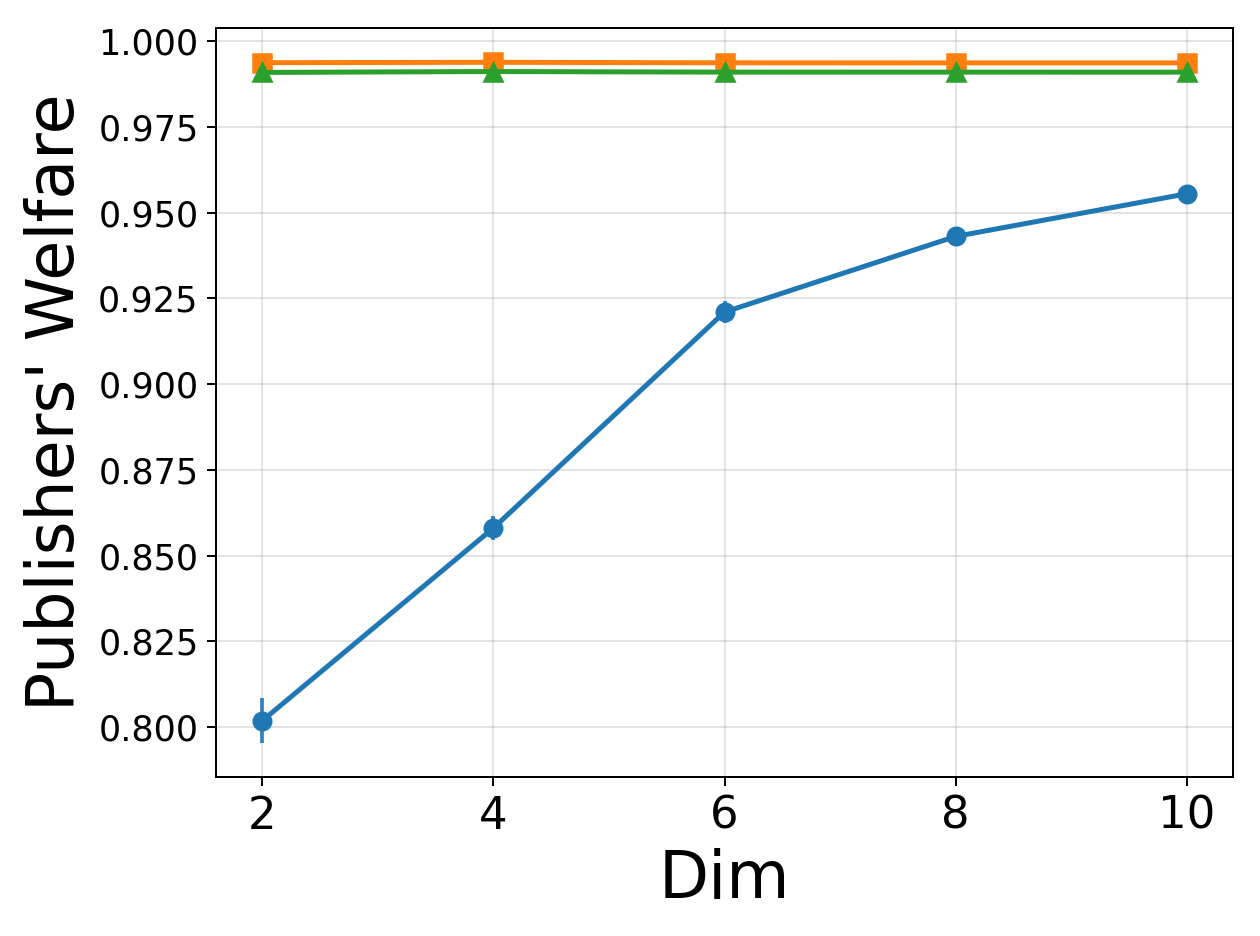}
        
        \caption{PW}\label{fig:Dimension_PW}
    \end{subfigure}
    \hfill
    \begin{subfigure}[t]{0.325\columnwidth}
        \centering
        \includegraphics[width=\linewidth]{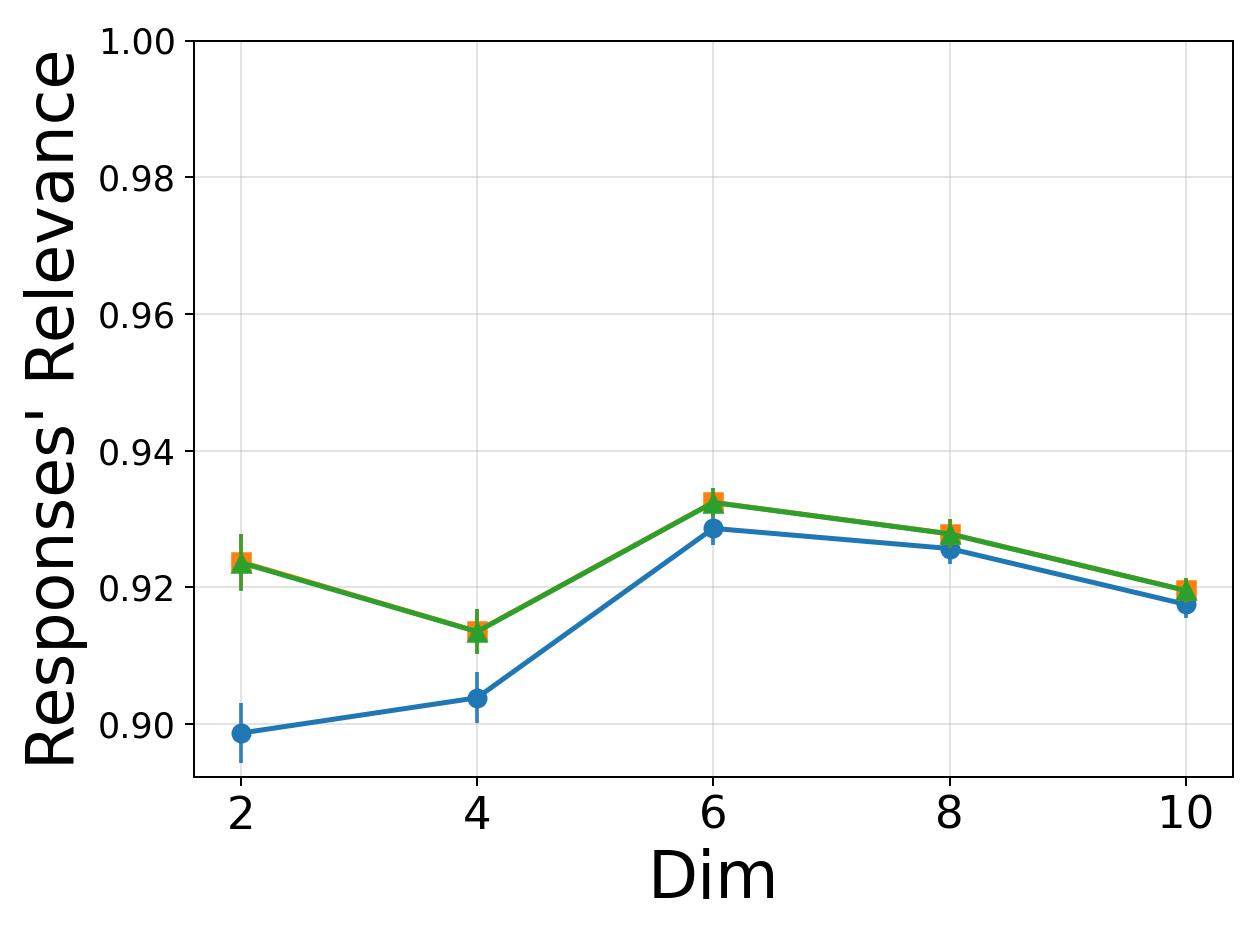}
        
        \caption{RR}\label{fig:Dimension_RR}
    \end{subfigure}
    \hfill
    \begin{subfigure}[t]{0.325\columnwidth} 
        \centering
        \includegraphics[width=\linewidth]{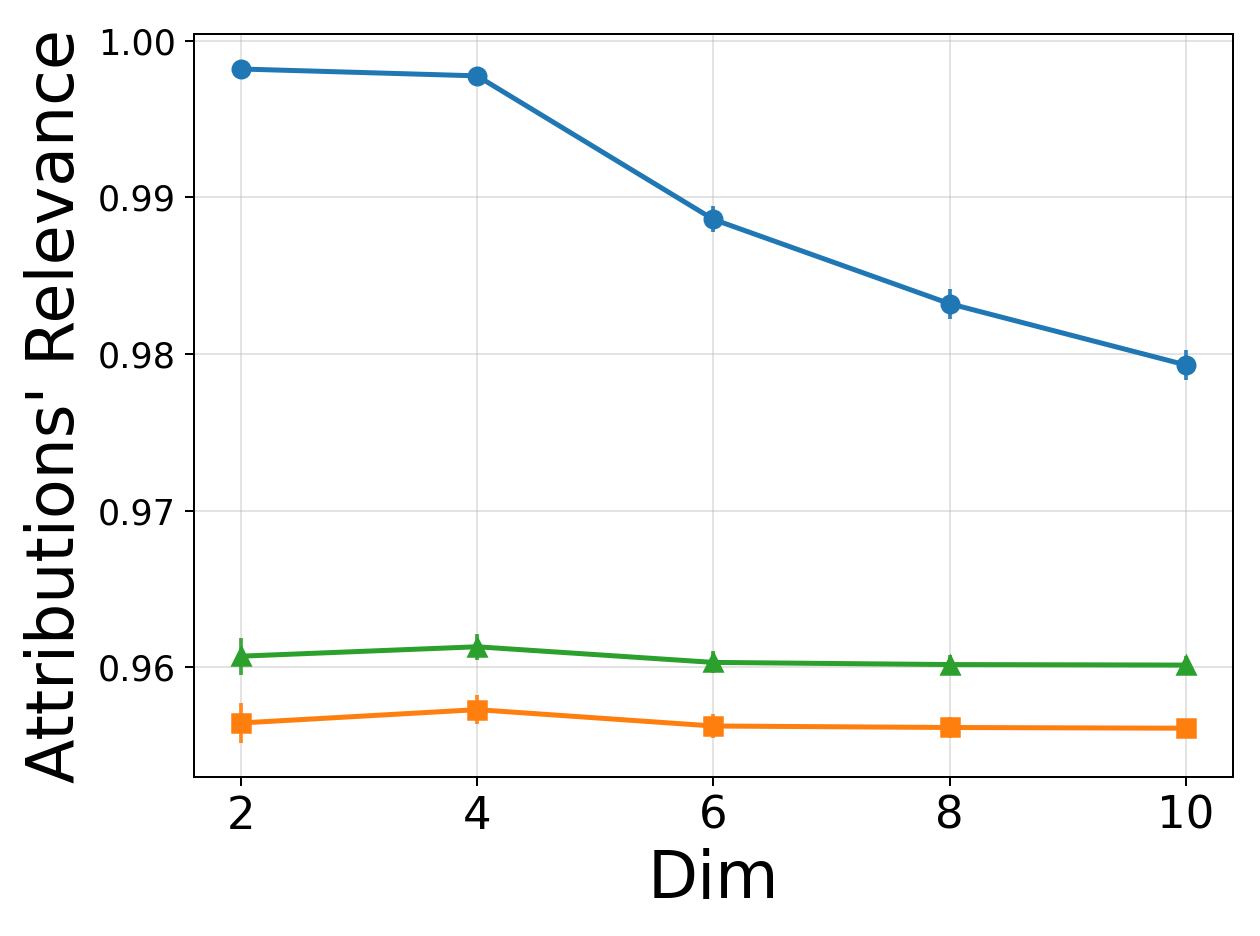}
        
        \caption{AR}\label{fig:Dimensio_AR}
    \end{subfigure}

\end{subfigure}

    \caption{Stability and welfare analysis of the \PRP{}, softmax and linear attribution functions for various values of the documents’ representations dimension $\text{\di}$, shown in blue circles, orange squares, and green triangles, respectively.
    The top charts report stability measures: convergence ratio and convergence rate. The bottom charts report welfare measures: publishers' welfare (PW), responses' relevance (RR), and attributions' relevance (AR).}
    \label{fig:dimention}
\end{figure}

\paragraph{Documents' Representation Dimension.}
Figure~\ref{fig:dimention} reports the stability and welfare measures for various values of the dimension of documents' representation in the embedding space $\text{\di}$.
We acknowledge that, due to resource limitations, our experiments are conducted using much smaller dimensions than those used by state-of-the-art LLMs (e.g., $768$ by BERT \cite{warner2024smarterbetterfasterlonger}). We hasten to point out that our goal is to identify trends, not concrete conclusions regarding a specific dimension.

Figure~\ref{fig:dimention} reports stability and welfare measures for various values of embedding dimensions \di{}. 
A key observation is that for all values of $\text{\di}$, the relative performance between our three attribution functions is preserved. 
In terms of convergence ratio, as one may expect, we can see in Figure~\ref{fig:Dimension_convergence_ratio} that it is not affected by the dimension value; specifically, the dynamics induced by the softmax and the linear functions converge for every value of \di{}.
Since, as the value of \di{} increases, convergence requires additional coordinate-wise adjustments, we observe in Figure~\ref{fig:Dimension_convergence_rate} that the convergence rate increases with the value of \di{}.

In terms of welfare, Figures \ref{fig:Dimension_PW}, \ref{fig:Dimension_RR}, and \ref{fig:Dimensio_AR} show that all measures admit diminishing gaps between the values induced by the \PRP{} function and those induced by the softmax and linear functions. Note that since the diminishing gaps may be a result of the normalized Euclidean distance, our focus is on relative performance between the functions.
The gap between the values induced by the softmax function and the linear function remains approximately the same. 
This observation suggests that, for high values of the dimension \di{}, welfare may be approximately the same across different mechanisms; however, this possibility requires further investigation.

\begin{figure}[t]
    \centering

\begin{subfigure}[t]{\columnwidth}
    \centering

    \begin{subfigure}[t]{0.325\columnwidth}
        \centering
        \includegraphics[width=\linewidth]{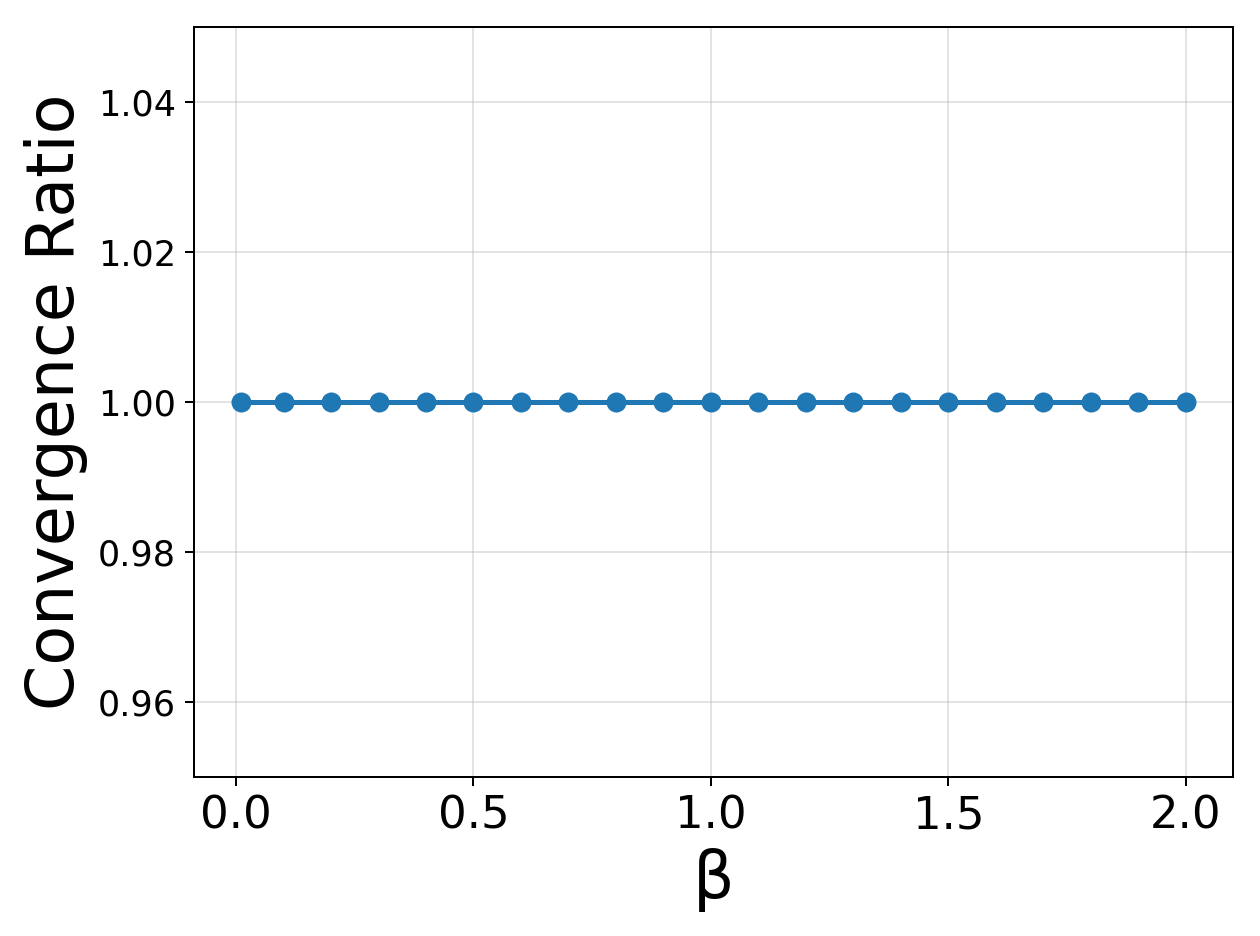}
        \caption{Convergence ratio}
        \label{fig:beta_effect_convergence_ratio}
    \end{subfigure}
    \hspace{0.06\columnwidth}
    \begin{subfigure}[t]{0.325\columnwidth}
        \centering
        \includegraphics[width=\linewidth]{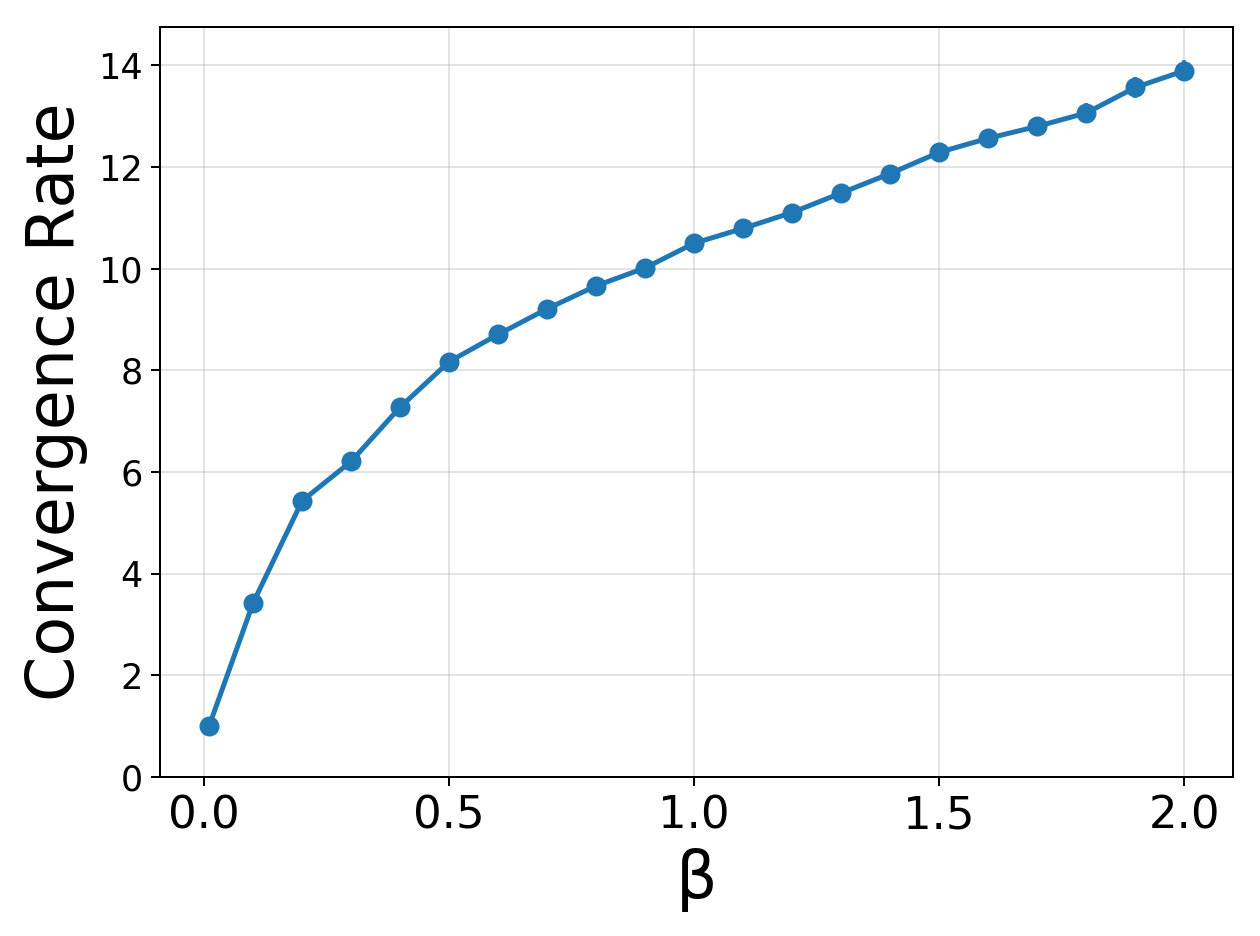}
        
        \caption{Convergence rate}\label{fig:beta_effect_convergence_rate}
    \end{subfigure}

\end{subfigure}

\begin{subfigure}[t]{\columnwidth}
    \centering

    \begin{subfigure}[t]{0.325\columnwidth}
        \centering
        \includegraphics[width=\linewidth]{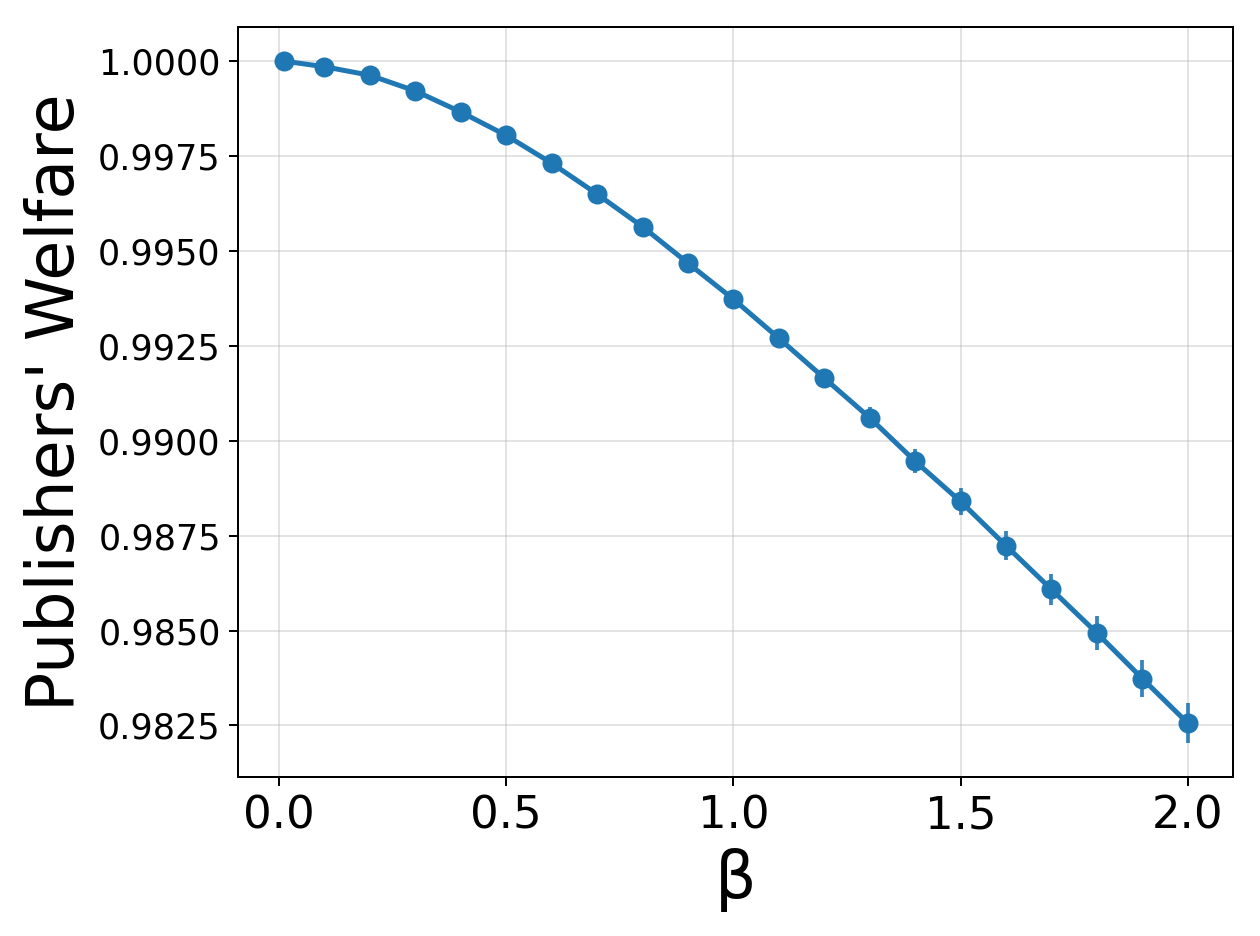}
        
        \caption{PW}\label{fig:beta_effect_PW}
    \end{subfigure}
    \hfill
    \begin{subfigure}[t]{0.325\columnwidth}
        \centering
        \includegraphics[width=\linewidth]{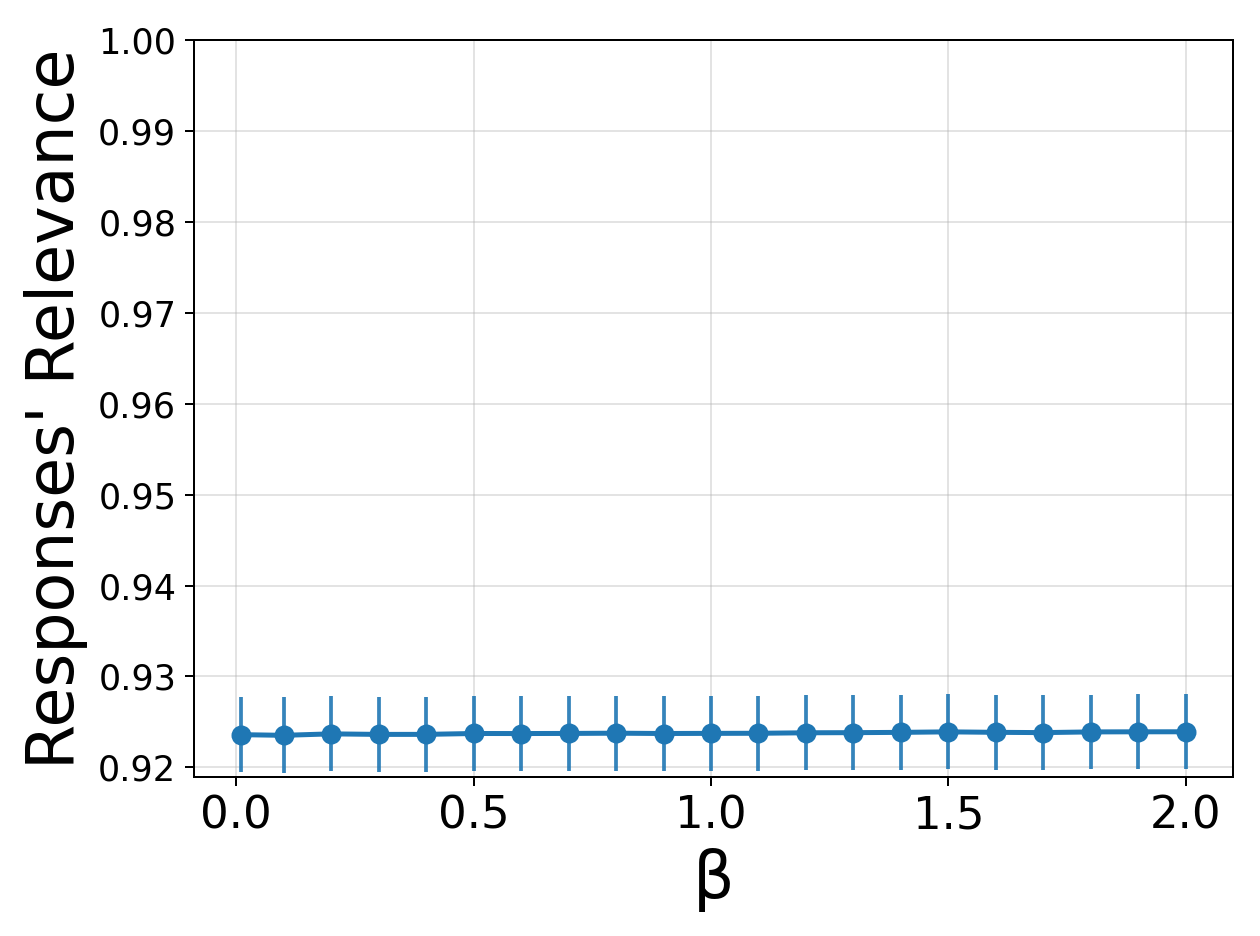}
        
        \caption{RR}\label{fig:beta_effect_RR}
    \end{subfigure}
    \hfill
    \begin{subfigure}[t]{0.325\columnwidth} 
        \centering
        \includegraphics[width=\linewidth]{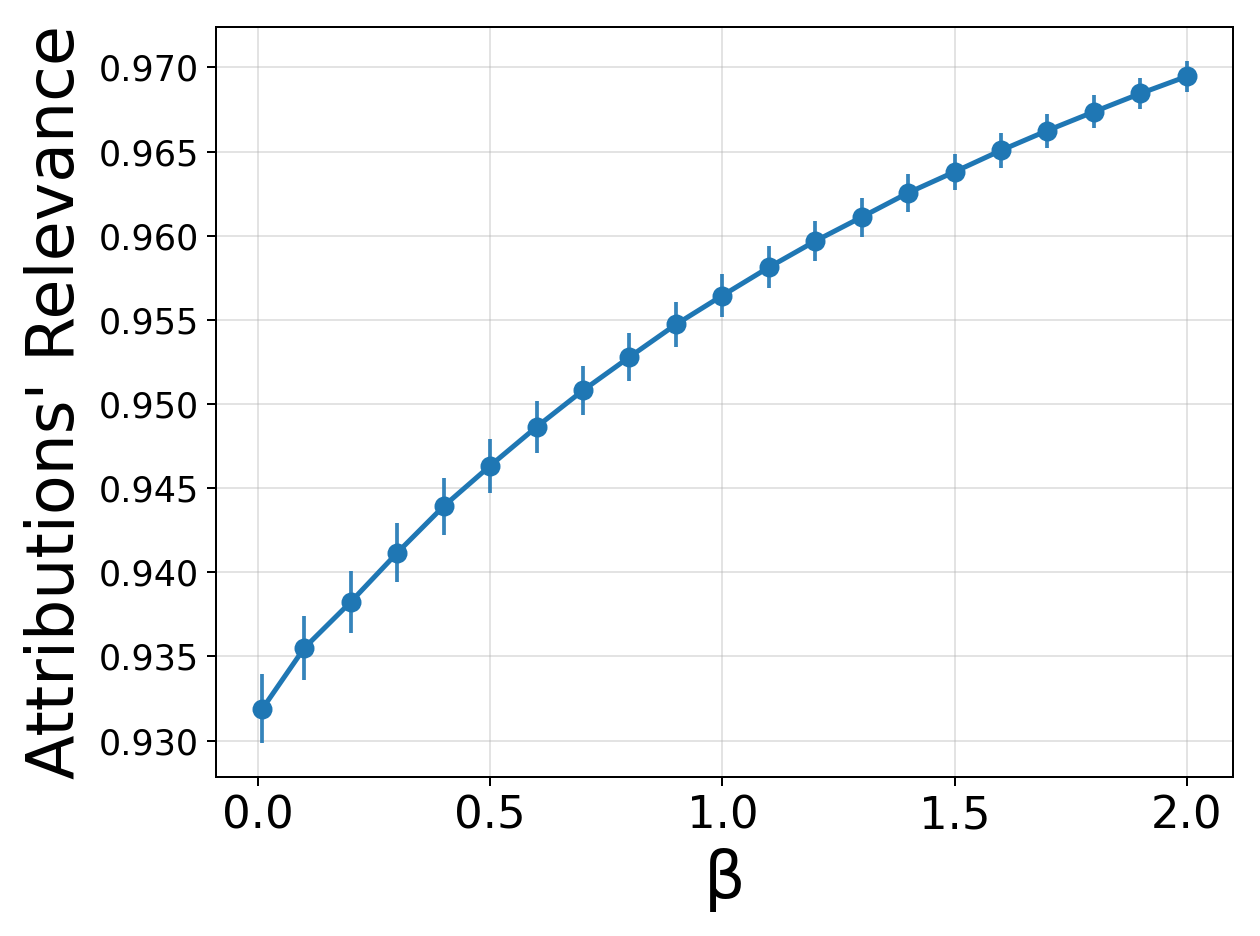}
        
        \caption{AR}\label{fig:beta_effect_AR}
    \end{subfigure}

\end{subfigure}

    \caption{Stability and welfare analysis of the softmax attribution functions for various values of the softmax function distribution's inverse-temperature $\beta$.
    The top charts report stability measures: convergence ratio and convergence rate. The bottom charts report welfare measures: publishers' welfare (PW), responses' relevance (RR), and attributions' relevance (AR).}
    \label{fig:beta_effect}
\end{figure}

\begin{figure}[t]
    \centering

\begin{subfigure}[t]{\columnwidth}
    \centering

    \begin{subfigure}[t]{0.325\columnwidth}
        \centering
        \includegraphics[width=\linewidth]{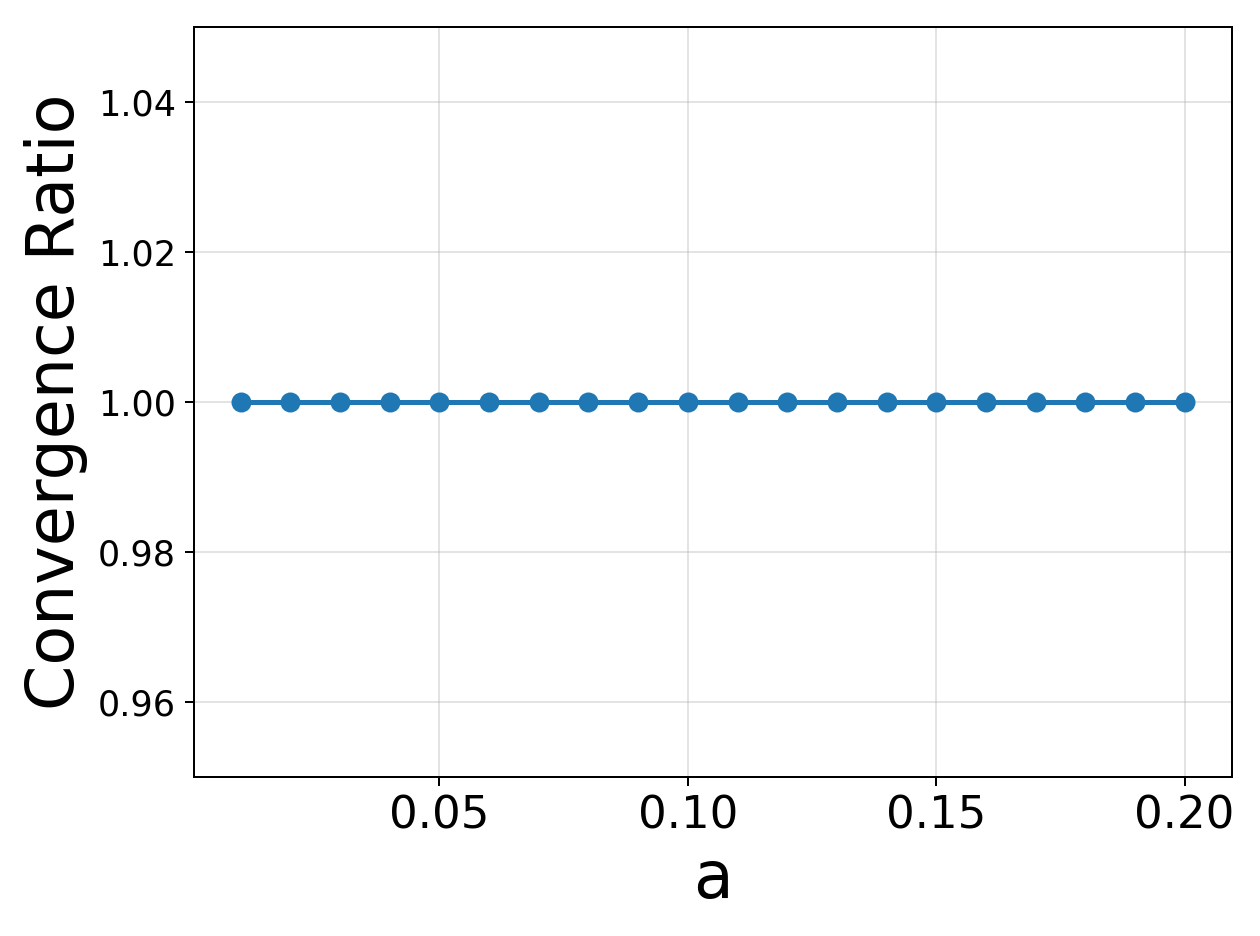}
        \caption{Convergence ratio}
        \label{fig:a_effect_convergence_ratio}
    \end{subfigure}
    \hspace{0.06\columnwidth}
    \begin{subfigure}[t]{0.325\columnwidth}
        \centering
        \includegraphics[width=\linewidth]{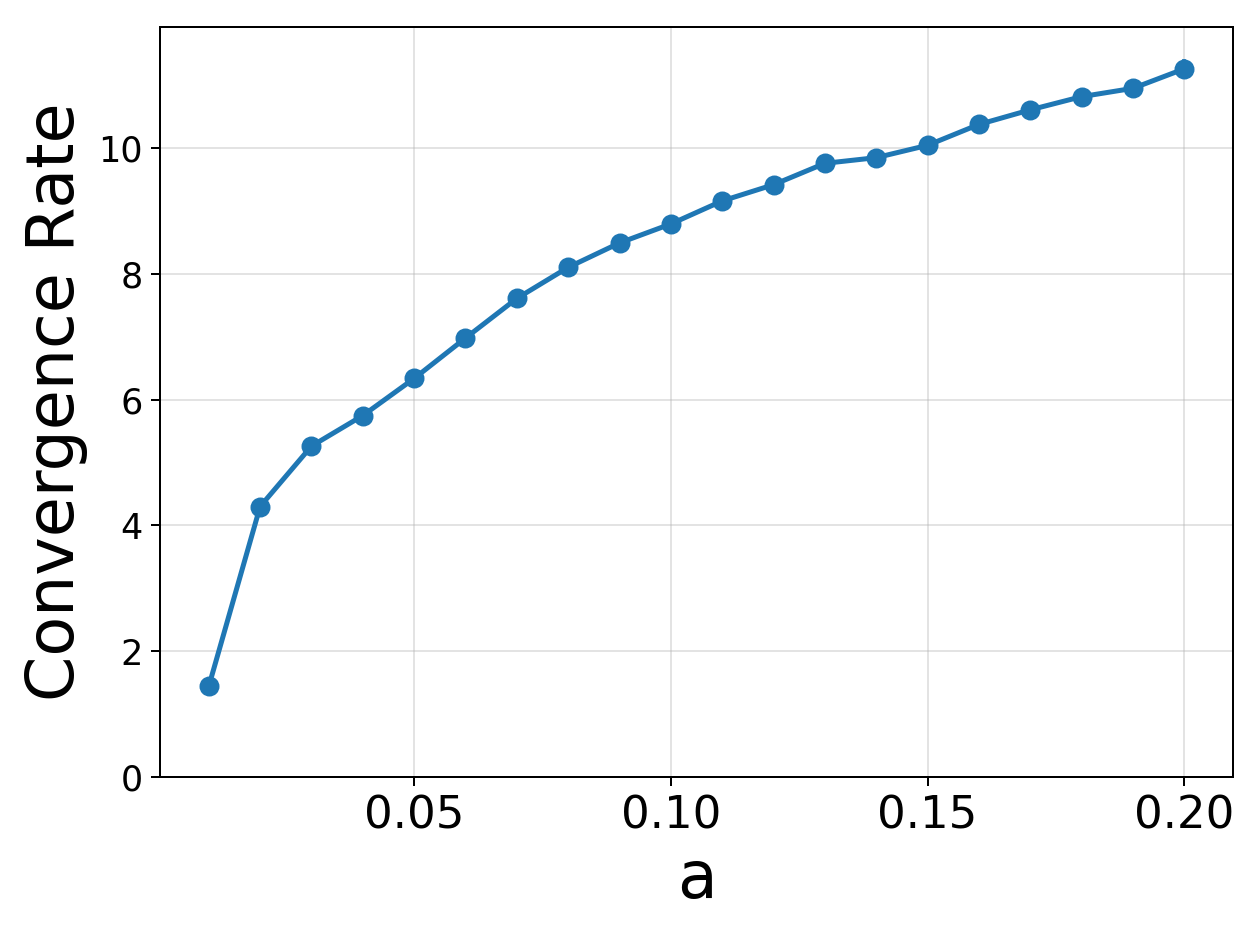}
        
        \caption{Convergence rate}\label{fig:a_effect_convergence_rate}
    \end{subfigure}

\end{subfigure}

\begin{subfigure}[t]{\columnwidth}
    \centering

    \begin{subfigure}[t]{0.325\columnwidth}
        \centering
        \includegraphics[width=\linewidth]{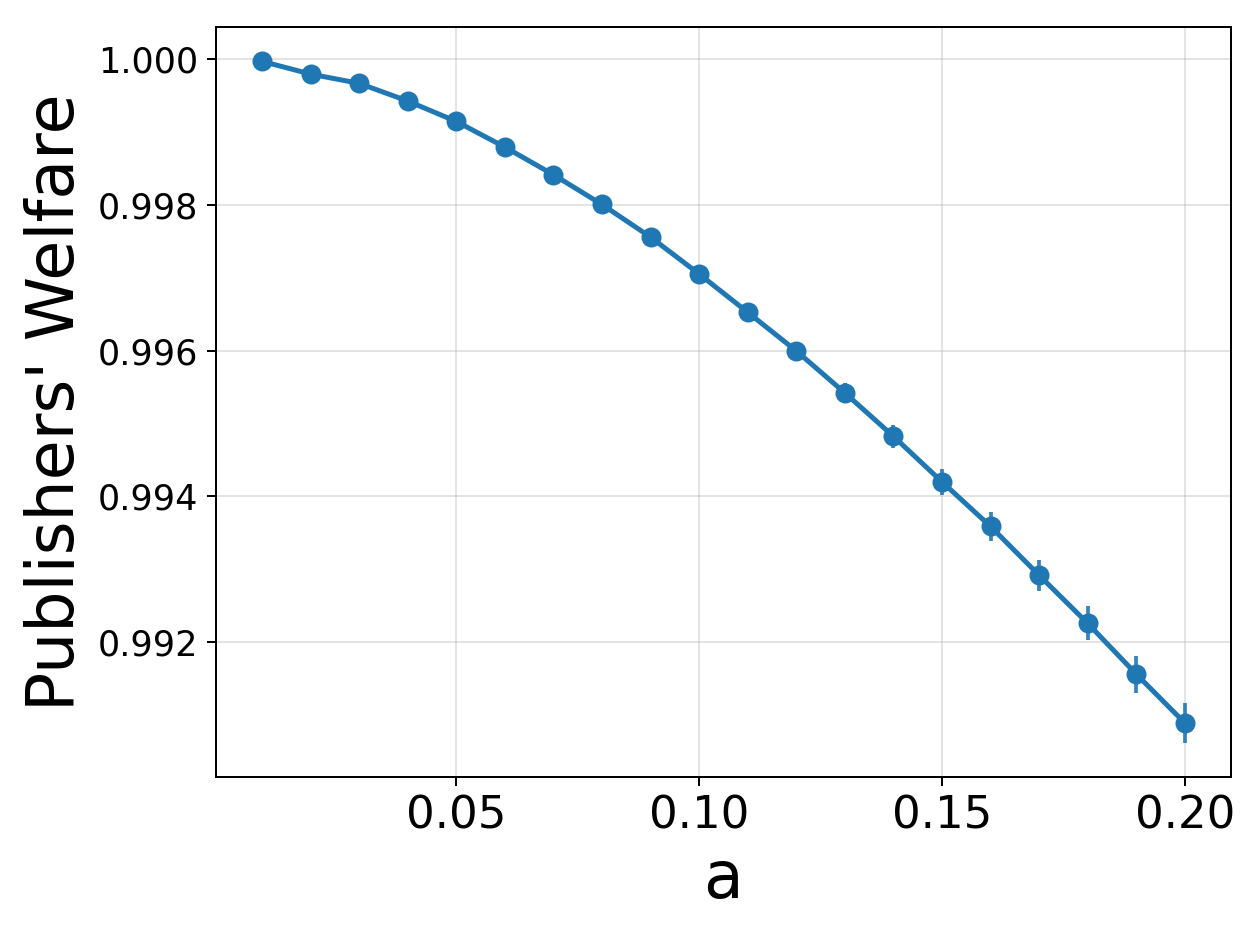}
        
        \caption{PW}\label{fig:a_effect_PW}
    \end{subfigure}
    \hfill
    \begin{subfigure}[t]{0.325\columnwidth}
        \centering
        \includegraphics[width=\linewidth]{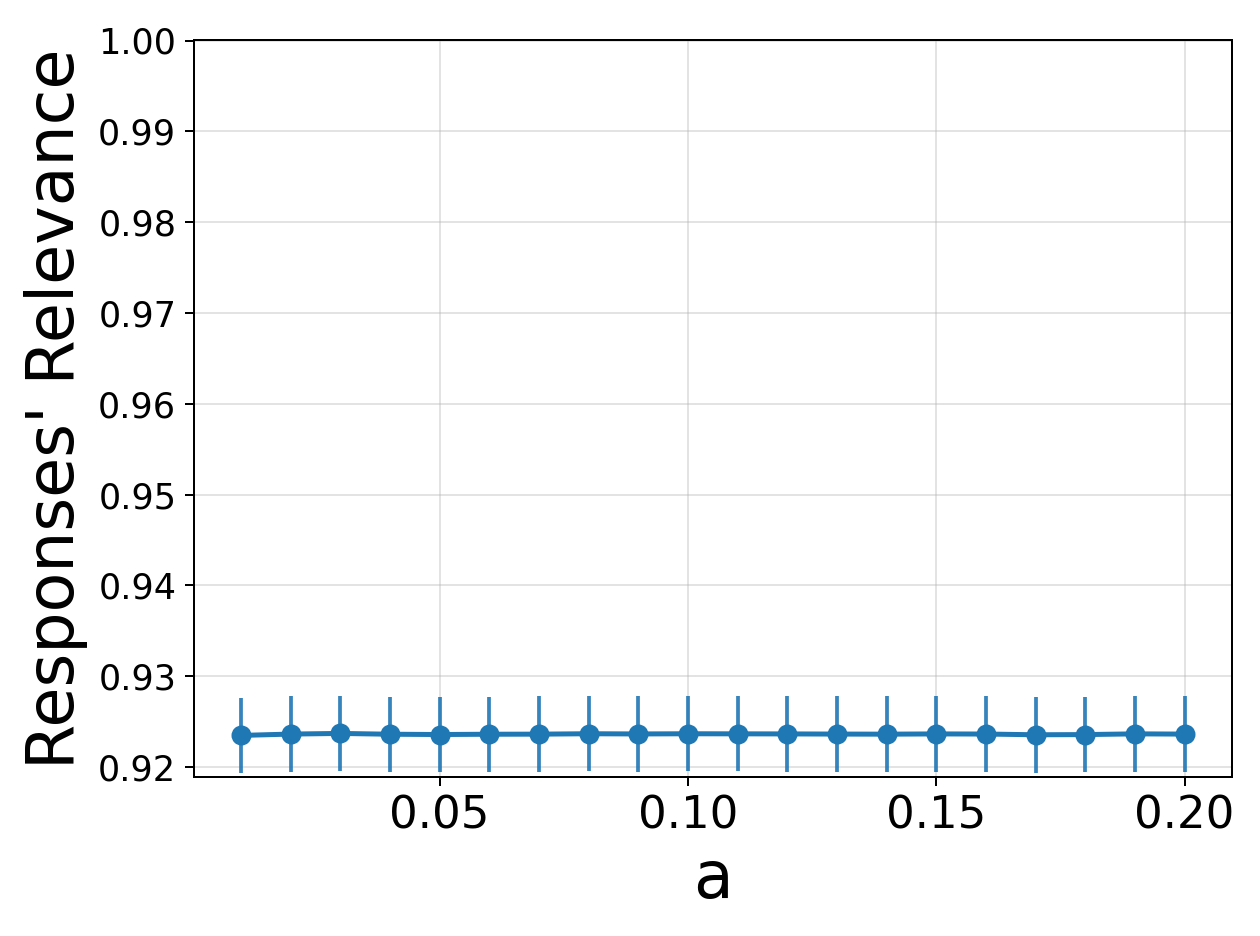}
        
        \caption{RR}\label{fig:a_effect_RR}
    \end{subfigure}
    \hfill
    \begin{subfigure}[t]{0.325\columnwidth} 
        \centering
        \includegraphics[width=\linewidth]{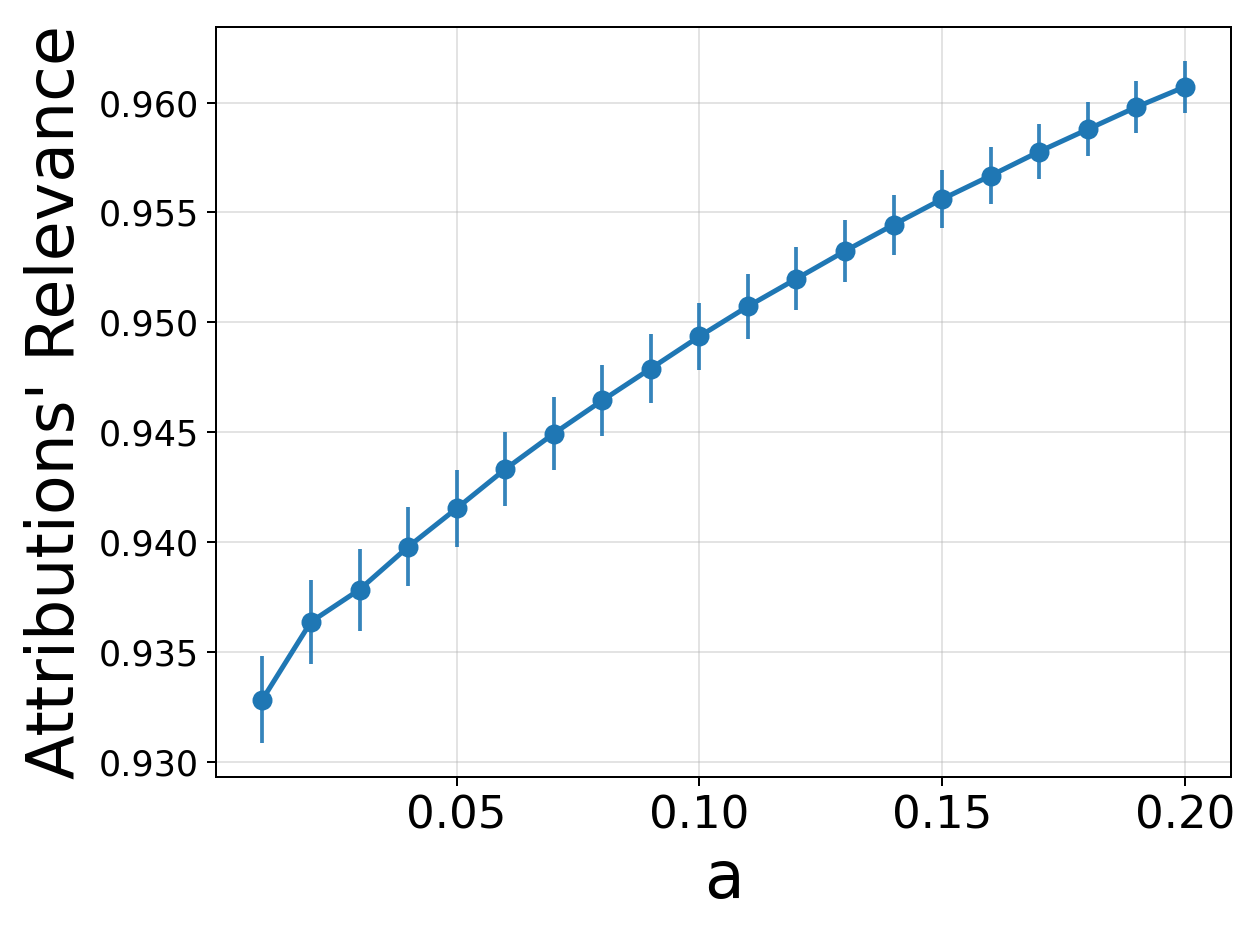}
        
        \caption{AR}\label{fig:a_effect_AR}
    \end{subfigure}

\end{subfigure}

    \caption{Stability and welfare analysis of the linear attribution functions for various values of the linear relative relevance distribution's slope $m$.
    The top charts report stability measures: convergence ratio and convergence rate. The bottom charts report welfare measures: publishers' welfare (PW), responses' relevance (RR), and attributions' relevance (AR).}
    \label{fig:a_effect}
\end{figure}

\paragraph{Distributions' Hyperparameters.}
We now turn to examine the effect of distributional hyperparameters induced by the softmax and linear functions. Specifically, the inverse-temperature $\beta$ of the softmax distribution and the slope $m$ of the linear distribution.

Figure~\ref{fig:beta_effect} presents the stability and welfare measures under the softmax attribution function for various values of $\beta$, ranging from $0.01$\footnote{The smallest value of $\beta$ is larger than $0$ since $\beta = 0$ entails a constant distribution over publishers; i.e., publishers' action does not affect their attribution score.} to $2$ in increments of $0.1$. 
We observe in Figure~\ref{fig:beta_effect_convergence_ratio} that the convergence ratio values remain constant for every $\beta$, i.e., the ecosystem is stable under the softmax function for any examined value of temperature.
As expected, Figure~\ref{fig:beta_effect_convergence_rate} shows that the convergence rate increases with the inverse-temperature $\beta$. This can be addressed by the fact that larger values of $\beta$ lead to a more winner-focused attribution score, which strengthens the incentives of publishers to deviate at each round.

Furthermore, we identify that while the values of the responses' relevance in Figure~\ref{fig:beta_effect_RR} are approximately constant with $\beta$, in Figure~\ref{fig:beta_effect_PW} and Figure~\ref{fig:beta_effect_AR}, a trade-off exists between publishers' welfare and attributions' relevance.
For smaller $\beta$ values, publishers whose documents are closer to the responses receive a higher attribution score (and, in turn, utility) compared to larger $\beta$ values. 
Consequently, publishers must apply more substantial modifications to their content to increase their utility, which in turn leads to fewer modifications (i.e., higher publishers' welfare) overall and less relevant attributions (i.e., lower attributions' relevance). 
We emphasize that the inverse-temperature value used in Section~\ref{sec:empirical_results}, $\beta = 1$, induces the standard softmax function and provides a natural balance in this trade-off.

Figure~\ref{fig:a_effect} presents the stability and welfare measures under the linear attribution function for various values of $m$, ranging from $0.01$ to $0.2$ (maximum value to induce an appropriate distribution in games with $n=5$) in increments of $0.01$. 
The observed trends are the same as those presented for the inverse-temperature $\beta$ of the softmax function above.
However, in the case of the slope $m$, the reason is different: smaller values induce a more uniform distribution over publishers, thus publishers are less motivated to modify their documents.

\begin{figure}[t]
\centering

\begin{subfigure}[t]{\columnwidth}
\centering

\includegraphics[width=0.325\linewidth]{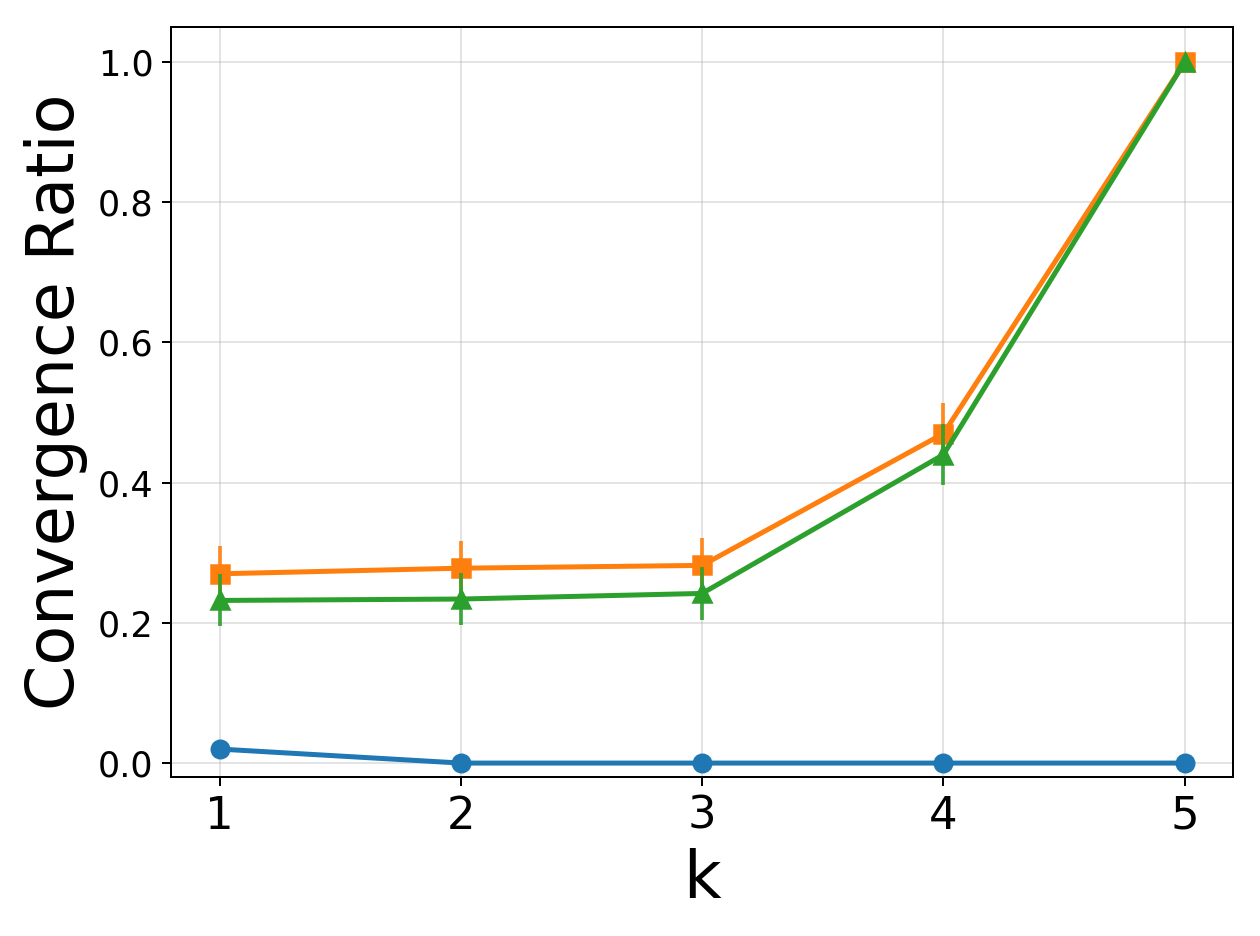}
\hspace{0.06\columnwidth}
\includegraphics[width=0.325\linewidth]{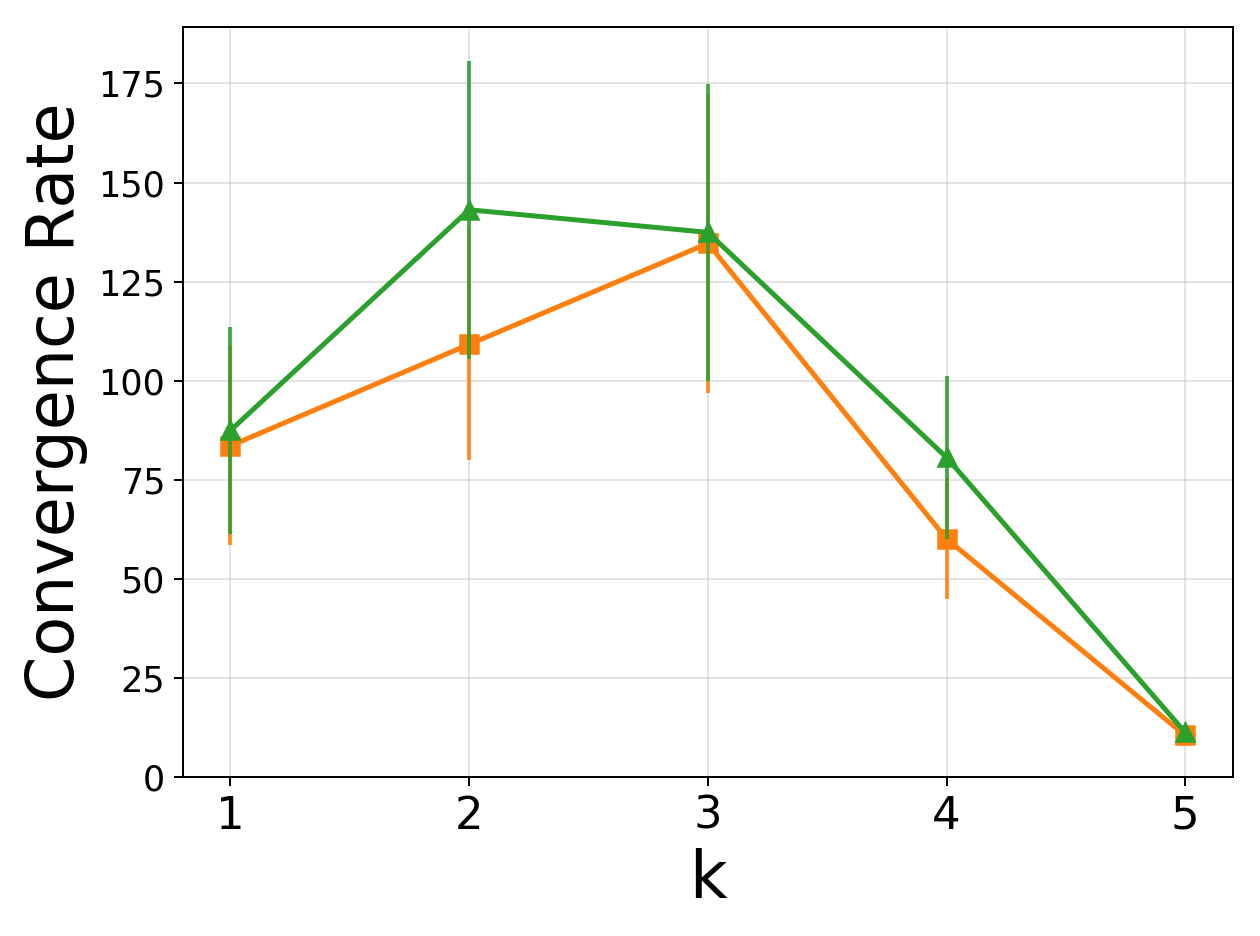}

\includegraphics[width=0.325\linewidth]{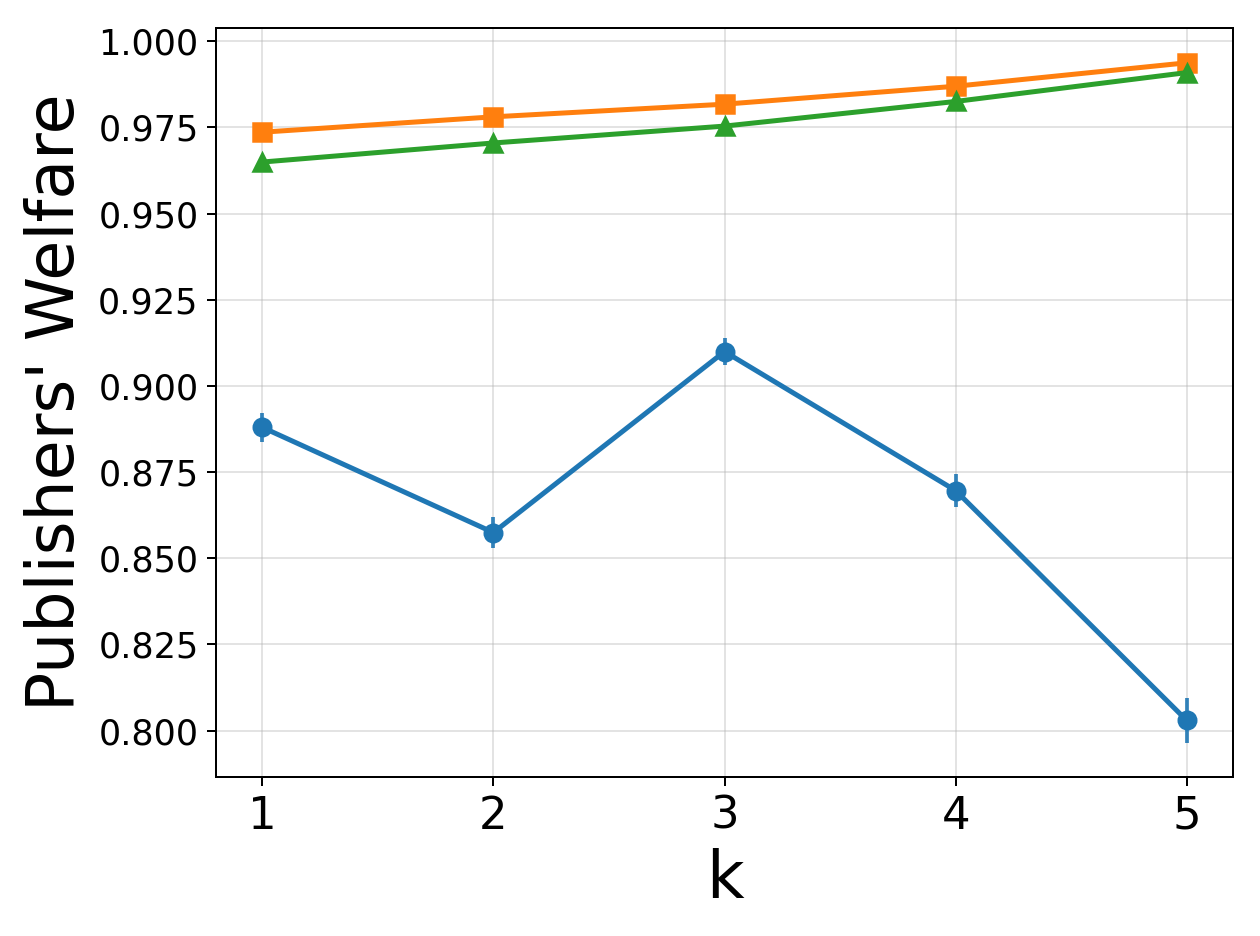}
\hfill
\includegraphics[width=0.325\linewidth]{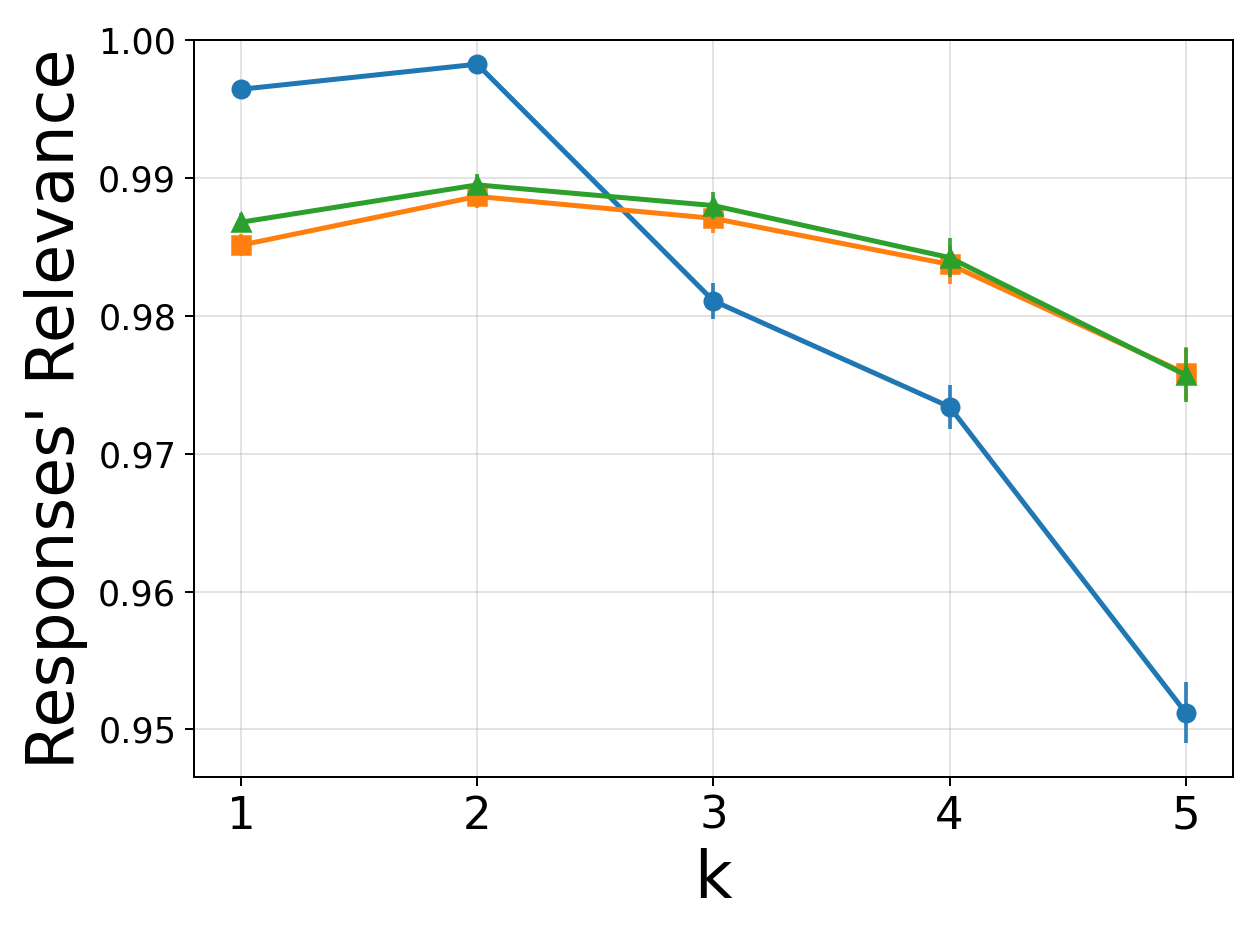}
\includegraphics[width=0.325\linewidth]{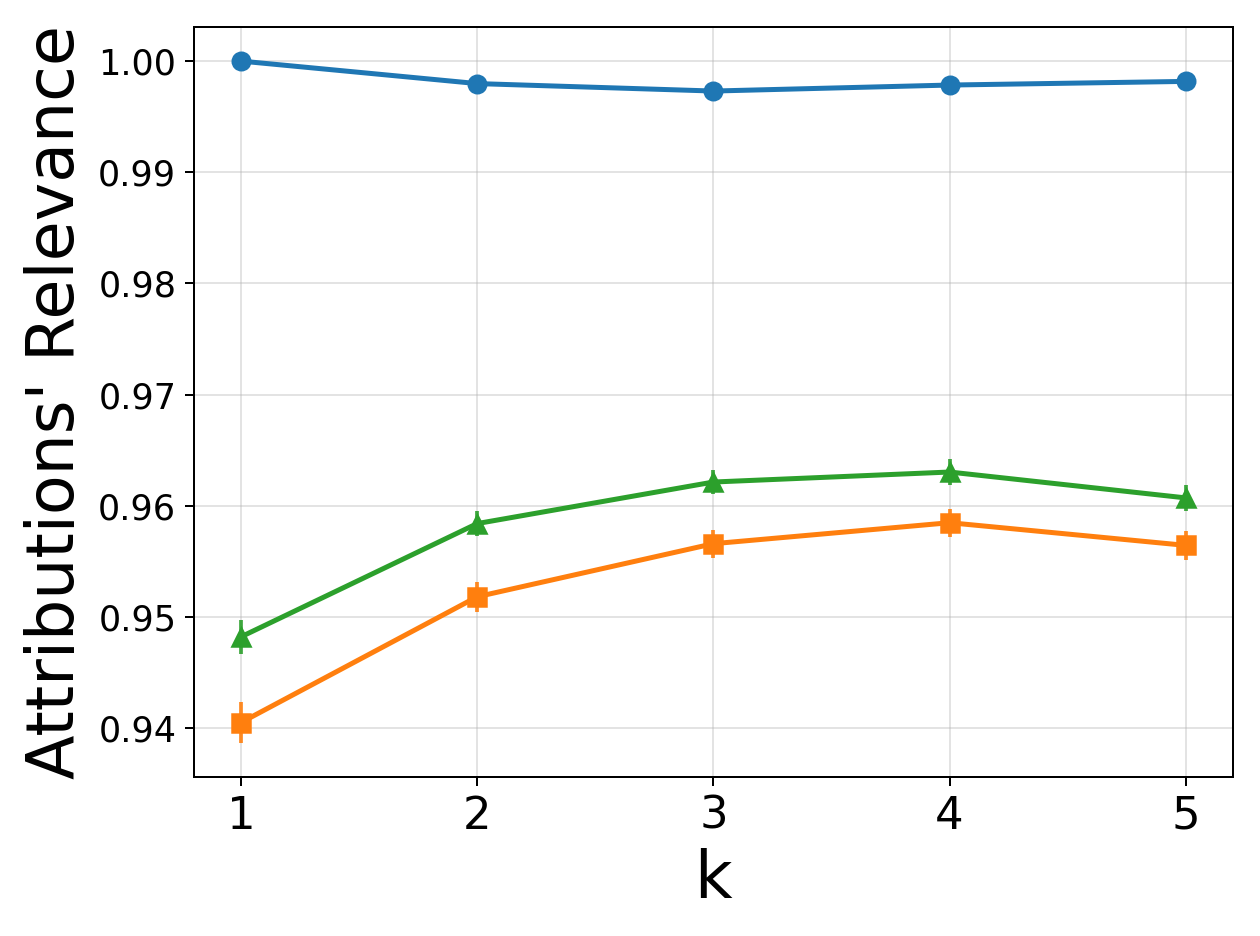}

\caption{$x^* \sim \mathcal{N}(\boldsymbol{0.5}_\text{\di{}}, 0.1 \cdot \mathcal{I})$, $\quad x^0_i \sim \mathcal{U}(\mathcal{V})$}\label{fig:normal_uniform}
\end{subfigure}
\hfill

\begin{subfigure}[t]{\columnwidth}
\centering

\includegraphics[width=0.325\linewidth]{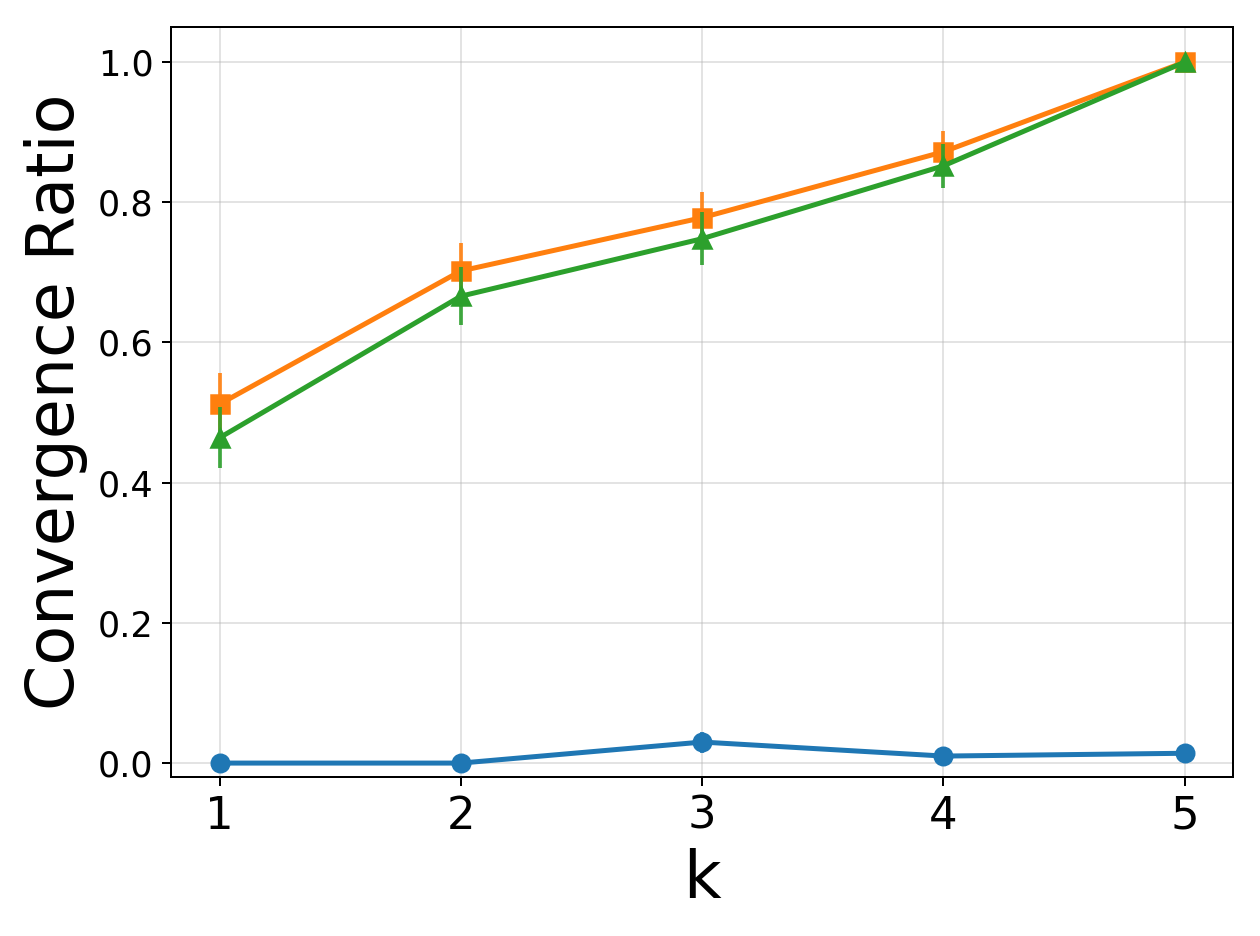}
\hspace{0.06\columnwidth}
\includegraphics[width=0.325\linewidth]{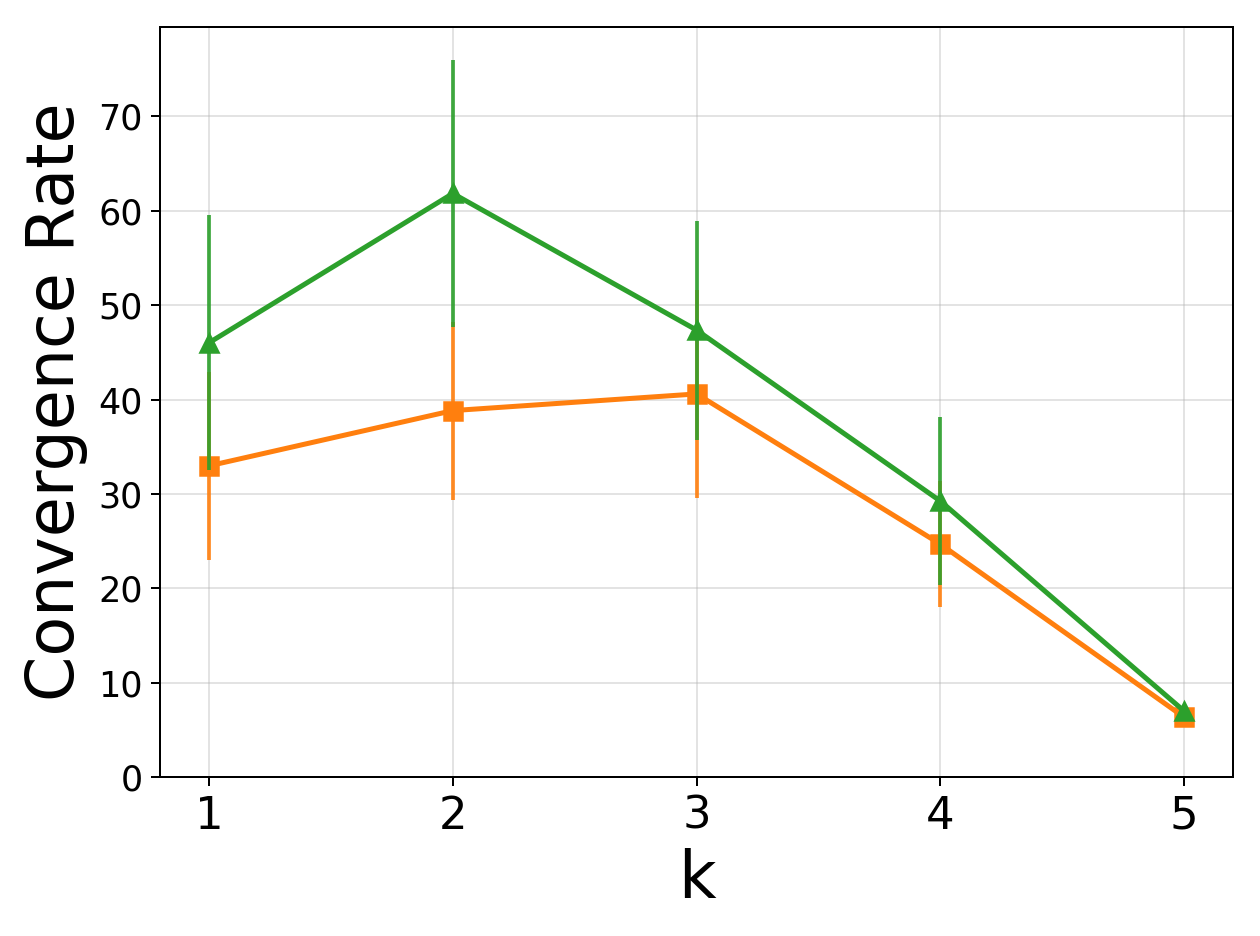}

\includegraphics[width=0.325\linewidth]{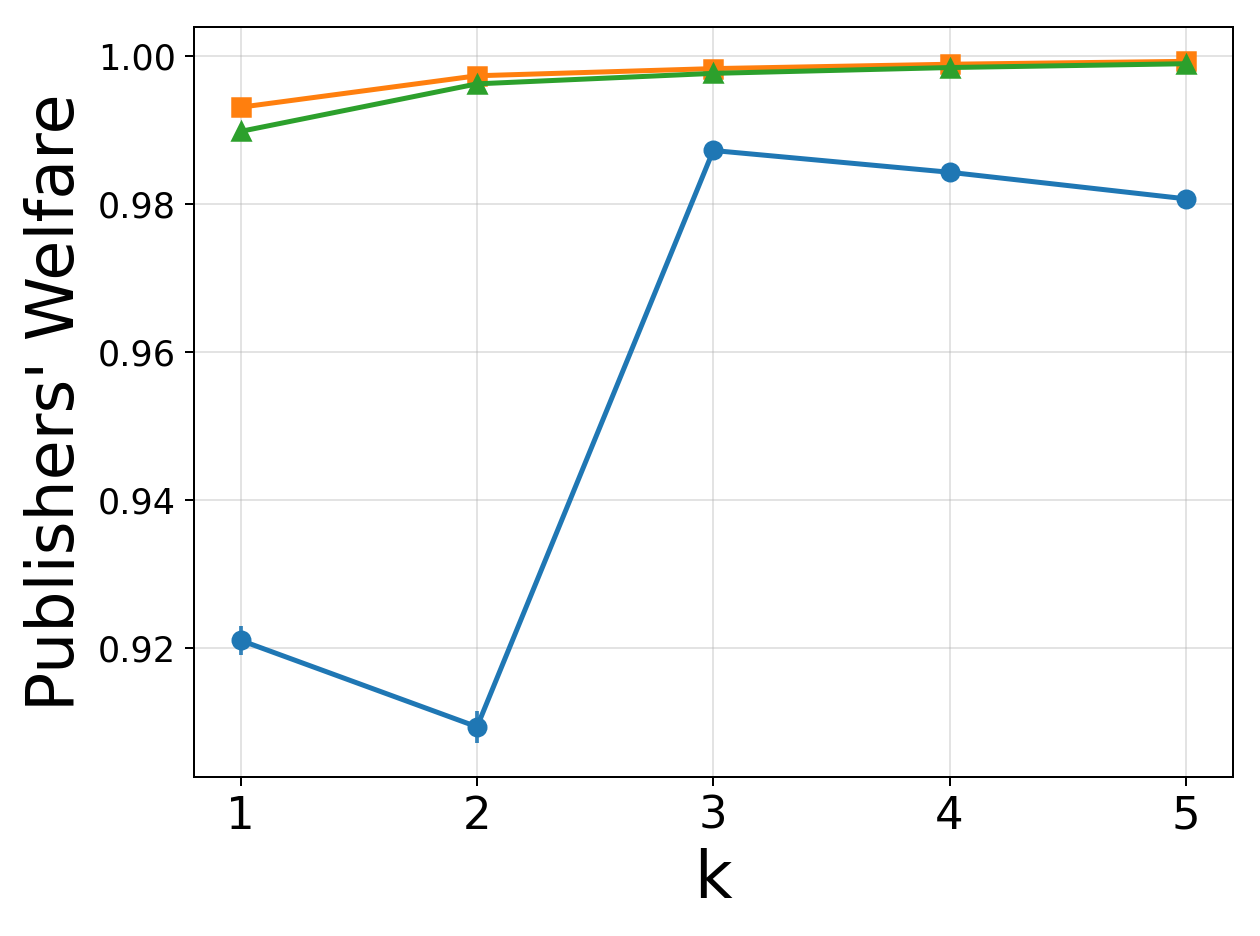}
\hfill
\includegraphics[width=0.325\linewidth]{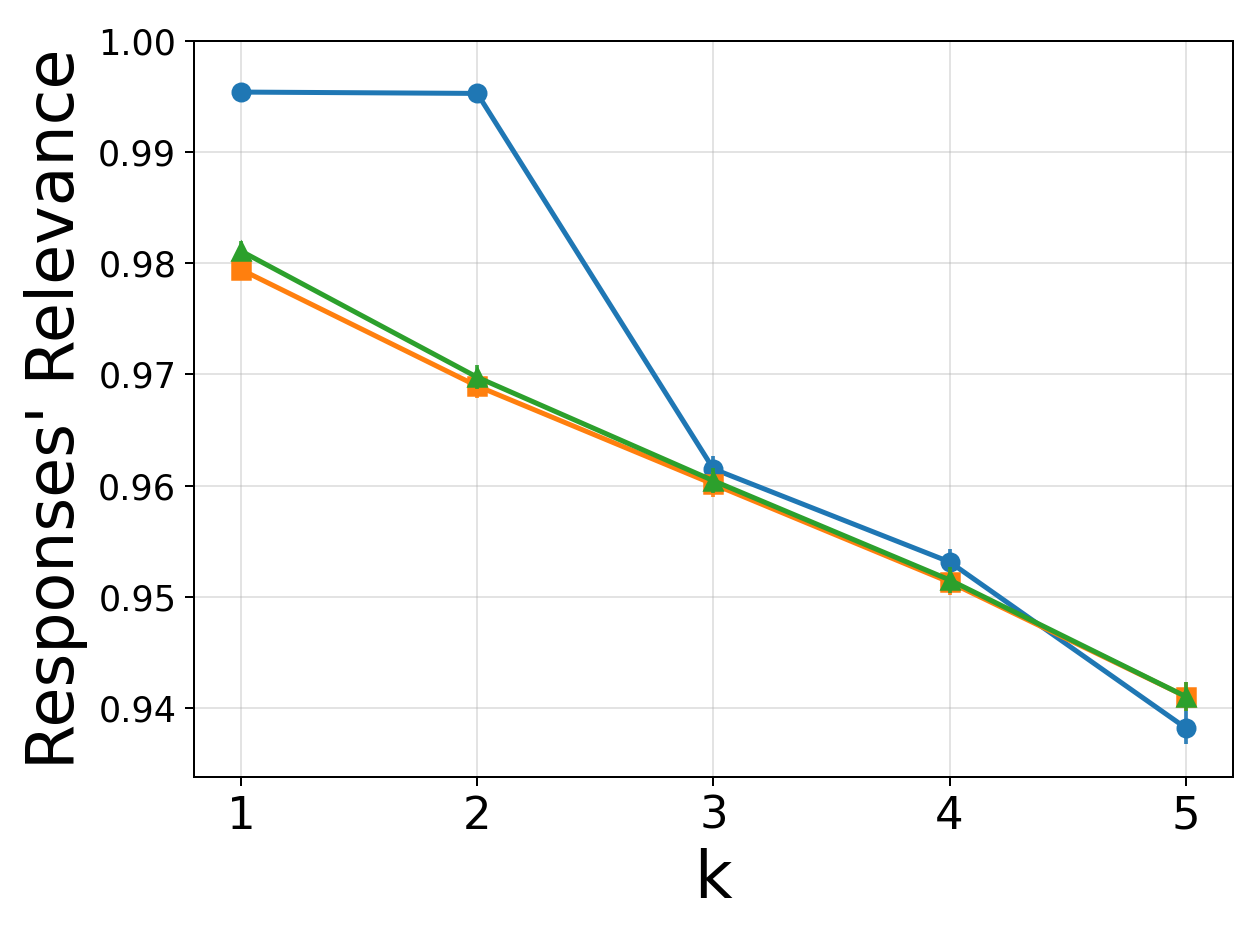}
\includegraphics[width=0.325\linewidth]{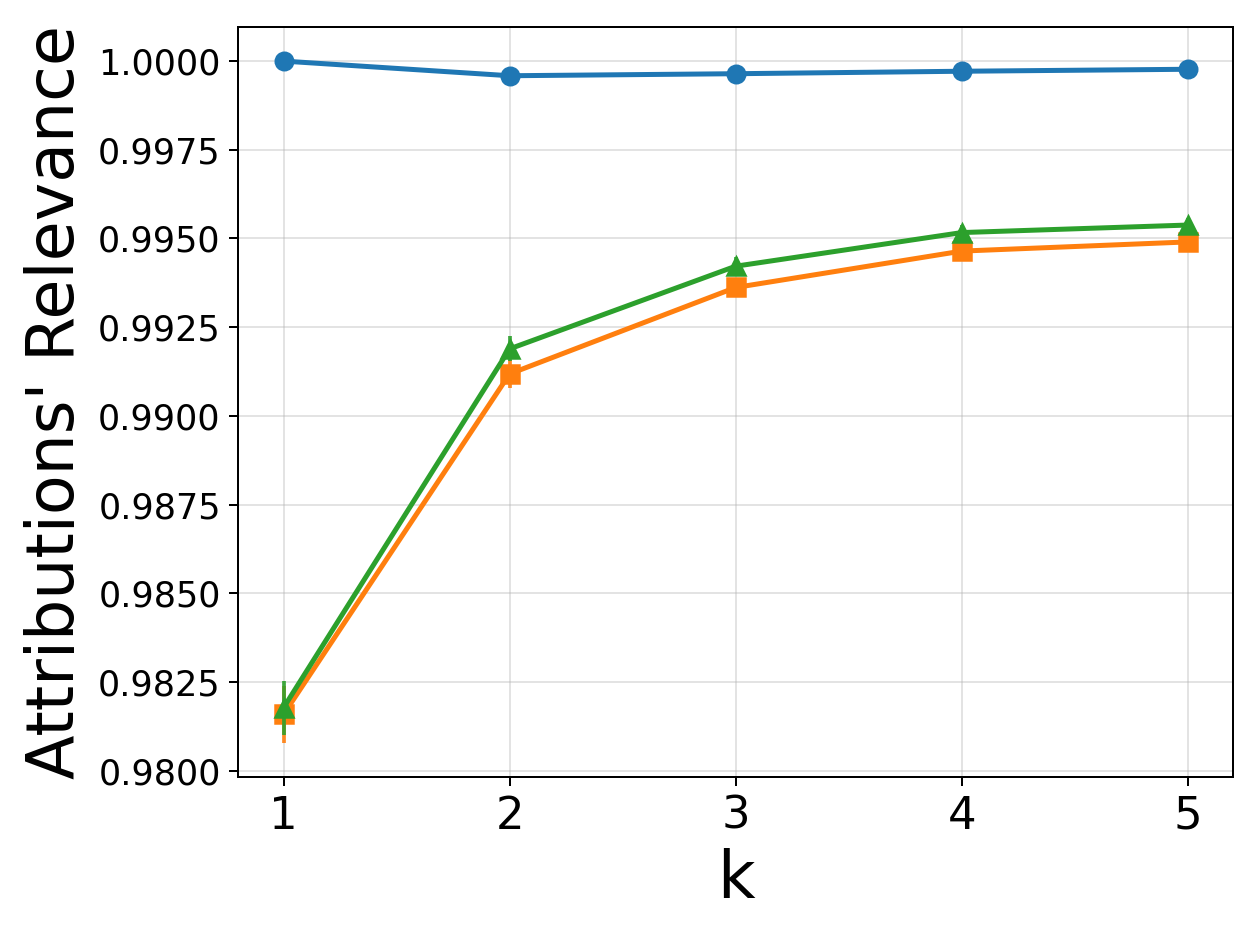}

\caption{$x^* \sim \mathcal{U}(\mathcal{V})$, $\quad x^0_i \sim \mathcal{N}(\boldsymbol{0.5}_\text{\di{}}, 0.1 \cdot \mathcal{I})$} \label{fig:uniform_normal}
\end{subfigure}
\hfill

\begin{subfigure}[t]{\columnwidth}
\centering

\includegraphics[width=0.325\linewidth]{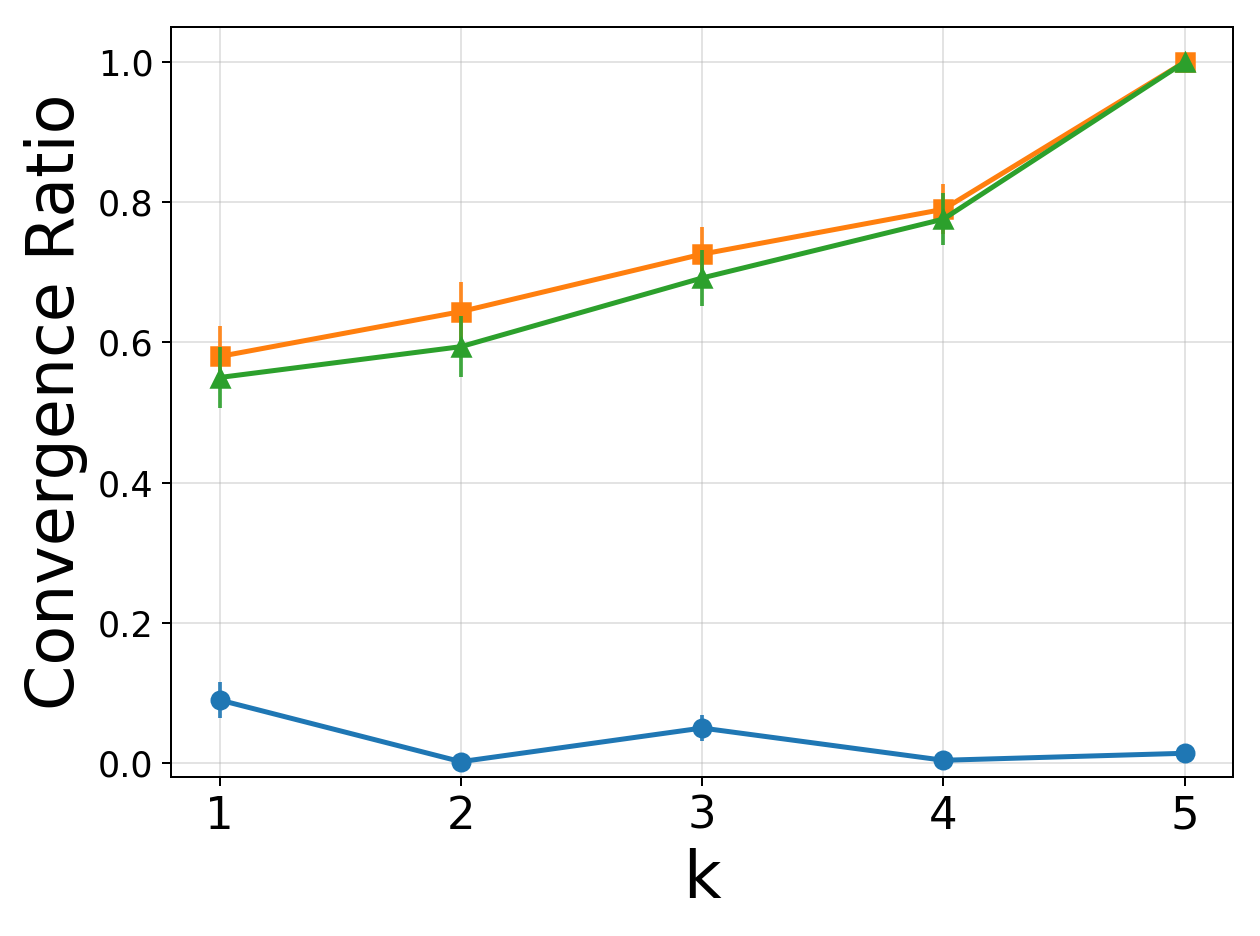}
\hspace{0.06\columnwidth}
\includegraphics[width=0.325\linewidth]{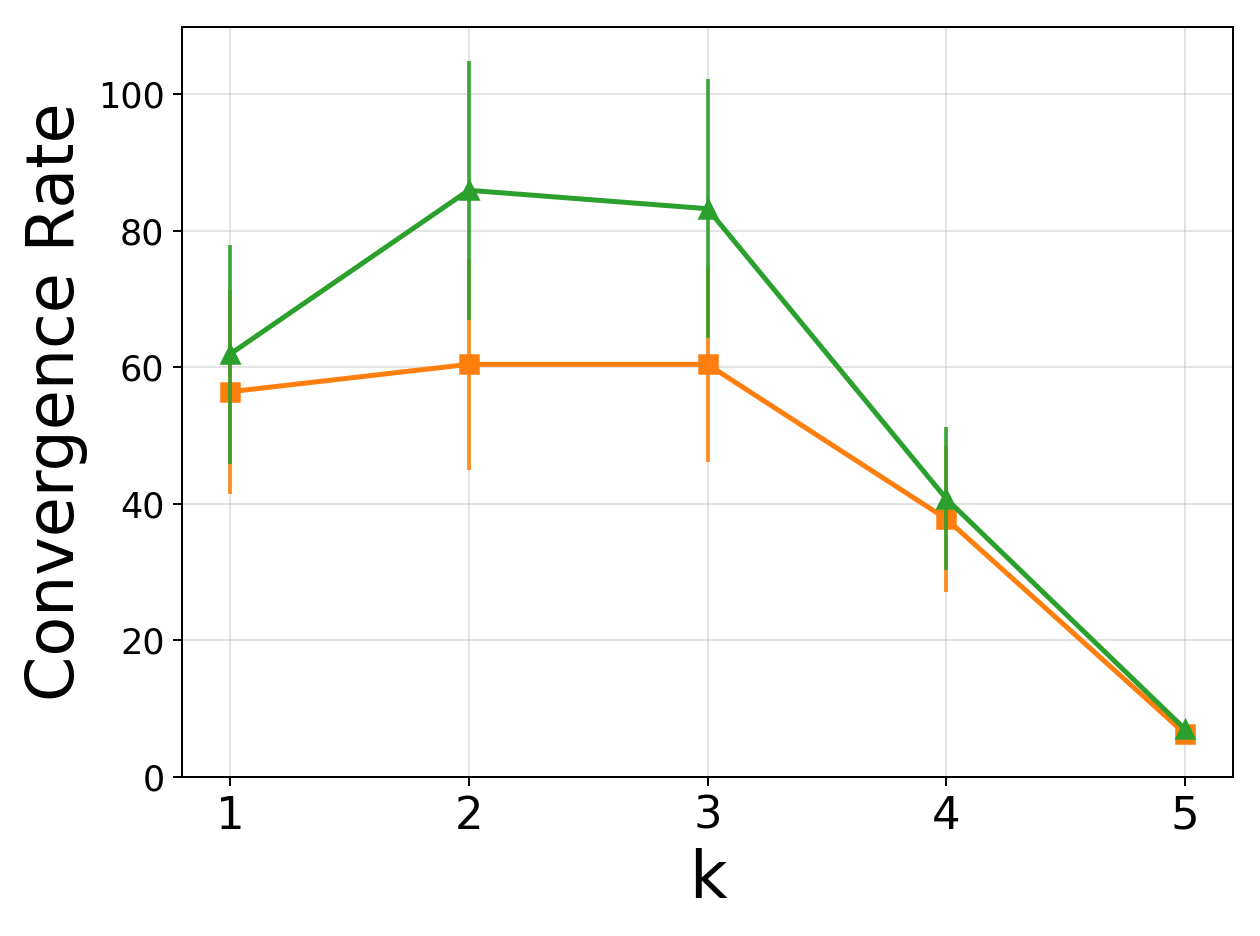}

\includegraphics[width=0.325\linewidth]{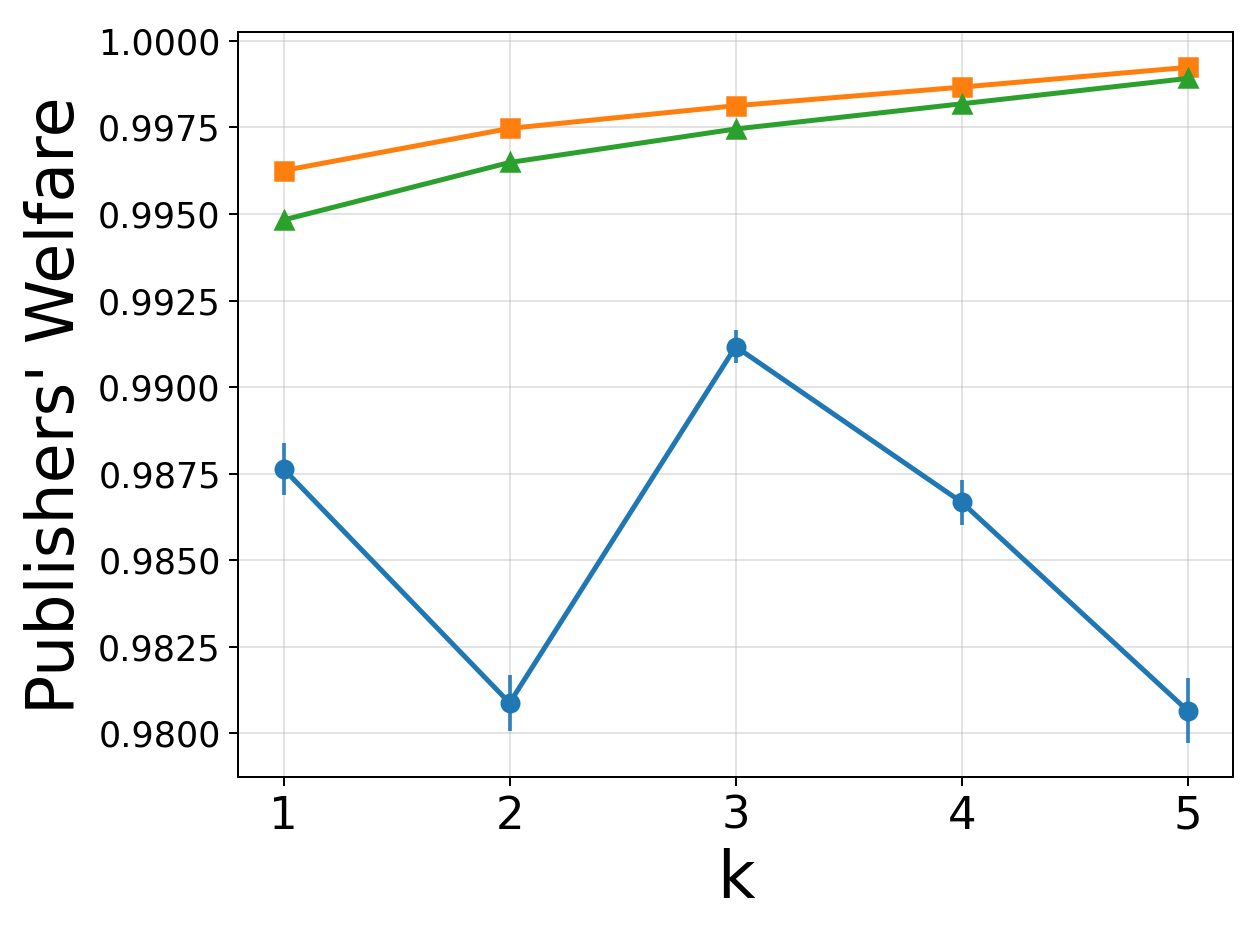}
\hfill
\includegraphics[width=0.325\linewidth]{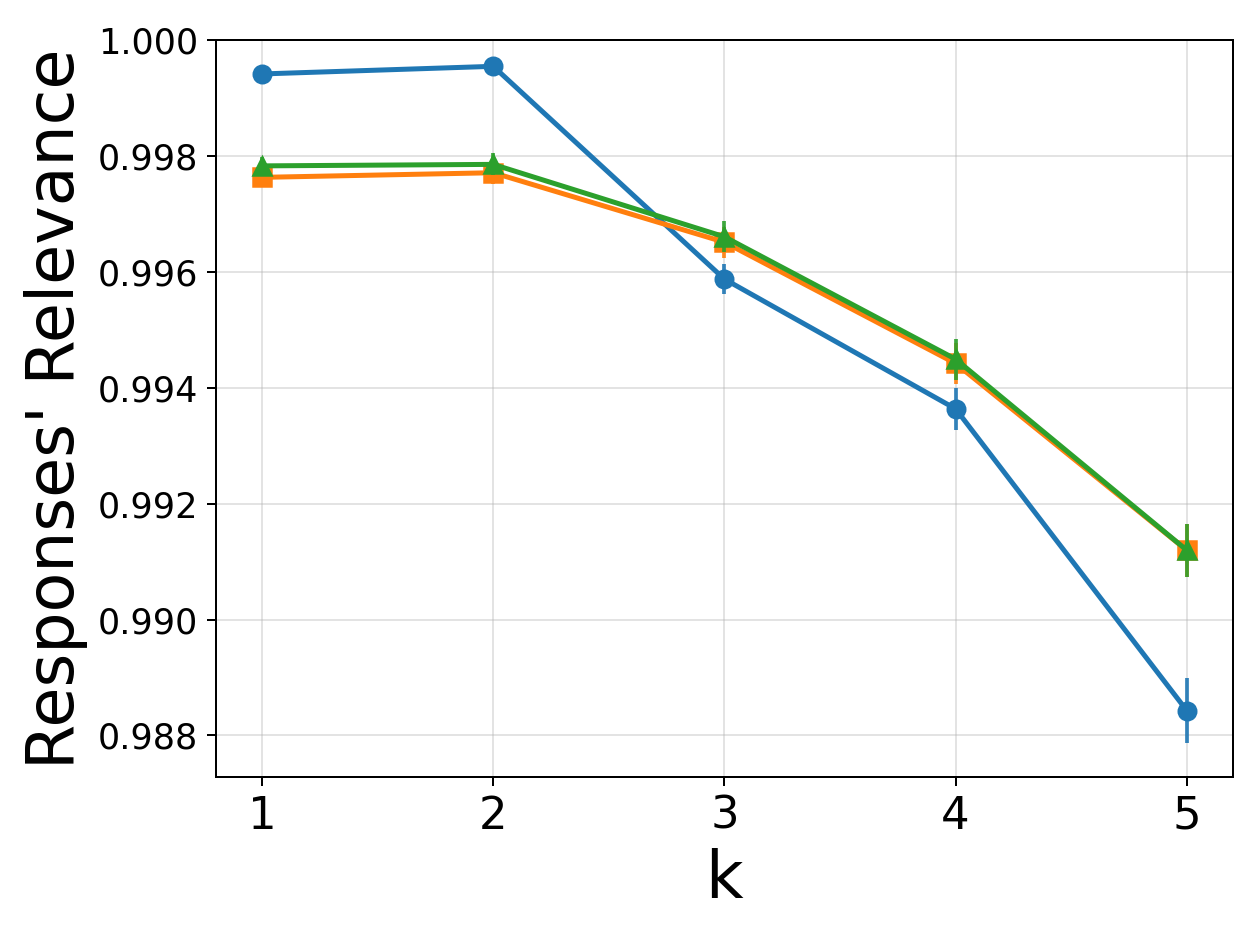}
\includegraphics[width=0.325\linewidth]{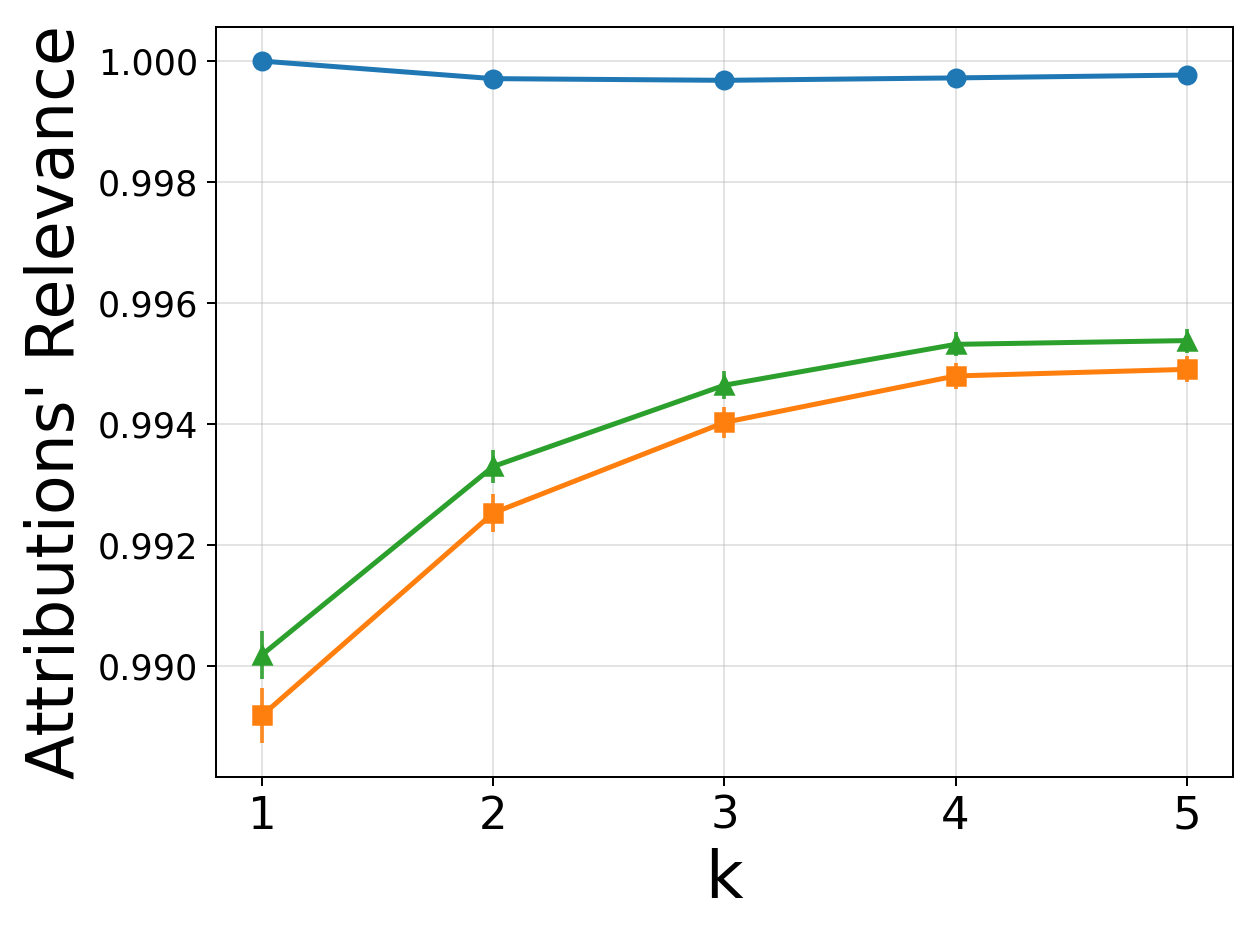}

\caption{$x^* \sim \mathcal{N}(\boldsymbol{0.5}_\text{\di{}}, 0.1 \cdot \mathcal{I})$, $\quad x^0_i \sim \mathcal{N}(\boldsymbol{0.5}_\text{\di{}}, 0.1 \cdot \mathcal{I})$}
\end{subfigure}

\caption{Stability and welfare analysis of the \PRP{}, softmax, and linear attribution functions for uniform and normal distribution combinations of the user's question $x^*$ and initial documents $x^0$, shown in blue circles, orange squares, and green triangles, respectively.
For each distribution configuration, the top charts report stability measures: convergence ratio and convergence rate. The bottom charts report welfare measures: publishers' welfare (PW), responses' relevance (RR), and attributions' relevance (AR).}
\label{fig:info_need_dist}
\end{figure}

\paragraph{Question and Initial Documents Distributions.}
We proceed to study the impact of the distribution of the user's question and the initial documents on the stability and welfare in the generative ecosystem.
Recall that in Section~\ref{sec:empirical_results} we draw both uniformly at random from $\mathcal{V}$. 
Here, we consider questions and initial documents that are drawn from combinations of a uniform distribution and a multivariate normal distribution.
For the multivariate normal distribution, we consider an isotropic Gaussian with a mean $\mu$ of $0.5$ and a covariance matrix $\Sigma$ of $0.1 \cdot \mathcal{I}_{\text{\di{}}}$.  

In Section~\ref{sec:empirical_results}, Figures~\ref{fig:stability_analysis} and \ref{fig:welfare_analysis} present the stability and welfare results using a uniform distribution over $\mathcal{V}$.
Figure~\ref{fig:info_need_dist} presents the results with the other combinations of uniform and normal distributions as stated above. 
The key observation is that the pattern of results as \K{} changes remains the same for all configurations.
For example, consider the convergence ratio, for which the \PRP{} function is constantly around $0$, while the softmax and linear functions increase with \K, meeting in a ratio of $1$ when \K{} $=n$.
Additionally, although we see different value ranges for the convergence rate, the patterns as \K{} varies, and the relative numbers between the softmax and linear functions remain the same. 

We can see in Figure~\ref{fig:stability_analysis} and Figure~\ref{fig:info_need_dist} that while the relative values of the publishers' welfare and the attributions' relevance induced by the attribution functions are maintained, some variation arises in terms of the responses' relevance. Specifically, in Figure~\ref{fig:uniform_normal} the \PRP{} function statistically significantly outperforms both the softmax and linear functions for \K{} $<n$, whereas in the other configurations this holds only for \K{} $\leq 2$.
However, we hasten to point out that the pattern of values of the attribution functions and the relative performance between them for most values of \K{} remain consistent.

\end{document}